\documentclass{lmcs} 
\pdfoutput=1

\usepackage{lastpage}
\lmcsdoi{15}{4}{8}
\lmcsheading{}{\pageref{LastPage}}{}{}%
{May~01,~2018}{Nov.~07,~2019}{}


\usepackage{hyperref}
\usepackage{cmll}
\usepackage[all]{xy}
\usepackage{amsmath, amssymb}
\usepackage{dirtytalk}
\usepackage[export]{adjustbox}
\usepackage{stmaryrd} 
\theoremstyle{plain} 

\begin{document}

\title[Lifting Coalg. Mod. and $\mathsf{MELL}$ Model Structure to E.-M. Categories]{Lifting Coalgebra Modalities and $\mathsf{MELL}$ Model Structure to Eilenberg-Moore Categories}

\author[J-S.\,P. Lemay]{Jean-Simon Pacaud Lemay}	
\address{University of Oxford, Computer Science Department, Oxford, UK}	
\email{jean-simon.lemay@kellogg.ox.ac.uk}  
\thanks{The author would like to thank Kellogg College, the Clarendon Fund, and the Oxford-Google DeepMind Graduate Scholarship for financial support.}	





\begin{abstract}   \noindent A categorical model of the multiplicative and exponential fragments of intuitionistic linear logic ($\mathsf{MELL}$), known as a \emph{linear category}, is a symmetric monoidal closed category with a monoidal coalgebra modality (also known as a linear exponential comonad). Inspired by Blute and Scott's work on categories of modules of Hopf algebras as models of linear logic, we study categories of algebras of monads (also known as Eilenberg-Moore categories) as models of $\mathsf{MELL}$. We define a $\mathsf{MELL}$ lifting monad on a linear category as a Hopf monad -- in the Brugui{\`e}res, Lack, and Virelizier sense --  with a special kind of mixed distributive law over the monoidal coalgebra modality. As our main result, we show that the linear category structure lifts to the category of algebras of $\mathsf{MELL}$ lifting monads. We explain how groups in the category of coalgebras of the monoidal coalgebra modality induce $\mathsf{MELL}$ lifting monads and provide a source for such groups from enrichment over abelian groups. Along the way we also define mixed distributive laws of symmetric comonoidal monads over symmetric monoidal comonads and lifting differential category structure. 
\end{abstract}

\maketitle

\subsection*{Acknowledgements:} The author would first like to thank the reviewers of LMCS for their very useful editorial comments and suggestions for this paper, as well as the reviewers of FSCD for their comments on the extended abstract version of this paper \cite{lemay2018lifting}. The author would also like to thank Jamie Vicary for useful discussions, editorial comments, and suggestions.

\section{Introduction}\label{intro}

Linear logic, as introduced by Girard \cite{girard1987linear}, is a resource-sensitive logic which due to its flexibility admits multiple different fragments and a wide range of applications. A categorical model of the multiplicative fragment of intuitionistic linear logic ($\mathsf{MILL}$) \cite{bierman1995categorical, girard1987linear2} is a symmetric monoidal closed category. Categories of modules of Hopf algebras \cite{majid2000foundations, sweedler1969hopf} (over a commutative ring) are important and of interest especially in representation theory \cite{wischnewsky1975linear} due in part as they are monoidal closed categories \cite{Bruguieres2011hopf, majid2000foundations}. Blute \cite{blute1996hopf} and Scott \cite{blute2004category} studied the idea of interpreting categories of modules of Hopf algebras as models of $\mathsf{MILL}$ with negation and its non-commutative variant. If one looks at the more general setting, categories of modules of cocommutative Hopf monoids in arbitrary symmetric monoidal closed categories are again symmetric monoidal closed categories and therefore models of $\mathsf{MILL}$. But for what kind of monoids in categorical models of the multiplicative and exponential fragments of intuitionistic linear logic ($\mathsf{MELL}$), is their categories of modules again a categorical model of $\mathsf{MELL}$?
 
The exponential fragment of $\mathsf{MELL}$ adds in the exponential modality which is a unary connective $\oc$, read as either ``of course'' or ``bang'', admitting four structural rules \cite{bierman1995categorical, mellies2009categorical}: promotion, dereliction, contraction, and weakening. In terms of the categorical semantics, the exponential modality $\oc$ is interpreted as a monoidal coalgebra modality \cite{blute2015cartesian} (see Definition \ref{moncoalgmod}), also known as a linear exponential comonad \cite{schalk2004categorical}, which in particular is a symmetric monoidal comonad (see Definition \ref{SMComonadef}), capturing the promotion and dereliction rules, such that for each object $A$, the exponential $\oc(A)$ comes equipped with a natural cocommutative comonoid structure (see Definition \ref{coalgmod}), capturing the contraction and weakening rules. Categorical models of $\mathsf{MELL}$ are known as linear categories \cite{bierman1995categorical, mellies2003categorical, mellies2009categorical}, which are symmetric monoidal closed categories with monoidal coalgebra modalities. We can now restate the question we seek to answer in this paper: 
 \begin{itemize}
\item  \textbf{Question 1:} \say{For what kind of monoid $A$ in a linear category, is the category of modules of $A$ also a linear category?}
\end{itemize}
Part of the answer regarding the symmetric monoidal closed structure of a linear category is already known: the monoid $A$ needs to be a cocommutative Hopf monoid (see Definition \ref{hopf}). What remains to be answered is how to `extend', or better yet `lift', the monoidal coalgebra modality to the category of modules of $A$. By `lifting' we mean that the forgetful functor from the Eilenberg-Moore category to base category preserves the required structure strictly. Observing that if $A$ is a monoid, then the endofunctor $A \otimes -$ is a monad (see Example \ref{monoidmonad}) and so the category of modules of $A$ corresponds precisely to the Eilenberg-Moore category of the monad $A \otimes -$. Therefore, we can further generalize the question we want to be answered: 
 \begin{itemize}
\item \textbf{Question 2:} \say{For what kind of monad on a linear category, is the Eilenberg-Moore category of algebras of that monad also a linear category?}
\end{itemize}
This now becomes a question of how to lift a comonad to the Eilenberg-Moore category of a monad and this brings us into the realm of distributive laws \cite{beck1969distributive, wisbauer2008algebras}.  Specifically, one requires a natural transformation known as a mixed distributive law \cite{harmer2007categorical,wisbauer2008algebras} (see Definition  \ref{mixeddef}) to obtain a lifting of a comonad to the Eilenberg-Moore category of a monad. In order to lift the symmetric monoidal closed structure of a linear category, the monad in question is required to be a symmetric Hopf monad (see Definition \ref{Hopfmonaddef}) in the sense of Brugui{\`e}res, Lack, and Virelizier \cite{Bruguieres2011hopf,Bruguieres2007hopf}, which are the monad analogue of cocommutative Hopf monoids (see Example \ref{hopfmonadex}). Finally to lift the monoidal coalgebra modality of a linear category the Eilenberg-Moore category of a symmetric Hopf monad one requires a symmetric monoidal mixed distributive law (Definition \ref{mixedmondef}), which is a special kind of mixed distributive law which we introduce in this paper. 

\subsection{Main Definitions and Results:} The main new technical definition of this paper is that of a symmetric monoidal mixed distributive law (Definition \ref{mixedmondef}) of a symmetric comonoidal monad over a symmetric monoidal comonad. Symmetric monoidal mixed distributive laws are mixed distributive laws between the underlying monads and comonads but with extra coherences regarding the symmetric (co)monoidal endofunctor structure. Proposition \ref{liftsymmix} shows that symmetric monoidal mixed distributive laws are in bijective correspondence to liftings of the symmetric monoidal comonad to the Eilenberg-Moore of the symmetric comonoidal monad and vice-versa, which partially answers Question 2. Proposition \ref{symmonresult} studies when a symmetric comonoidal monad of the form $A \otimes -$ admits a symmetric mixed distributive law over a symmetric monoidal comonad, which is related to answering Question 1, and proves that $A$ must be a cocommutative bimonoid in the Eilenberg-Moore category of the comonad. Two new definitions regarding the answer of Question 2 are that of exponential lifting monads (Definition \ref{expliftdef}) and $\mathsf{MELL}$ lifting monads (Definition \ref{mellliftdef}). An exponential lifting monad is a symmetric comonoidal monad with a symmetric monoidal mixed distributive law over a monoidal coalgebra modality, while a $\mathsf{MELL}$ lifting monad is an exponential lifting monad on a linear category which is also a symmetric Hopf monad. Proposition \ref{liftmoncoalg} provides a partial answer to Question 2, while Theorem \ref{mellthm} provides the full answer. We summarize these two main results as follows: 
\begin{itemize}
\item The Eilenberg-Moore category of an exponential lifting monad admits a monoidal coalgebra modality (Proposition \ref{liftmoncoalg}).
\item The Eilenberg-Moore category of a $\mathsf{MELL}$ lifting monad is a linear category (Theorem \ref{mellthm}).
\end{itemize}
Proposition \ref{bigresultprop} and Theorem \ref{bigresultthm} are dedicated to answering Question 1, which we summarize as (by abusing notation slightly): 
\begin{itemize}
\item Monoids $A$ in the Eilenberg-Moore category of a monoidal coalgebra modality induce exponential lifting monads $A \otimes -$, and therefore the category of modules over $A$ admits a monoidal coalgebra modality  (Proposition \ref{bigresultprop}).
\item Groups $H$ in the Eilenberg-Moore category of a monoidal coalgebra modality of a linear category induce $\mathsf{MELL}$ lifting monads $H \otimes -$ and therefore the category of modules over $H$ is a linear category (Theorem \ref{bigresultthm}).
\end{itemize}
We note that since the Eilenberg-Moore category of a monoidal coalgebra modality in a Cartesian monoidal category \cite{schalk2004categorical}, it follows that groups in the Eilenberg-Moore category are in fact also cocommutative Hopf monoids in the base category (see Definition \ref{cart} and Definition \ref{groupdef} for more details). We also address the special case of Lafont categories \cite{mellies2009categorical}, that is when the monoidal coalgebra modality is, in fact, a free exponential modality \cite{mellies2017explicit,slavnov2015banach}. In this case, Corollary \ref{Lafontcor} states that the category of modules over a cocommutative Hopf monoid is again a Lafont category. In the process of studying liftings of monoidal coalgebra modalities, we also discuss liftings of the strictly weaker notion of coalgebra modalities (Definition \ref{coalgmix}) in order to also discuss lifting differential category structure (Definition \ref{difflift}). In discussing differential categories, we also explain how additive structure and negatives provide a source of the desired monoids for answering Question 2 (Proposition \ref{!bimonoid} and Theorem \ref{!hopf}). 

\subsection{Organization of this paper} Section \ref{monoidsec} is dedicated to reviewing the familiar concepts of symmetric monoidal categories, (co)monoids, bimonoids, and Hopf monoids. Section \ref{monadsec} reviews the notions of (co)monads, symmetric comonoidal monads, symmetric monoidal comonads, and Hopf monads, where in particular we also discuss in detail the structures on their Eilenberg-Moore categories. In Section \ref{coalgsec} we review the categorical semantics of $\mathsf{MELL}$ by defining coalgebra modalities, monoidal coalgebra modalities, linear categories, and Lafont categories. Section \ref{SMCmixedsec} begins by reviewing the notion of a mixed distributive law and later also introduces and studies symmetric monoidal mixed distributive laws. Section \ref{coalgmixsec} is dedicated to answering Question 1 and Question 2 by introducing exponential lifting monads and $\mathsf{MELL}$ lifting monads and stating the main results of this paper. Finally, Section \ref{diffsec} addresses lifting differential category structure and how additive structure induces a source of exponential lifting monads and $\mathsf{MELL}$ lifting monads. 

\section{Monoids, Comonoids, Bimonoids, and Hopf Monoids}\label{monoidsec}

In this section, we review the notions of monoids, comonoids, bimonoids, and Hopf monoids. We refer the reader to \cite{kelly1964maclane, mellies2009categorical} for a more detailed introduction to monoids and comonoids, and to \cite{blute2004category, blute1996hopf} for details on bimonoids and Hopf monoids. We conclude this section with examples of Hopf monoids (and therefore also examples of monoids, comonoids, and bimonoids) in various examples of well known symmetric monoidal closed categories (Example \ref{HEX}). 

\subsection{Symmetric Monoidal (Closed) Categories}

We begin by first recalling the definitions of symmetric monoidal categories, Cartesian monoidal categories, and symmetric monoidal closed categories. The axiomatization of these concepts can be given in various equivalent ways. In this paper, we have chosen to express them as found in \cite{mellies2009categorical}. For a more in-depth introduction to symmetric monoidal categories, we refer to reader to the following introductory sources \cite{barr1990category,mac2013categories,mellies2009categorical}

\begin{defi}\label{SMC} A \textbf{symmetric monoidal category} \cite{mellies2009categorical} is a septuple 
\[(\mathbb{X}, \otimes, K, \alpha, \ell, \rho, \sigma)\] consisting of a category $\mathbb{X}$, a functor $\otimes: \mathbb{X} \times \mathbb{X} \to \mathbb{X}$ called the \textbf{monoidal product}, an object $K$ of $\mathbb{X}$ called the \textbf{monoidal unit}, and four natural isomorphisms:
\begin{align*}
\alpha_{A,B,C}&: A \otimes (B \otimes C) \xrightarrow{\cong} (A \otimes B) \otimes C & \ell_{A}&: K \otimes A \xrightarrow{\cong}  A\\
 \sigma_{A, B}&: A \otimes B \xrightarrow{\cong}  B \otimes A &  \rho_{A}&: A \otimes K \xrightarrow{\cong}  A
\end{align*}
such that the following diagrams commute: 

\begin{equation}\label{SMCaxioms}\begin{gathered} 
 \xymatrixcolsep{5pc}\xymatrix{
   A \otimes (B \otimes (C \otimes D))  \ar[r]^-{\alpha_{A, B, C\otimes D}} \ar[d]_{1_A \otimes \alpha_{B,C,D}} &  (A \otimes B) \otimes (C \otimes D)  \ar[dd]^-{\alpha_{A \otimes B, C, D}}\\
      A \otimes ((B \otimes C) \otimes D)  \ar[d]_-{\alpha_{A, B \otimes C, D}}    \\ 
          (A \otimes (B \otimes C)) \otimes D \ar[r]_-{\alpha_{A, B, C} \otimes 1_D} &  ((A \otimes B) \otimes C) \otimes D      
   } \\
    \xymatrixcolsep{5pc}\xymatrix{
          A \otimes (B \otimes C) \ar[r]^-{\alpha_{A,B,C}} \ar[d]_-{1_A \otimes \sigma_{B,C}} & (A \otimes B) \otimes C \ar[d]^-{\sigma_{A \otimes B,C}} \\ 
          A \otimes (C \otimes B) \ar[d]_-{\alpha_{A,C,B}}   & C \otimes (A \otimes B) \ar[d]^-{\alpha_{C,A,B}} \\
           (A \otimes C) \otimes B \ar[r]_-{\sigma_{A,C} \otimes 1_B} & (C \otimes A) \otimes B
   } \\
          \xymatrixcolsep{5pc}\xymatrix{
            A \otimes (K \otimes B) \ar[r]^-{\alpha_{A, K,B}} \ar[dr]_-{1_A \otimes \ell_B} &    (A \otimes K) \otimes B    \ar[d]^-{\rho_A \otimes 1_B}   \\
           & A \otimes B
   }\\
       \xymatrixcolsep{5pc}\xymatrix{
            A \otimes B  \ar[r]^{\sigma_{A,B}}  \ar@{=}[dr]   &  B \otimes A \ar[d]^-{\sigma_{B,A}} \\
            & A \otimes B
   }  \end{gathered}\end{equation}      
 \end{defi}
 
 There are several notions of ``morphisms'' of symmetric monoidal categories, that is, functors between symmetric monoidal categories which are compatible with the symmetric monoidal structures in a certain sense. In this paper we discuss  two such as examples: symmetric comonoidal functors (Definition \ref{SCmonaddef}.i) and symmetric monoidal functors (Definition \ref{SMComonadef}.i). From now on (to avoid overloading definitions with notation), symmetric monoidal categories will simply be denoted as triples $(\mathbb{X}, \otimes, K)$ when there is no confusion about the natural isomorphisms $\alpha$, $\ell$, $\rho$, and $\sigma$. This should hopefully not cause any confusion as this paper only works with symmetric monoidal categories and none of the weaker alternatives (such as monoidal categories or braided monoidal categories \cite{mellies2009categorical}). 
 
A particularly useful natural isomorphism in a symmetric monoidal category is the \textbf{interchange map} $\tau_{A,B,C,D}: (A \otimes B) \otimes (C \otimes D) \xrightarrow{\cong} (A \otimes C) \otimes (B \otimes D)$ which is defined as the following composite: 

\begin{equation}\label{interchange}\begin{gathered} \xymatrixcolsep{5pc} \xymatrixrowsep{1pc}\xymatrix{ (A \otimes B) \otimes (C \otimes D) \ar[dddddd]_-{\tau_{A,B,C,D}} \ar[r]^-{\alpha_{A \otimes B, C, D}} & ((A \otimes B) \otimes C) \otimes D \ar[dd]^-{\alpha^{-1}_{A,B,C} \otimes 1_D} \\
&\\
&  (A \otimes (B \otimes C)) \otimes D \ar[dd]^-{(1_A \otimes \sigma_{B,C}) \otimes 1_D}  \\
& \\ 
& (A \otimes (C \otimes B)) \otimes D \ar[dd]^-{\alpha_{A,C,B} \otimes 1_D}\\
& \\
(A \otimes C) \otimes (B \otimes D)& ((A \otimes C) \otimes B) \otimes D \ar[l]^-{\alpha^{-1}_{A \otimes C, B, D}}  
  } \end{gathered}\end{equation}
The interchange map appears in the coherences of a bimonoid (Definition \ref{bimonoiddef}) and of a monoidal coalgebra modality (Definition \ref{moncoalgmod}). 

Special examples of symmetric monoidal categories are those whose monoidal structure is, in fact, a finite product structure. 

\begin{defi}\label{cart} A \textbf{Cartesian monoidal category} \cite{mellies2009categorical} is a symmetric monoidal category $(\mathbb{X}, \otimes, K)$ such that $K$ is a \textbf{terminal object} \cite{mac2013categories}, that is, for each object $A$ there exists a unique map ${\mathsf{t}_A: A \to K}$; and $A \otimes B$ is a \textbf{product} \cite{mac2013categories} of $A$ and $B$ with projection maps:  
  \[  \xymatrixcolsep{3pc}\xymatrix{A \otimes B \ar[r]^-{1_A \otimes \mathsf{t}_B}  & A \otimes K \ar[r]^-{\rho_A} & A
  } \quad \quad \quad \xymatrixcolsep{3pc}\xymatrix{A \otimes B \ar[r]^-{\mathsf{t}_B \otimes 1_B}  & K \otimes A \ar[r]^-{\ell_B} & B  } \]
that is, for every pair of maps ${f: C \to A}$ and $g: C \to B$, there exists a \emph{unique} map ${\langle f, g \rangle: C \to A \otimes B}$ such that the following diagram commutes: 
\begin{equation}\label{}\begin{gathered} \xymatrixcolsep{5pc}\xymatrix{ & & C \ar@{-->}[d]^-{\exists\oc \langle f, g \rangle} \ar@/^0.5pc/[drr]^-{g} \ar@/_0.5pc/[dll]_-{f} \\
A & \ar[l]^-{\rho_A} A \otimes K & \ar[l]^-{1_A \otimes \mathsf{t}_B} \ar[r]_-{\mathsf{t}_B \otimes 1_B} A \otimes B & K \otimes A \ar[r]_-{\ell_B} & B 
  } \end{gathered}\end{equation}
 The map $\langle f, g \rangle$ is called the pairing of $f$ and $g$.  
\end{defi}

Any category with chosen products and terminal object is a Cartesian symmetric monoidal category where the natural isomorphisms $\alpha$, $\ell$, $\rho$, and $\sigma$ are all induced by the universal property of the product \cite{selinger2010survey}. It is also worth mentioning that a category with chosen coproducts is also a symmetric monoidal category and therefore its opposite category is a Cartesian monoidal category. 

\begin{defi} A \textbf{symmetric monoidal closed category} \cite{mellies2009categorical} is a symmetric mon- oidal category $(\mathbb{X}, \otimes, K)$ such that for each pair of objects $A$ and $B$, there is an object $A \multimap B$, called the \textbf{internal hom} of $A$ and $B$, and a map $\mathsf{ev}_{A,B}: (A \multimap B) \otimes A \to B$, called the \textbf{evaluation map}, such that for every map $f: C \otimes A \to B$, there exists a \emph{unique} map $\overline{f}: C \to A \multimap B$ such that the following diagram commutes: 
\begin{equation}\label{closedeq}\begin{gathered} \xymatrixcolsep{5pc}\xymatrix{ C \otimes A \ar[d]_-{\overline{f}} \ar[r]^-{f} & B \\
(A \multimap B) \otimes B \ar[ur]_-{~~~\mathsf{ev}_{A,B}}
  } \end{gathered}\end{equation}
A \textbf{Cartesian closed category} \cite{barr1990category} is a Cartesian monoidal category which is also a symmetric monoidal closed category.   
\end{defi}

 \begin{exas}\label{SMCex} \normalfont Here are some well-known examples of symmetric monoidal closed categories: 
\begin{enumerate}
\item Let $\mathsf{SET}$ be the category of sets and functions between them. Then $(\mathsf{SET}, \times, \lbrace \ast \rbrace)$ is a Cartesian monoidal category where the product is given by the Cartesian product of sets $\times$, the terminal object is a chosen singleton $\lbrace \ast \rbrace$. $(\mathsf{SET}, \times, \lbrace \ast \rbrace)$ is also a Cartesian closed category where in particular $X \multimap Y$ is the set of all functions from $X$ to $Y$.  
\item Let $\mathsf{REL}$ be the category of sets and relations, where recall that a relation from a set $X$ to a set $Y$, denoted $R: X \to Y$, is a subset $R \subseteq X \times Y$. Then $(\mathsf{REL}, \times, \lbrace \ast \rbrace)$ is a symmetric monoidal category where the monoidal product and monoidal unit are again given by the Cartesian product of sets and a chosen singleton respectively. It is important to note that the Cartesian product of sets is not a product (in the categorical sense) for $\mathsf{REL}$, and therefore $(\mathsf{REL}, \times, \lbrace \ast \rbrace)$ is not a Cartesian monoidal category. $(\mathsf{REL}, \times, \lbrace \ast \rbrace)$ is also a symmetric monoidal closed category with $X \multimap Y := X \times Y$. 
\item Let $\mathbb{K}$ be a field and $\mathsf{VEC}_\mathbb{K}$ the category of $\mathbb{K}$-vector spaces and $\mathbb{K}$-linear maps between them. Then $(\mathsf{VEC}_\mathbb{K}, \otimes, \mathbb{K})$ is a symmetric monoidal closed category where the monoidal product is given by the algebraic tensor product of vector spaces $\otimes$, and the monoidal unit is the field $\mathbb{K}$. $(\mathsf{VEC}_\mathbb{K}, \otimes, \mathbb{K})$ is also symmetric monoidal closed category where $V \multimap W$ is the $\mathbb{K}$-vector space of all $\mathbb{K}$-linear maps from $V$ to $W$. 
\end{enumerate}
\end{exas}

\subsection{Monoids and Comonoids}

Here we review the notion of monoids and comonoids. While monoids and comonoids are dual notions (in the categorical sense), we take the pain of defining both concepts explicitly. 

\begin{defi} Let $(\mathbb{X}, \otimes, K)$ be a symmetric monoidal category. A \textbf{monoid} \cite{mellies2009categorical} in $(\mathbb{X}, \otimes, K)$ is a triple $(A, \nabla, \mathsf{u})$ consisting of an object $A$, a map ${\nabla: A \otimes A\to A}$ called the \textbf{multiplication}, and a map $\mathsf{u}: K \to A$ called the \textbf{unit}, and such that the following diagrams commute: 

      \begin{equation}\label{monoideq}\begin{gathered} \xymatrixcolsep{5pc}\xymatrix{
        A \ar[d]_-{\rho^{-1}_A}  \ar@{=}[ddrr]^-{} \ar[r]^-{\ell^{-1}_A} & K \otimes A   \ar[r]^-{\mathsf{u} \otimes 1_A} & A \otimes A \ar[dd]^-{\nabla}  \\ 
       A \otimes K  \ar[d]_-{1_A \otimes \mathsf{u}} \\       
        A \otimes A \ar[rr]_-{\nabla} & & A }\\
         \xymatrixcolsep{5pc}\xymatrix{
        A \otimes (A \otimes A) \ar[r]^-{\alpha_{A, A, A}} \ar[dd]_-{1_A \otimes \nabla} & (A \otimes A) \otimes A \ar[r]^-{\nabla \otimes 1_A} & A \otimes A \ar[dd]^-{\nabla}  \\ \\
         A \otimes A \ar[rr]_-{\nabla} && A  }\end{gathered}\end{equation}
A monoid $(A, \nabla, \mathsf{u})$ is said to \textbf{commutative} if the following diagram commutes: 
 \begin{equation}\label{}\begin{gathered} \xymatrixcolsep{5pc}\xymatrix{ A \otimes A \ar[dr]_-{\nabla} \ar[r]^-{\sigma_{A,A}} & A \otimes A \ar[d]^-{\nabla} \\
& A
  } \end{gathered}\end{equation}  
A \textbf{monoid morphism} $f: (A, \nabla, \mathsf{u}) \to (B, \nabla^\prime, \mathsf{u}^\prime)$ between monoids $(A, \nabla, \mathsf{u})$ and $(B, \nabla^\prime, \mathsf{u}^\prime)$ is a map $f: A \to B$ which preserves the multiplication and the unit, that is, the following diagrams commute: 
 \begin{equation}\label{}\begin{gathered} \xymatrixcolsep{5pc}\xymatrix{ A \otimes A \ar[d]_-{\nabla} \ar[r]^-{f \otimes f} & B \otimes B \ar[d]^-{\nabla^\prime} & K \ar[dr]_-{\mathsf{u}^\prime} \ar[r]^-{\mathsf{u}} & A \ar[d]^-{f} \\
 A \ar[r]_-{f} & B & & B 
  } \end{gathered}\end{equation} 
\end{defi}

\begin{defi} Let $(\mathbb{X}, \otimes, K)$ be a symmetric monoidal category. A \textbf{comonoid} \cite{mellies2009categorical} in $(\mathbb{X}, \otimes, K)$ is a triple $(C, \Delta, \mathsf{e})$ consisting of an object $C$, a map $\Delta: C \to C \otimes C$ called the \textbf{comultiplication}, and a map $\mathsf{e}: C \to K$ called the \textbf{counit} such that the following diagrams commute:
\begin{equation}\label{comonoid}\begin{gathered} \xymatrixcolsep{5pc}\xymatrix{C \ar@{=}[drr]^-{} \ar[d]_-{\Delta} \ar[r]^-{\Delta} & C \otimes C \ar[r]^-{\mathsf{e} \otimes 1_C} & K \otimes C \ar[d]^-{\ell_C}   \\
C \otimes C \ar[r]_-{1_C \otimes \mathsf{e}} & C \otimes K \ar[r]_-{\rho_C} & C
 } \\
 \xymatrixcolsep{5pc}\xymatrix{C \ar[r]^-{\Delta} \ar[d]_-{\Delta} & C \otimes C \ar[r]^-{1_C \otimes \Delta} & C \otimes (C \otimes C) \ar[d]^-{\alpha_{C,C,C}} \\
C \otimes C \ar[rr]_-{\Delta\otimes 1_C} && (C \otimes C) \otimes C
 }\end{gathered}\end{equation}
A comonoid $(C, \Delta, \mathsf{e})$ is said to be \textbf{cocommutative} if the following diagram commutes: 
\begin{equation}\label{}\begin{gathered} \xymatrixcolsep{5pc}\xymatrix{ C \ar[dr]_-{\Delta} \ar[r]^-{\Delta} & C \otimes C \ar[d]^-{\sigma_{C,C}} \\
& C \otimes C
  } \end{gathered}\end{equation}       
A \textbf{comonoid morphism} $f: (C, \Delta, \mathsf{e})\to (D, \Delta^\prime, \mathsf{e}^\prime)$ between comonoids $(C, \Delta, \mathsf{e})$ and $(D, \Delta^\prime, \mathsf{e}^\prime)$ is a map $f: C\to D$ which preserves the comultiplication and counit, that is, the following diagrams commute: 
   \begin{equation}\begin{gathered} \xymatrixcolsep{5pc}\xymatrix{
        C   \ar[r]^-{\Delta} \ar[d]_-{f} & C \otimes C \ar[d]^-{f \otimes f} & C  \ar[r]^-{f} \ar[dr]_-{\mathsf{e}} & D \ar[d]^{\mathsf{e}^\prime}\\
        D  \ar[r]_-{\Delta^\prime} &  D \otimes D & & K}\end{gathered}\end{equation}
\end{defi}

Every object of a Cartesian monoidal category comes equipped with a unique cocommutative comonoid structure using the diagonal map and terminal map. In the categorical semantics of linear logic, this is of particular interests when considering the Eilenberg-Moore category of a monoidal coalgebra modality (Section \ref{moncoalgsec}).   

\begin{lemC}[\cite{porst2008categories}]\label{Cartcom} Let $(\mathbb{X}, \otimes, K)$ be a Cartesian monoidal category. Then every object $A$ admits a unique cocommutative comonoid structure $(A, \Delta_A, \mathsf{t}_A)$ where the counit $\mathsf{t}_A: A \to K$ is the unique map to the terminal object $K$, and $\Delta_A := \langle 1_A, 1_A \rangle$, that is, the unique map which makes the following diagram commute: 
\begin{equation}\label{}\begin{gathered} \xymatrixcolsep{5pc}\xymatrix{ & & A \ar@{-->}[d]^-{\exists\oc \Delta_A} \ar@{=}[drr]^-{} \ar@{=}[dll]_-{} \\
A & \ar[l]^-{\rho_A} A \otimes K & \ar[l]^-{1_A \otimes \mathsf{t}_B} \ar[r]_-{\mathsf{t}_B \otimes 1_B} A \otimes A & K \otimes A \ar[r]_-{\ell_A} & A 
  } \end{gathered}\end{equation}
  Furthermore, for every map $f: A \to B$, $f: (A, \Delta_A, \mathsf{t}_A) \to (B, \Delta_B, \mathsf{t}_B)$ is a comonoid morphism. 
\end{lemC}

\subsection{Bimonoids and Hopf Monoids}

Bimonoids are both monoids and comonoids such that monoid and comonoid structures are compatible.

\begin{defi}\label{bimonoiddef} In a symmetric monoidal category $(\mathbb{X}, \otimes, K)$, a \textbf{bimonoid} \cite{porst2008categories} is a quintuple $(A, \nabla, \mathsf{u}, \Delta, \mathsf{e})$ such that $(A, \nabla, \mathsf{u})$ is a monoid, $(A, \Delta, \mathsf{e})$ is a comonoid, and the following diagrams commute: 
\begin{equation}\label{bimonoideq}\begin{gathered}\xymatrixcolsep{3pc}\xymatrix{K \ar@{=}[dr] \ar[r]^-{\mathsf{u}} & A \ar[d]^-{\mathsf{e}} & A \otimes A \ar[dd]_-{\nabla} \ar[rr]^-{\Delta \otimes \Delta}&& (A \otimes A) \otimes (A \otimes A) \ar[d]^-{\tau_{A,A,A,A}}   \\ 
& K  & && (A \otimes A) \otimes (A \otimes A) \ar[d]^-{\nabla \otimes \nabla} \\
&&   A \ar[rr]_-{\Delta} && A \otimes A 
  } \\
  \xymatrixcolsep{5pc}\xymatrix{A \otimes A \ar[d]_{\mathsf{e} \otimes \mathsf{e}} \ar[r]^-{\nabla} & A  \ar[d]^-{\mathsf{e}}& K \ar[d]_-{\ell_K} \ar[r]^-{\mathsf{u}} & A \ar[d]^-{\Delta} \\ 
K \otimes K \ar[r]_-{\ell_K}&K& K \otimes K \ar[r]_-{\mathsf{u} \otimes \mathsf{u}} & A \otimes A
  }\end{gathered}\end{equation}
where $\tau$ is the interchange map as defined in (\ref{interchange}). A bimonoid is said to be (co)commutative if the underlying (co)monoid is (co)commutative.   
\end{defi}

For a Cartesian monoidal category, every monoid induces a cocommutative bimonoid. This is because the bimonoid identities are equivalent to requiring that the multiplication and unit be comonoid morphisms, and dually that the comultiplication and counit be monoid morphisms. 

\begin{lemC}[\cite{porst2008categories}]\label{Cartbim} Let $(\mathbb{X}, \otimes, K)$ be a Cartesian monoidal category and $(A, \nabla, \mathsf{u})$ a monoid in $(\mathbb{X}, \otimes, K)$. Then $(A, \nabla, \mathsf{u}, \Delta_A, \mathsf{t}_A)$ is a cocommutative bimonoid, where $(A,\Delta_A, \mathsf{t}_A)$ is the unique cocommutative comonoid structure on $A$ from Lemma \ref{Cartcom}. \end{lemC}

\begin{defi}In a symmetric monoidal category $(\mathbb{X}, \otimes, K)$, a \textbf{Hopf monoid} \cite{blute2004category, blute1996hopf} is a sextuple $(H, \nabla, \mathsf{u}, \Delta, \mathsf{e}, \mathsf{S})$ consisting of a bimonoid $(H, \nabla, \mathsf{u}, \Delta, \mathsf{e})$ and a map $\mathsf{S}: H \to H$ called the antipode such that the following diagram commutes: 
\begin{equation}\label{hopf}\begin{gathered}\xymatrix{& H \otimes H \ar[rr]^-{1 \otimes \mathsf{S}} & & H \otimes H \ar[dr]^-{\nabla}  \\
H \ar[ur]^-{\Delta} \ar[dr]_-{\Delta} \ar[rr]^-{\mathsf{e}} &   & K \ar[rr]^-{\mathsf{u}} & & H \\
&H \otimes H \ar[rr]_-{\mathsf{S} \otimes 1} & & H \otimes H \ar[ur]_-{\nabla}
  } \end{gathered}\end{equation}
A Hopf monoid is said to be (co)commutative if its underlying bimonoid is (co)commutative.   
\end{defi}

Here are a few well known properties about Hopf monoids: 

\begin{lemC}[\cite{blute1996hopf, majid2000foundations}]\label{lemmahopf}  Let $(H, \nabla, \mathsf{u}, \Delta, \mathsf{e}, \mathsf{S})$ be a Hopf monoid in a symmetric monoidal category. 
\begin{enumerate}[{\em (i)}]
\item The antipode is unique, that is, if $(H, \nabla, \mathsf{u}, \Delta, \mathsf{e}, \mathsf{S}^\prime)$ is also a Hopf monoid then $\mathsf{S} = \mathsf{S}^\prime$.
\item If $(H, \nabla, \mathsf{u}, \Delta, \mathsf{e}, \mathsf{S})$ is commutative (resp. cocommutative) then the antipode $\mathsf{S}: (H, \nabla, \mathsf{u}) \to (H, \nabla, \mathsf{u})$ is a monoid morphism (resp. $\mathsf{S}: (H, \Delta, \mathsf{e}) \to (H, \Delta, \mathsf{e})$ is a comonoid morphism).  
\end{enumerate}
\end{lemC} 

In Cartesian monoidal categories, Hopf monoids are better known as groups: 

\begin{defi}\label{groupdef} In a Cartesian monoidal category $(\mathbb{X}, \otimes, K)$, a \textbf{(abelian) group} \cite{hasegawa2010bialgebras} is a quadruple $(G, \nabla, \mathsf{u}, \mathsf{S})$ consisting of a (commutative) monoid $(G, \nabla, \mathsf{u})$ and a map $\mathsf{S}: G \to G$ such that $(G, \nabla, \mathsf{u}, \Delta_A, \mathsf{t}_A, \mathsf{S})$ is a (commutative and) cocommutative Hopf monoid, where $(A,\Delta_A, \mathsf{t}_A)$ is the unique cocommutative comonoid structure on $A$ from Lemma \ref{Cartcom}. 
\end{defi}

\begin{exas} \label{HEX} \normalfont As cocommutative Hopf monoids will be of particular interest for this paper, we conclude this section with some examples of Hopf monoids in the symmetric monoidal categories from Example \ref{SMCex}. 
 \begin{enumerate}
\item Since $(\mathsf{SET}, \times, \lbrace \ast \rbrace)$ is a Cartesian monoidal category, by Lemma \ref{Cartcom}, every set admits a unique cocommutative comonoid structure. Monoids in $(\mathsf{SET}, \times, \lbrace \ast \rbrace)$ coincide with the classical notion of monoids, and by Lemma \ref{Cartbim} are therefore cocommutative bimonoids. Similarly, groups in $(\mathsf{SET}, \times, \lbrace \ast \rbrace)$ correspond to the classical notion of groups. Explicitly if $G$ is a group in the classical sense with binary operation $\cdot$, unit element $u$, and inverse operation $(-)^{-1}$, then $(G, \nabla, \mathsf{u}, \Delta, \mathsf{e}, \mathsf{S})$ is a cocommutative Hopf monoid in $(\mathsf{SET}, \times, \lbrace \ast \rbrace)$ where: 
\[\xymatrixcolsep{5pc}\xymatrixrowsep{0.1pc}\xymatrix{ G \times G \ar[r]^-{\nabla} & G && \lbrace \ast \rbrace \ar[r]^-{\mathsf{u}} & G  \\
(g,h) \ar@{|->}[r]^-{} & g \cdot h && \ast \ar@{|->}[r]^-{} & u  
  } \]
  
  \[\xymatrixcolsep{5pc}\xymatrixrowsep{0.1pc}\xymatrix{ G  \ar[r]^-{\Delta} & G\times G && G \ar[r]^-{\mathsf{e}} & \lbrace \ast \rbrace  \\
g \ar@{|->}[r]^-{} & (g,g) && g \ar@{|->}[r]^-{} & \ast 
  } \]
  
  \[\xymatrixcolsep{5pc}\xymatrixrowsep{0.1pc}\xymatrix{G \ar[r]^-{\mathsf{S}} & G \\
g \ar@{|->}[r]^-{} & g^{-1} 
  } \]
\item Monoids in $(\mathsf{REL}, \times, \lbrace \ast \rbrace)$ are studied in \cite{jenvcova2017monoids}, while bimonoids and Hopf monoids in $(\mathsf{REL}, \times, \lbrace \ast \rbrace)$ are studied in \cite{hasegawa2010bialgebras}. In particular, every group in $(\mathsf{SET}, \times, \lbrace \ast \rbrace)$ induces a cocommutative Hopf monoid in $(\mathsf{REL}, \times, \lbrace \ast \rbrace)$. Explicitly, if $G$ is a group (in the classical sense with binary operation, unit, and inverse operation as in the previous example), then the following relations provide a cocommutative Hopf monoid structure on $G$ in $(\mathsf{REL}, \times, \lbrace \ast \rbrace)$: 
\[ \nabla := \lbrace \left((g,h), g\cdot h \vert~ g, h \in G \right) \rbrace \subseteq (G \times G) \times G \quad \quad \mathsf{u} := \lbrace (\ast, u) \rbrace \subseteq  \lbrace \ast \rbrace \times G \]

\[\Delta := \lbrace \left(g, (g,g) \right)  \rbrace \subseteq G \times (G \times G) \quad \quad \quad \mathsf{e} := \lbrace (g, \ast) \vert~ g \in G \rbrace \subseteq G \times \lbrace \ast \rbrace \]

\[\mathsf{S} := \lbrace (g, g^{-1}) \vert~ g \in G \rbrace \subseteq G \times G\]
Note that since $(\mathsf{REL}, \times, \lbrace \ast \rbrace)$ is not a Cartesian monoidal category, there are examples of comonoids in $(\mathsf{REL}, \times, \lbrace \ast \rbrace)$ whose comultiplication and counit are not given as the diagonal relation or terminal relation as above. See Example \ref{linex}.2 for such an example. 

\item Monoids, comonoids, bimonoids, and Hopf monoids in $(\mathsf{VEC}_\mathbb{K}, \otimes, \mathbb{K})$ are more commonly referred to as $\mathbb{K}$-algebras, $\mathbb{K}$-coalgebras, $\mathbb{K}$-bialgebras, and $\mathbb{K}$-Hopf algebras respectively. A simple example of a cocommutative $\mathbb{K}$-Hopf algebra are group $\mathbb{K}$-algebras \cite{majid2000foundations}. Let $G$ be a group (again in the classical sense) and let $\mathbb{K}[G]$ be the free $\mathbb{K}$-vector space over the set $G$, where we denote the basis elements of $\mathbb{K}[G]$ as $e_g$ for all $g \in G$. $\mathbb{K}[G]$ is known as the group $\mathbb{K}$-algebra and it is a cocommutative $\mathbb{K}$-Hopf algebra with the following $\mathbb{K}$-linear maps (which are defined on the basis elements):
 \[\xymatrixcolsep{5pc}\xymatrixrowsep{0.1pc}\xymatrix{ \mathbb{K}[G] \otimes \mathbb{K}[G] \ar[r]^-{\nabla} & \mathbb{K}[G] && \mathbb{K} \ar[r]^-{\mathsf{u}} &  \mathbb{K}[G]  \\
e_g \otimes e_h \ar@{|->}[r]^-{} & e_{g \cdot h} && 1 \ar@{|->}[r]^-{} & e_u  
  } \]
  
   \[\xymatrixcolsep{5pc}\xymatrixrowsep{0.1pc}\xymatrix{  \mathbb{K}[G]  \ar[r]^-{\Delta} & \mathbb{K}[G]\otimes \mathbb{K}[G] && \mathbb{K}[G] \ar[r]^-{\mathsf{e}} & \mathbb{K}  \\
e_g \ar@{|->}[r]^-{} & e_g \otimes e_g && e_g \ar@{|->}[r]^-{} & 1
  } \]
  
   \[ \xymatrixcolsep{5pc}\xymatrixrowsep{0.1pc}\xymatrix{\mathbb{K}[G] \ar[r]^-{\mathsf{S}} & \mathbb{K}[G] \\
e_g \ar@{|->}[r]^-{} & e_{g^{-1}}
  } \]
Other examples of $\mathbb{K}$-Hopf algebras include the polynomial rings $\mathbb{K}[x_1, \hdots, x_n]$, the tensor algebra $\mathsf{T}(V)$ and the symmetric algebra $\mathsf{S}(V)$ over a $\mathbb{K}$-vector spaces $V$. For more details on these examples and other examples of Hopf algebras, see \cite{blute2004category, blute1996hopf, majid2000foundations, sweedler1969hopf}. 
\end{enumerate}
\end{exas}

\section{Monads and Comonads on Symmetric Monoidal Categories}\label{monadsec}

In this section, we review the notion of monads and comonads and their respective Eilenberg-Moore categories. In this paper, we are particularly interested in monads and comonads on symmetric monoidal categories whose Eilenberg-Moore category is again a symmetric monoidal category such that the forgetful functor preserves the symmetric monoidal structure strictly. These special monads and comonads are known as symmetric comonoidal monads and symmetric monoidal comonads respectively. Again while these concepts are dual notions, we define both in detail to provide a complete story and to introduce notion. We also discuss Hopf monads, in the sense of Brugui{\`e}res, Lack, and Virelizier, in order to lift symmetric monoidal closed structure. We refer the reader to \cite{mac2013categories} for a detailed introduction of monads and comonads, and to \cite{Bruguieres2011hopf,Bruguieres2007hopf,moerdijk2002monads} for more details on comonoidal monads and Hopf monads, and to \cite{bierman1995categorical, hyland2003glueing, mellies2003categorical, mellies2009categorical} for more details on symmetric monoidal comonads.

\subsection{Monads and their Algebras}

\begin{defi} Let $\mathbb{X}$ be a category. 
\begin{enumerate}[{\em (i)}]
\item A \textbf{monad} \cite{mac2013categories} on a category $\mathbb{X}$ is a triple $(\mathsf{T}, \mu, \eta)$ consisting of a functor ${\mathsf{T}: \mathbb{X} \to \mathbb{X}}$ and two natural transformations $\mu_A: \mathsf{T}\mathsf{T}(A) \to \mathsf{T}(A)$ and $\eta_A: A \to \mathsf{T}(A)$ such that the following diagrams commute:
      \begin{equation}\label{monadeq}\begin{gathered} \xymatrixcolsep{5pc}\xymatrix{
        \mathsf{T}  A  \ar[r]^-{\mathsf{T}(\eta_A)} \ar[d]_-{\eta_{\mathsf{T}(A)}} \ar@{=}[dr]^-{}& \mathsf{T} \mathsf{T}(A) \ar[d]^-{\mu_A}  & \mathsf{T} \mathsf{T} \mathsf{T}(A)  \ar[r]^-{\mu_{\mathsf{T}(A)}} \ar[d]_-{\mathsf{T}(\mu_A)} & \mathsf{T} \mathsf{T}(A)  \ar[d]^-{\mu_A}\\
        \mathsf{T} \mathsf{T}(A) \ar[r]_-{\mu_A} & \mathsf{T}(A)  & \mathsf{T} \mathsf{T}(A) \ar[r]_-{\mu_A} & \mathsf{T}(A)}\end{gathered}\end{equation} 
 \item A \textbf{$\mathsf{T}$-algebra} is a pair $(A, \nu)$ consisting of an object $A$ and a map $\nu: \mathsf{T}(A) \to A$ of $\mathbb{X}$ such that the following diagrams commute:
\begin{equation}\label{Talg}\begin{gathered} \xymatrixcolsep{5pc}\xymatrix{
        A  \ar[r]^-{\eta_A} \ar@{=}[dr]^-{}& \mathsf{T}(A) \ar[d]^-{\nu}  &  \mathsf{T} \mathsf{T}(A)  \ar[r]^-{\mu_A} \ar[d]_-{\mathsf{T}(\nu)} & \mathsf{T}(A)  \ar[d]^-{\nu}\\
        & A  & \mathsf{T}(A) \ar[r]_-{\nu} & A}\end{gathered}\end{equation}
\item For each object $A$, the \textbf{free $\mathsf{T}$-algebra} over $A$ is the $\mathsf{T}$-algebra $(\mathsf{T}(A), \mu_A)$. 
\item A \textbf{$\mathsf{T}$-algebra morphism} $f: (A, \nu)\to (B, \nu^\prime)$ is a map $f: A\to B$ such that the following diagrams commute: 
\begin{equation}\begin{gathered} \xymatrixcolsep{5pc}\xymatrix{
       \mathsf{T}(A)  \ar[r]^-{\nu} \ar[d]_-{\mathsf{T}(f)} & A \ar[d]^-{f} \\
        \mathsf{T}(B) \ar[r]_-{\nu^\prime} & B} \end{gathered}\end{equation}     
\item The category of $\mathsf{T}$-algebras and $\mathsf{T}$-algebra morphisms is denoted $\mathbb{X}^\mathsf{T}$ and is called as the \textbf{Eilenberg-Moore category of algebras} of the monad $(\mathsf{T}, \mu, \eta)$. There is a forgetful functor $\mathsf{U}^{\mathsf{T}}: \mathbb{X}^\mathsf{T} \to \mathbb{X}$, which is defined on objects as $\mathsf{U}^\mathsf{T}(A, \nu)= A$ and on maps as $\mathsf{U}^\mathsf{T}(f)=f$.  
\end{enumerate}
\end{defi}

A particular source of examples of monads which we will be interested in for this paper are monads on symmetric monoidal category whose endofunctors are of the form $A \otimes -$. Monads of this form are induced by monoids \cite{Bruguieres2007hopf} and their algebras are more commonly known as modules. 

\begin{exas}\label{monoidmonad} \normalfont For a symmetric monoidal category $(\mathbb{X}, \otimes, K)$, every object $A$ induces an endofunctor $A \otimes -: \mathbb{X} \to \mathbb{X}$ defined in the obvious way. If $(A, \nabla, \mathsf{u})$ is a monoid in $(\mathbb{X}, \otimes, K)$, then the triple $(A \otimes -, \mu^\nabla, \eta^{u})$ is a monad on $\mathbb{X}$ where the monad structure $\mu^\nabla_X: A \otimes (A \otimes X) \to A \otimes X$ and $\eta_A^\mathsf{u}: X \to A \otimes X$ are defined as follows: 
\begin{equation}\label{}\begin{gathered} \mu^\nabla_X := \xymatrixcolsep{5pc}\xymatrix{A \otimes (A \otimes X) \ar[r]^-{\alpha_{A,A,X}} & (A \otimes A) \otimes X \ar[r]^-{\nabla \otimes 1_X} & A \otimes X 
  } \\ \eta^\mathsf{u}_X := \xymatrixcolsep{5pc}\xymatrix{X \ar[r]^-{\ell^{-1}_X} & K \otimes X \ar[r]^-{\mathsf{u} \otimes 1_X} & A \otimes X 
  } \end{gathered}\end{equation}
 Conversly, if $(A \otimes -, \mu, \eta)$ is a monad on $\mathbb{X}$, then the triple $(A, \mathsf{m}^\mu, \mathsf{u}^\eta)$ is a monoid in  $(\mathbb{X}, \otimes, K)$ where $\mathsf{m}^\mu: A \otimes A \to A$ and $\mathsf{u}^\eta: K \to A$ are defined as follows: 
 \begin{equation}\label{}\begin{gathered} \mathsf{m}^\mu := \xymatrixcolsep{5pc}\xymatrix{A \otimes A \ar[r]^-{1_A \otimes \rho^{-1}_A} & A \otimes (A \otimes K) \ar[r]^-{\mu_K} & A \otimes K \ar[r]^-{\rho_A} & A 
  } \\
   \mathsf{u}^\eta := \xymatrixcolsep{5pc}\xymatrix{K \ar[r]^-{\eta_K} & A \otimes K   \ar[r]^-{\rho_A} & A 
  } \end{gathered}\end{equation}
 In fact, these constructions provide a bijective correspondence between monoid structures on $A$ and monad structures on $A \otimes -$ \cite{Bruguieres2007hopf} . For a monoid $(A, \nabla, \mathsf{u})$, the algebras of the induced monad $(A \otimes -, \mu^\nabla, \eta^{u})$ are better known as (left) $(A, \nabla, \mathsf{u})$-modules \cite{brandenburg2014tensor}. An $(A, \nabla, \mathsf{u})$-module is a pair $(M, \alpha)$ consisting of an object $M$ and a map $\nu: A \otimes M \to M$ satisfying (\ref{Talg}), that is, the following diagrams commute:
\begin{equation}\label{}\begin{gathered} \xymatrixcolsep{5pc}\xymatrix{A \otimes (A \otimes M) \ar[d]_-{1_A \otimes \nu} \ar[r]^-{\alpha_{A,A,M}} & (A \otimes A) \otimes M \ar[r]^-{\nabla \otimes 1_M} & A \otimes M \ar[d]^-{\nu} \\
A \otimes M \ar[rr]_-{\nu} && M  \\
M \ar@{=}[drr]^-{}  \ar[r]^-{\ell^{-1}_M} & K \otimes M \ar[r]^-{\mathsf{u} \otimes 1_M} & A \otimes M \ar[d]^-{\nu} \\
&& M
  } \end{gathered}\end{equation}  
  The Eilenberg-Moore category of $(A \otimes -, \mu^\nabla, \eta^{u})$ will be denoted as $\mathsf{MOD}(A, \nabla, \mathsf{u}) := \mathbb{X}^{A \otimes -}$. 
\end{exas}

\subsection{Symmetric Comonoidal Monads}

Before giving the definition, we should address the terminology regarding comonoidal monads. Comonoidal monads were originally introduced under the name Hopf monad by Moerdijk \cite{moerdijk2002monads}, but are also referred to as \emph{bimonad} in \cite{Bruguieres2011hopf,Bruguieres2007hopf} since bimonoids and Hopf monoids induce examples of such monads. However, in the linear logic community, the terminology for the dual notion is symmetric monoidal comonad \cite{bierman1995categorical, hyland2003glueing, mellies2003categorical, mellies2009categorical}. Therefore we have elected to use the same terminology in this paper and thus refer to these monads as symmetric comonoidal monads. 

\begin{defi}\label{SCmonaddef} Let $(\mathbb{X}, \otimes, K)$ be a symmetric monoidal category. A \textbf{symmetric \\ \noindent comonoidal endofunctor} \footnote{Comonoidal functors are also known as oplax monoidal functors \cite{mccrudden2002opmonoidal}.} \cite{mccrudden2002opmonoidal} on $(\mathbb{X}, \otimes, K)$ is a triple $(\mathsf{T}, \mathsf{n}, \mathsf{n}_K)$ consisting of an endofunctor $\mathsf{T}: \mathbb{X} \to \mathbb{X}$, a natural transformation $\mathsf{n}_{A,B}: \mathsf{T}(A \otimes B) \to \mathsf{T}(A) \otimes \mathsf{T}(B)$, and a map ${\mathsf{n}_K: \mathsf{T}(K) \to K}$ such that the following diagrams commute: 
  \begin{equation}\label{smcendo}\begin{gathered} 
 \xymatrixcolsep{4pc}\xymatrix{\mathsf{T}\left( A \otimes (B \otimes C) \right) \ar[d]_-{\mathsf{T}(\alpha_{A,B,C})} \ar[r]^-{\mathsf{n}_{A, B \otimes C}} & \mathsf{T}(A) \otimes \mathsf{T}(B \otimes C) \ar[r]^-{1_{\mathsf{T}(A)} \otimes \mathsf{n}_{B,C}} &   \mathsf{T}(A) \otimes \left( \mathsf{T}(B) \otimes \mathsf{T}(C) \right) \ar[d]^-{\alpha_{\mathsf{T}(A), \mathsf{T}(B), \mathsf{T}(C)}} \\
\mathsf{T}\left( (A \otimes B) \otimes C \right) \ar[r]_-{\mathsf{n}_{A \otimes B, C}} & \mathsf{T}(A \otimes B) \otimes \mathsf{T}(C) \ar[r]_-{\mathsf{n}_{A,B} \otimes 1_{\mathsf{T}(C)}} & \left( \mathsf{T}(A) \otimes \mathsf{T}(B) \right) \otimes \mathsf{T}(C) 
 } \\
  \xymatrixcolsep{4pc}\xymatrix{
\mathsf{T}(A) \ar[d]_-{\mathsf{T}(\rho^{-1}_A)} \ar@{=}[ddrr]^-{} \ar[r]^-{\mathsf{T}(\ell^{-1}_A)} & \mathsf{T}(K \otimes A) \ar[r]^-{\mathsf{n}_{K,A}} & \mathsf{T}(K) \otimes \mathsf{T}(A) \ar[d]^-{\mathsf{n}_K \otimes 1_{\mathsf{T}(A)}} \\
 \mathsf{T}(A \otimes K) \ar[d]_-{\mathsf{n}_{A,K}}& &  K \otimes \mathsf{T}(A) \ar[d]^-{\ell_{\mathsf{T}(A)}}  \\
 \mathsf{T}(A) \otimes \mathsf{T}(K) \ar[r]_-{1_{\mathsf{T}(A)} \otimes \mathsf{n}_K}& \mathsf{T}(A) \otimes K \ar[r]_-{\rho_{\mathsf{T}(A)}} &  \mathsf{T}(A) 
 } \\
  \xymatrixcolsep{4pc}\xymatrix{
 \mathsf{T}(A \otimes B)  \ar[r]^-{\mathsf{T}(\sigma_{A,B})}  \ar[d]_-{\mathsf{n}_{A,B}} & \mathsf{T}( B \otimes  A) \ar[d]^-{\mathsf{n}_{B,A}} \\
  \mathsf{T}(A) \otimes \mathsf{T}(B) \ar[r]_-{\sigma_{\mathsf{T}(A),\mathsf{T}(B)}} &  \mathsf{T}(B) \otimes \mathsf{T}(A)
 }
 \end{gathered}\end  {equation}
\item A \textbf{symmetric comonoidal monad} \cite{moerdijk2002monads} on $(\mathbb{X}, \otimes, K)$ is a quintuple $(\mathsf{T}, \mu, \eta, \mathsf{n}, \mathsf{n}_{K})$ consisting of a monad $(\mathsf{T}, \mu, \eta)$ and a symmetric comonoidal endofunctor $(\mathsf{T}, \mathsf{n}, \mathsf{n}_{K})$ such that the following diagrams commute: 
\begin{equation}\label{symbimonad}\begin{gathered} \xymatrixcolsep{4.5pc}\xymatrix{ \mathsf{T} \mathsf{T} (A \otimes B) \ar[d]_-{\mathsf{T}(\mathsf{n}_{A,B})} \ar[r]^-{\mu_{A  \otimes B}} &\mathsf{T}(A \otimes B) \ar[dd]^-{\mathsf{n}_{A,B}} & \mathsf{T} \mathsf{T}(K) \ar[r]^-{\mu_K} \ar[d]_-{\mathsf{T}(\mathsf{n}_{K})} & \mathsf{T}(K) \ar[d]^-{\mathsf{n}_{K}}\\
\mathsf{T}(\mathsf{T}(A) \otimes\mathsf{T}(B)) \ar[d]_-{\mathsf{n}_{\mathsf{T}(A),\mathsf{T}(B)}} & &  \mathsf{T}(K) \ar[r]_-{\mathsf{n}_{K}} & K   \\
\mathsf{T}\mathsf{T}(A) \otimes \mathsf{T}\mathsf{T}(B) \ar[r]_-{\mu_A \otimes \mu_B} & \mathsf{T}(A) \otimes\mathsf{T}(B) 
  } \\  
  \xymatrixcolsep{5pc}\xymatrix{ A \otimes B \ar[dr]_-{\eta_A \otimes \eta_B} \ar[r]^-{\eta_A} & \mathsf{T}(A \otimes B) \ar[d]^-{\mathsf{n}_{A,B}} &  K \ar@{=}[dr]^-{} \ar[r]^-{\eta_A} & \mathsf{T}(K) \ar[d]^-{\mathsf{n}_{K}}\\
& \mathsf{T}(A) \otimes\mathsf{T}(B) & & K 
  } 
  \end{gathered}\end  {equation}
\end{defi} 

As observed by Moerdijk in \cite[Proposition 3.2]{moerdijk2002monads}, the Eilenberg-Moore category of a symmetric comonoidal monad is a symmetric monoidal category. In fact for a monad on a symmetric monoidal category, symmetric comonoidal structures on the monad are in bijective correspondence with symmetric monoidal structures on the Eilenberg-Moore category which are strictly preserved by the forgetful functor \cite{wisbauer2008algebras}. Indeed, let $(\mathsf{T}, \mu, \eta, \mathsf{n}, \mathsf{n}_{K})$ be a symmetric comonoidal monad on a symmetric monoidal category $(\mathbb{X}, \otimes, K)$. Then $\left(\mathbb{X}^\mathsf{T}, \otimes^\mathsf{n}, (K, \mathsf{n}_{K}), \tilde{\alpha}, \tilde{\ell}, \tilde{\rho}, \tilde{\sigma} \right)$ is a symmetric monoidal category. Here the tensor product $\otimes^\mathsf{n}$ of a pair of $\mathsf{T}$-algebras $(A, \nu)$ and $(B, \nu^\prime)$, is defined as the pair $(A, \nu) \otimes^\mathsf{n} (B, \nu^\prime) := (A \otimes B, \nu \otimes^\mathsf{n} \nu^\prime)$ where $\nu \otimes^\mathsf{n} \nu^\prime$ is defined as follows:
\begin{equation}\label{eilentensor}\begin{gathered} \nu \otimes^\mathsf{n} \nu^\prime := \xymatrixcolsep{5pc}\xymatrix{\mathsf{T}(A \otimes B) \ar[r]^-{\mathsf{n}_{A,B}} & \mathsf{T}(A) \otimes\mathsf{T}(B)\ar[r]^-{\nu \otimes \nu^\prime} & A \otimes B
  } \end  {gathered}\end  {equation}
while the tensor product of $\mathsf{T}$-algebra morphisms $f$ and $g$ is defined as $f \otimes^{\mathsf{n}} g := f \otimes g$. This tensor product is well defined by the left hand side diagrams of (\ref{symbimonad}) and also that $\mathsf{n}$ is a $\mathsf{T}$-algebra morphism. While the right hand side diagrams of (\ref{symbimonad}) imply that $(K, \mathsf{n}_K)$ is a $\mathsf{T}$-algebras. This is well defined by the two left hand side diagrams of (\ref{symbimonad}). The natural isomorphisms $\tilde{\alpha}$, $\tilde{\ell}$, $\tilde{\rho}$ and $\tilde{\sigma}$ are defined as:
\begin{align*}
\tilde{\alpha}_{(A,\nu), (B, \nu^\prime), (C. \nu^\prime\prime)} &:= \alpha_{A,B,C} &  \tilde{\ell}_{(A,\nu)} &:= \ell_A\\
\tilde{\sigma}_{(A,\nu), (B, \nu^\prime)} &:= \sigma_{A,B} & \tilde{\rho}_{(A,\nu)} &:= \rho_A
\end{align*}
which are well defined by the diagrams of (\ref{smcendo}). 

As was the case for monads, particular examples of symmetric comonoidal monads of interest for this paper are those whose endofunctors are of the form $A \otimes -$. Symmetric comonoidal monads of this form are induced by cocommutative bimonoids, which is why Brugui{\`e}res, Lack, and Virelizier refer to comonoidal monads as bimonads \cite{Bruguieres2011hopf}. 

\begin{exas}\label{bimonadex} \normalfont Let $(A, \nabla, \mathsf{u}, \Delta, \mathsf{e})$ be a cocommutative bimonoid in a symmetric monoidal category $(\mathbb{X}, \otimes, K)$. Then $(A \otimes -, \mu^\nabla, \eta^\mathsf{u}, \mathsf{n}^\Delta, \mathsf{n}^\mathsf{e}_K)$ is a symmetric comonoidal monad on $(\mathbb{X}, \otimes, K)$ where the monad $(A \otimes -. \mu^\nabla, \eta^\mathsf{u})$ is defined as in Example \ref{monoidmonad} while \\ \noindent $\mathsf{n}^\Delta_{X,Y}: A \otimes (X \otimes Y) \to (A \otimes X) \otimes (A \otimes Y)$ and $\mathsf{n}^\mathsf{e}_K: A \otimes K \to K$ are defined as follows: 
\begin{equation}\label{bimonadn}\begin{gathered} \mathsf{n}^\Delta_{X,Y} := \xymatrixcolsep{5pc}\xymatrixrowsep{0.5pc}\xymatrix{ A \otimes (X \otimes Y) \ar[r]^-{\Delta \otimes (1_X \otimes 1_Y)} &\\
 (A \otimes A) \otimes (X \otimes Y) \ar[r]^-{\tau_{A,A,X,Y}} & (A \otimes X) \otimes (A \otimes Y) 
  } \\ \\
   \mathsf{n}^\mathsf{e}_K := \xymatrixcolsep{5pc}\xymatrix{ A \otimes K \ar[r]^-{\rho_A} & A \ar[r]^-{\mathsf{e}} & K
  } \end{gathered}\end{equation}
where $\tau$ is the interchange map as defined in (\ref{interchange}). Therefore $\left( \mathsf{MOD}(A, \nabla, \mathsf{u}), \otimes^{\mathsf{n}^\Delta}, (K, \mathsf{n}^\mathsf{e}_K) \right)$ is a symmetric monoidal category. Conversly if $(A \otimes -, \mu, \eta, \mathsf{n}, \mathsf{n}_K)$ is a symmetric comonoidal monad on $(\mathbb{X}, \otimes, K)$, then $(A, \mathsf{m}^\mu, \mathsf{u}^\eta, \Delta^\mathsf{n}, \mathsf{e}^{n_K})$ is a cocommutative bimonoid where the monoid $(A, \mathsf{m}^\mu, \mathsf{u}^\eta)$ is defined as in Example \ref{monoidmonad} while $\Delta^{n}: A \to A \otimes A$ and $\mathsf{e}^{n_K}: A \to K$ are defined as follows:
\begin{equation}\label{}\begin{gathered} \Delta^\mathsf{n} := \xymatrixcolsep{2.75pc}\xymatrix{ A \ar[r]^-{\rho^{-1}_A} & A \otimes K \ar[r]^-{1_A \otimes \rho^{-1}_K} & A \otimes (K \otimes K) \ar[r]^-{\mathsf{n}_{K,K}} & (A \otimes K) \otimes (A \otimes K) \ar[r]^-{\rho_A \otimes \rho_A} & A \otimes A
  } \\
\mathsf{e}^{n_K} := \xymatrixcolsep{3pc}\xymatrix{A \ar[r]^-{\rho^{-1}_A} & A \otimes K \ar[r]^-{\mathsf{n}_K} & K
  } \end{gathered}\end{equation}
These constructions are inverse to each other and gives a bijective correspondence between cocommutative bimonoid structure on $A$ and symmetric comonoidal monad structure on $A \otimes -$ \cite{Bruguieres2007hopf}.  \end{exas}

\subsection{Symmetric Hopf Monads}

We now address how to lift symmetric monoidal closed structure to Eilenberg-Moore categories of symmetric comonoidal monads in such a way that the forgetful functor also preserves the closed structure strictly. For this, we turn to Brugui{\`e}res, Lack, and Virelizier's notion of a Hopf monad (which differs to Moerdijk's definition \cite{moerdijk2002monads}). Hopf monads were originally introduced by Brugui{\`e}res and Virelizier for monoidal categories with duals \cite{Bruguieres2007hopf}, but the definition of Hopf monads was later extended to arbitrary monoidal categories by the two previous authors and Lack \cite{Bruguieres2011hopf}. We choose the latter of the two as the definition is somewhat simpler.

\begin{defi}\label{Hopfmonaddef} Let $(\mathsf{T}, \mu, \eta, \mathsf{n}, \mathsf{n}_{K})$ be a symmetric comonoidal monad on a symmetric monoidal category $(\mathbb{X}, \otimes, K)$. The (left and right) \textbf{fusion operators} \cite{Bruguieres2011hopf} of $(\mathsf{T}, \mu, \eta, \mathsf{n}, \mathsf{n}_{K})$ are the natural transformations $\mathsf{h}^l_{A,B}:  \mathsf{T}(A \otimes \mathsf{T} B) \to \mathsf{T}(A) \otimes\mathsf{T}(B)$ and $\mathsf{h}^r_{A,B}: \mathsf{T}(\mathsf{T}(A) \otimes B) \to \mathsf{T}(A) \otimes\mathsf{T}(B)$ defined respectively as: 
    \[\mathsf{h}^l_{A,\mathsf{T}(B)} :=  \xymatrixcolsep{5pc}\xymatrix{ \mathsf{T}(A \otimes \mathsf{T} B) \ar[r]^-{\mathsf{n}_{A,B}} & \mathsf{T}(A) \otimes \mathsf{T}\mathsf{T}(B)\ar[r]^-{1 \otimes \mu_A} & \mathsf{T}(A) \otimes\mathsf{T}(B) 
  } \]
    \[\mathsf{h}^r_{A,B} :=  \xymatrixcolsep{5pc}\xymatrix{\mathsf{T}(\mathsf{T}(A) \otimes B) \ar[r]^-{\mathsf{n}_{\mathsf{T}(A),B}} & \mathsf{T}\mathsf{T}(A) \otimes\mathsf{T}(B)\ar[r]^-{\mu_A \otimes 1} & \mathsf{T}(A) \otimes\mathsf{T}(B)
  } \]
  A \textbf{symmetric Hopf monad} \cite{Bruguieres2011hopf} on a symmetric monoidal category is a symmetric comonoidal monad whose fusion operators are natural isomorphisms.
    \end{defi}
    
Extending on \cite[Theorem 3.6]{Bruguieres2011hopf}, the Eilenberg-Moore category of a Hopf monad on a symmetric monoidal closed category is again a symmetric monoidal closed category such that the forgetful functor preserves the symmetric monoidal closed structure strictly. Explicitly, let $(\mathsf{T}, \mu, \eta, \mathsf{n}, \mathsf{n}_{K})$ be a symmetric Hopf monad on a symmetric monoidal closed category $(\mathbb{X}, \otimes, K)$, with internal hom $\multimap$ and evaluation map $\mathsf{ev}$. Then $\left(\mathbb{X}^\mathsf{T}, \otimes^\mathsf{n}, (K, \mathsf{n}_{K})\right)$ is a symmetric monoidal closed category where the internal hom of a pair of $(A, \nu)$ and $(B, \nu^\prime)$, is defined as the pair $(A, \nu) \multimap (B, \nu^\prime) := (A \multimap B, \overline{\gamma}_{\nu, \nu^\prime})$ where $\overline{\gamma}_{\nu, \nu^\prime}: \mathsf{T}(A \multimap B)\to A \multimap B$ is the currying (as defined in (\ref{closedeq})) of the map ${\gamma_{\nu, \nu^\prime}: \mathsf{T}(A \multimap B) \otimes A\to B}$ defined as the following composite: 
  \[\xymatrixrowsep{1pc}\xymatrixcolsep{4pc}\xymatrix{\mathsf{T}(A \multimap B) \otimes A \ar[r]^-{1_{\mathsf{T}(A \multimap B)} \otimes \eta_A} & \mathsf{T}(A \multimap B) \otimes \mathsf{T}(A) \ar[r]^-{\mathsf{h}^{l^{-1}}_{A \multimap B, A}} &\\
   \mathsf{T} (( A \multimap B) \otimes \mathsf{T}(A)) \ar[r]^-{\mathsf{T}(1_{A \multimap B} \otimes \nu)} & \\
  \mathsf{T} (( A \multimap B) \otimes A) \ar[r]_-{\mathsf{T}(\mathsf{ev}_{A,B})} &  \ar[r]_-{\nu^\prime} \mathsf{T}(B)& B
  } \]
and with evaluation map $\mathsf{ev}_{(A, \nu),(B, \nu^\prime)} := \mathsf{ev}_{A,B}$. 

The name Hopf monad comes from the fact that every Hopf monoid (with invertible antipode, which is the case for cocommutative Hopf monoids) defines a Hopf monad whose endofunctor is of the form $H \otimes -$. 

\begin{exas}\label{hopfmonadex} \normalfont Let $(H, \nabla, \mathsf{u}, \Delta, \mathsf{e}, \mathsf{S})$ be a cocommtuative Hopf monoid in a symmetric monoidal category $(\mathbb{X}, \otimes, K)$. Then for the induced $(H \otimes -, \mu^\nabla, \eta^\mathsf{u}, \mathsf{n}^\Delta, \mathsf{n}^\mathsf{e}_K)$ symmetric comonoidal monad as defined in Example \ref{bimonadex}, the inverse of the left fusion operator, $\mathsf{h}^{l^{-1}}_{A,B}: (H \otimes A) \otimes (H \otimes B) \to H \otimes \left( A \otimes (H \otimes B) \right)$, is given by the following composite: 
  \[ \xymatrixrowsep{1pc} \xymatrixcolsep{7pc}\xymatrix{ (H \otimes A) \otimes (H \otimes B) \ar[r]^-{\alpha_{H,A, H \otimes B}} & \\
  H \otimes \left( A \otimes (H \otimes B) \right) \ar[r]^-{\mathsf{n}_{A, H \otimes B}^\Delta} & \\
  (H \otimes A) \otimes \left(H \otimes (H \otimes B) \right) \ar[r]^-{(1_H \otimes 1_A) \otimes (\mathsf{S} \otimes (1_H \otimes 1_B))} & \\
  (H \otimes A) \otimes \left(H \otimes (H \otimes B) \right) \ar[r]^-{(1_H \otimes 1_A) \otimes \mu^\mathsf{m}_B} & \\
  (H \otimes A) \otimes (H \otimes B) \ar[r]^-{\alpha_{H,A,H \otimes B}} & H \otimes \left( A \otimes (H \otimes B) \right)
  } \]
and the inverse of the right fusion operator is defined similarly. Therefore if $(\mathbb{X}, \otimes, K)$ is closed, then so is $\left( \mathsf{MOD}(H, \nabla, \mathsf{u}), \otimes^{\mathsf{n}^\Delta}, (K, \mathsf{n}^\mathsf{e}_K) \right)$. Conversly if $(H \otimes -, \mu, \eta, \mathsf{n}, \mathsf{n}_K)$ is a symmetric Hopf monad on $(\mathbb{X}, \otimes, K)$, then $(A, \mathsf{m}^\mu, \mathsf{u}^\eta, \Delta^\mathsf{n}, \mathsf{e}^{n_K}, \mathsf{S}^\mathsf{h})$ is a cocommutative Hopf monoid where the cocommutative bimonoid structure is defined as in Example \ref{bimonadex} while the antipode $\mathsf{S}^\mathsf{h}: H \to H$ is defined as the following composite: 
  \[ \xymatrixrowsep{0.75pc}  \xymatrixcolsep{4pc}\xymatrix{H \ar[r]^-{\rho^{-1}_H} & H \otimes K \ar[r]^-{\rho^{-1}_{H \otimes K}} &\\
   (H \otimes K) \otimes K \ar[rr]^-{(1_H \otimes 1_K) \otimes \eta^{\mathsf{u}}_K} && \\
   (H \otimes K) \otimes (H \otimes K) \ar[r]^-{\mathsf{h}^{\ell^{-1}}_{K,K}} &
   H \otimes \left( K \otimes (H \otimes K) \right) \ar[r]^-{\alpha^{-1}_{H,K,H \otimes K}} & \\ (H \otimes K) \otimes (H \otimes K) \ar[r]^-{\mathsf{n}^\mathsf{e}_K \otimes \rho_H} & K \otimes H \ar[r]^-{\ell_H} & H
  } \]
  For more details on this example and a string diagram representation, see \cite[Example 2.10]{Bruguieres2011hopf}. 
\end{exas}

\subsection{Comonads and their Coalgebras}

The dual notion of a monad is that a comonad. Comonads are key structures in categorical semantics of the exponential fragment of linear logic. Throughout this paper, we will encounter various different kinds of comonads with added structure such as symmetric monoidal comonads (Definition \ref{SMComonadef}), coalgebra modalities (Definition \ref{coalgmod}), and monoidal coalgebra modalities (Definition \ref{moncoalgmod}). 

\begin{defi} Let $\mathbb{X}$ be a category. 
\begin{enumerate}[{\em (i)}]
\item A \textbf{comonad} \cite{mac2013categories} on a category $\mathbb{X}$ is a triple $(\oc, \delta, \varepsilon)$ consisting of a functor ${\oc: \mathbb{X} \to \mathbb{X}}$ and two natural transformations $\delta_A: \oc(A) \to \oc \oc(A)$ and $\varepsilon_A: \oc(A) \to A$ such that the following diagrams commute:
      \begin{equation}\label{comonadeq}\begin{gathered} \xymatrixcolsep{5pc}\xymatrix{ 
        \oc(A)  \ar[r]^-{\delta_A} \ar[d]_-{\delta_A} \ar@{=}[dr]^-{}& \oc \oc(A) \ar[d]^-{\varepsilon_{\oc(A)}}  & \oc(A)  \ar[r]^-{\delta_A} \ar[d]_-{\delta_A} & \oc \oc(A)  \ar[d]^-{\delta_{\oc(A)}}\\
        \oc \oc(A) \ar[r]_-{\oc(\varepsilon_A)} & \oc(A)  & \oc \oc(A) \ar[r]_-{\oc(\delta_A)} & \oc \oc \oc(A)}\end{gathered}\end{equation}
\item A \textbf{$\oc$-coalgebra} is a pair $(A, \omega)$ consisting of an object $A$ and a map $\omega: A \to \oc(A)$ of $\mathbb{X}$ such that the following diagrams commute:
        \begin{equation}\label{!coalg}\begin{gathered} \xymatrixcolsep{5pc}\xymatrix{
        A  \ar[r]^-{\omega} \ar@{=}[dr]^-{}& \oc(A) \ar[d]^-{\varepsilon_A}  &  A  \ar[r]^-{\omega} \ar[d]_-{\omega} &  \oc(A)  \ar[d]^-{\delta_A}\\
        &  A  & \oc(A) \ar[r]_-{\oc(\omega)} & \oc \oc(A)}\end{gathered}\end{equation}
\item For each object $A$, the \textbf{cofree $\oc$-coalgebra} over $A$ is the $\oc$-coalgebra $(\oc(A), \delta_A)$. 
\item A \textbf{$\oc$-coalgebra morphism} $f: (A, \omega)\to (B, \omega^\prime)$ is a map $f: A\to B$ such that the following diagram commutes: 
\begin{equation}\label{!coalgmap}\begin{gathered} \xymatrixcolsep{5pc}\xymatrix{
       A  \ar[r]^-{\omega} \ar[d]_-{f} & A \ar[d]^-{\oc(f)} \\
        B \ar[r]_-{\omega^\prime} & \oc(B) } \end{gathered}\end{equation} 
\item The category of $\oc$-coalgebras and $\oc$-coalgebra morphisms is denoted $\mathbb{X}^\oc$ and is also known as the \textbf{Eilenberg-Moore category of coalgebras} of the comonad $(\oc, \delta, \varepsilon)$. There is a forgetful functor $\mathsf{U}^{\oc}: \mathbb{X}^\oc \to \mathbb{X}$, which is defined on objects as $\mathsf{U}^\oc(A, \omega)= A$ and on maps as $\mathsf{U}^\oc(f)=f$.      
\end{enumerate}         
\end{defi}

As comonads are dual to monads, comonoids induce examples of comonads in the dual sense of Example \ref{monoidmonad}. Though these are not usually the examples of comonads of interest in linear logic. We discuss examples of comonads of interest in linear logic and to this paper in Example \ref{linex}.  

\subsection{Symmetric Monoidal Comonads}\label{SMcomonadsec}

\begin{defi}\label{SMComonadef} Let $(\mathbb{X}, \otimes, K)$ be a symmetric monoidal category. A \textbf{symmetric monoidal endofunctor} on $(\mathbb{X}, \otimes, K)$ is a triple $(\oc, \mathsf{m}, \mathsf{m}_K)$ consisting of an endofunctor $\oc: \mathbb{X} \to \mathbb{X}$, a natural transformation $\mathsf{m}_{A,B}: \oc(A) \otimes \oc(B)\to \oc(A \otimes B)$, and a map ${\mathsf{m}_K: K\to \oc(K)}$ such that the following diagrams commute: 

   \begin{equation}\label{smcmonendo}\begin{gathered} \xymatrixcolsep{4.5pc}\xymatrix{\oc(A) \otimes \left( \oc(B) \otimes \oc(C) \right) \ar[d]_-{\alpha_{\oc(A), \oc(B), \oc(C)}} \ar[r]^-{1_{\oc(A)} \otimes \mathsf{m}_{B,C}} & \oc(A) \otimes \oc(B \otimes C) \ar[r]^-{\mathsf{m}_{A, B \otimes C}} & \oc\left( A \otimes (B \otimes C) \right) \ar[d]^-{\oc(\alpha_{A,B,C})} \\
  \left( \oc(A) \otimes \oc(B) \right) \otimes \oc(C) \ar[r]_-{\mathsf{m}_{A,B} \otimes 1_{\oc(C)}} & \oc(A \otimes B) \otimes \oc(C) \ar[r]_-{\mathsf{m}_{A \otimes B, C}}& \oc\left( (A \otimes B) \otimes C \right) 
 } \\
 \xymatrixcolsep{4.5pc}\xymatrix{\oc(A) \ar[d]_-{\rho^{-1}_{\oc(A)}} \ar@{=}[ddrr]^-{} \ar[r]^-{\ell^{-1}_{\oc(A)}} & K \otimes \oc(A) \ar[r]^-{\mathsf{m}_K \otimes 1_{\oc(A)}} & \oc(K) \otimes \oc(A) \ar[d]^-{\mathsf{m}_{K,A}} \\
\oc(A) \otimes K \ar[d]_-{1_{\oc(A)} \otimes \mathsf{m}_K} && \oc(K \otimes A) \ar[d]^-{\oc(\ell_A)}  \\
\oc(A) \otimes \oc(K) \ar[r]_-{\mathsf{m}_{A,K}}& \oc(A \otimes K) \ar[r]_-{\oc(\rho_A)}& \oc(A)
 } \\
 \xymatrixcolsep{4.5pc}\xymatrix{\oc(A) \otimes \oc(B)  \ar[r]^-{\sigma_{A,B}}  \ar[d]_-{\mathsf{m}_{A,B}} & \oc(B) \otimes \oc(A) \ar[d]^-{\mathsf{m}_{B,A}} \\
\oc(A \otimes B) \ar[r]_-{\oc(\sigma_{\oc(A),\oc(B)})} &  \oc(B \otimes A)
 }
      \end{gathered}\end{equation}
A \textbf{symmetric monoidal comonad} on $(\mathbb{X}, \otimes, K)$ is a quintuple $(\oc, \delta, \varepsilon, \mathsf{m}, \mathsf{m}_K)$ consisting of a comonad $(\oc, \delta, \varepsilon)$ and a symmetric monoidal endofunctor $(\oc, \mathsf{m}, \mathsf{m}_K)$ such that the following diagrams commute: 
      \begin{equation}\label{symcomonad}\begin{gathered}
\xymatrixcolsep{5pc}\xymatrix{\oc(A) \otimes \oc(B)  \ar[d]_-{\delta_A \otimes \delta_B} \ar[r]^-{\mathsf{m}_{A,B}}  & \oc(A \otimes B) \ar[dd]^-{\delta_{A \otimes B}} & K \ar[r]^-{ \mathsf{m}_K} \ar[d]_-{ \mathsf{m}_K} & \oc(K) \ar[d]^-{\delta_K}  \\
  \oc\oc(A) \otimes \oc\oc(B)  \ar[d]_-{\mathsf{m}_{\oc(A),\oc(B)}}  & &\oc(K) \ar[r]_-{\oc( \mathsf{m}_K)}& \oc \oc(K) \\
  \oc (\oc(A) \otimes \oc(B))  \ar[r]_-{\oc(\mathsf{m}_{A,B})} & \oc(A \otimes B)
  }   
  \\
  \xymatrixcolsep{5pc}\xymatrix{\oc(A) \otimes \oc(B) \ar[dr]_-{\varepsilon_A \otimes \varepsilon_B} \ar[r]^-{\mathsf{m}_{A,B}} & \oc(A \otimes B) \ar[d]^-{\varepsilon_{A \otimes B}} & K \ar@{=}[dr]^-{} \ar[r]^-{ \mathsf{m}_K} & \oc(K) \ar[d]^-{\varepsilon_A} \\
& A \otimes B  & & \oc(K) 
  }   
\end{gathered}\end{equation}
\end{defi} 

As symmetric monoidal comonads are the dual notion of symmetric comonoidal monads, their Eilenberg-Moore categories are again symmetric monoidal categories. Let $(\oc, \delta, \varepsilon, \mathsf{m}, \mathsf{m}_K)$ be a symmetric monoidal comonad on a symmetric monoidal category $(\mathbb{X}, \otimes, K)$. Then (abusing notation slightly) $\left(\mathbb{X}^\oc, \otimes^\mathsf{m}, (K, \mathsf{m}_{K}), \tilde{\alpha}, \tilde{\ell}, \tilde{\rho}, \tilde{\sigma} \right)$ is a symmetric monoidal category where in particular the tensor product $\otimes^\mathsf{m}$ of a pair of $\oc$-coalgebras $(A, \omega)$ and $(B, \omega^\prime)$, is defined as the pair $(A, \omega) \otimes^\mathsf{m} (B, \omega^\prime) := (A \otimes B, \omega \otimes^\mathsf{n} \omega^\prime)$ where $\omega \otimes^\mathsf{m} \omega^\prime$ is defined as follows:
\begin{equation}\label{coeilentensor}\begin{gathered} \omega \otimes^\mathsf{m} \omega^\prime := \xymatrixcolsep{5pc}\xymatrix{A \otimes B \ar[r]^-{\omega \otimes \omega^\prime} & \oc(A) \otimes \oc(B)\ar[r]^-{\mathsf{m}_{A,B}} & \oc(A \otimes B) 
  } \end  {gathered}\end  {equation}
while the tensor product of $\oc$-coalgebra morphisms $f$ and $g$ is defined as $f \otimes^{\mathsf{m}} g := f \otimes g$. The diagrams of (\ref{symcomonad}) state that $\mathsf{m}$ is a $\oc$-coalgebra morphism and that $(K, \mathsf{m}_{K})$ is a $\oc$-coalgebra. The diagrams of (\ref{smcmonendo}) are the required coherences to make $\alpha$, $\ell$, $\rho$, and $\sigma$ all $\oc$-coalgebra morphisms.

\section{Coalgebra Modalities and Linear Categories}\label{coalgsec}

In this section, we review the categorical semantics of the exponential modality in the multiplicative and exponential fragment of linear logic ($\mathsf{MELL}$). In particular, we review the notions of coalgebra modalities, monoidal coalgebra modalities, and linear categories. We also discuss the particular case when the monoidal coalgebra modality is induced by cofree cocommutative comonoids, and linear categories of this form are better known as Lafont categories. For more details and an in-depth introduction on the categorical semantics of $\mathsf{MELL}$ we refer the reader to  \cite{bierman1995categorical, hyland2003glueing, mellies2003categorical, mellies2009categorical}. 

\subsection{Coalgebra Modalities}

Coalgebra modalities were defined by Blute, Cockett, and Seely when they introduced differential categories \cite{blute2006differential} and are a strictly weaker notion of monoidal coalgebra modalities. While monoidal coalgebra modalities are much more popular as they give categorical models of the exponential fragment of linear logic, coalgebra modalities have sufficient structure to axiomatize differentiation. Therefore, we believe it of interest to discuss coalgebra modalities in order to also discuss lifting differential category structure (see Section \ref{diffsec}). Interesting examples of coalgebra modalities which are not monoidal can be found in \cite{Blute2019}. 

 \begin{defi}\label{coalgmod} A \textbf{coalgebra modality} \cite{blute2006differential} on a symmetric monoidal category $(\mathbb{X}, \otimes, K)$ is a quintuple $(\oc, \delta, \varepsilon, \Delta, \mathsf{e})$ consisting of a comonad $(\oc, \delta, \varepsilon)$, a natural transformation ${\Delta_A: \oc(A)\to \oc(A) \otimes \oc(A)}$, and a natural transformation $\mathsf{e}_A: \oc(A)\to K$ such that for each object $A$, the triple $(\oc(A), \Delta_A, \mathsf{e}_A)$ is a cocommutative comonoid and $\delta_A: (\oc(A), \Delta_A, \mathsf{e}_A) \to (\oc \oc(A), \Delta_{\oc(A)}, \mathsf{e}_{\oc(A)})$ is a comonoid morphism. \end{defi}
 
 Note that naturality of $\Delta$ and $\mathsf{e}$ is equivalent to asking that for every map $f: A \to B$, $\oc(f): (\oc(A), \Delta_A, \mathsf{e}_A) \to (\oc(B), \Delta_B, \mathsf{e}_B)$ is a comonoid morphism. Therefore it immediately follows that every comonad on a Cartesian monoidal category is a coalgebra modality. 
 
 \begin{lem}\label{coalgCart} Let $(\oc, \delta, \varepsilon)$ be a comonad on a Cartesian monoidal category $(\mathbb{X}, \otimes, K)$. Then $(\oc, \delta, \varepsilon, {\hat \Delta}, \mathsf{e})$ is a coalgebra modality on $(\mathbb{X}, \otimes, K)$ where $(\oc(A), {\hat \Delta}_A, \mathsf{e}_A)$ is the unique cocommutative comonoid structure on $\oc(A)$ from Lemma \ref{Cartcom}. \end{lem}

 Every $\oc$-coalgebra of a coalgebra modality comes equipped with a cocommutative comonoid structure \cite{mellies2009categorical,schalk2004categorical}. Let $(\oc, \delta, \varepsilon, \Delta, \mathsf{e})$ be a coalgebra modality on a symmetric monoidal category $(\mathbb{X}, \otimes, K)$ and $(A, \omega)$ an $\oc$-coalgebra of $(\oc, \delta, \varepsilon)$. Then the triple $(A, \Delta^\omega, \mathsf{e}^\omega)$ is a cocommutative comonoid where the comultiplication and counit are defined as follows: 
\begin{equation}\label{!coalgcom}\begin{gathered}
\Delta^\omega := \xymatrixcolsep{5pc}\xymatrix{A \ar[r]^-{\omega} & \oc(A) \ar[r]^-{\Delta_A} & \oc(A) \otimes \oc(A) \ar[r]^-{\varepsilon_A \otimes \varepsilon_A} & A \otimes A  
  } \\
   \mathsf{e}^\omega :=  \xymatrixcolsep{5pc} \xymatrix{A \ar[r]^-{\omega} & \oc(A) \ar[r]^-{\mathsf{e}_A} & K   } 
 \end  {gathered}\end  {equation} 
Notice that since $\delta_A$ is a comonoid morphism, when applying this construction to a cofree $\oc$-coalgebra $(\oc(A), \delta_A)$ we re-obtain $\Delta_A$ and $\mathsf{e}_A$, that is, $\Delta_A^{\delta_A}=\Delta_A$ and $\mathsf{e}_A^{\delta_A}=\mathsf{e}_A$. Furthermore, by naturality of $\Delta_A$ and $\mathsf{e}_A$, every $\oc$-coalgebra morphism becomes a comonoid morphism on the induced comonoid structures. In particular this implies that for every $\oc$-coalgebra $(A, \omega)$, we have that $\omega$ is a comonoid morphism, that is, the following diagrams commute: 
   \begin{equation}\label{omegacomon}\begin{gathered} \xymatrixcolsep{5pc}\xymatrix{
       A   \ar[r]^-{\Delta^\omega} \ar[d]_-{\omega} & A \otimes A \ar[d]^-{\omega \otimes \omega} & A  \ar[r]^-{\omega} \ar[dr]_-{\mathsf{e}^\omega} &  \oc(A) \ar[d]^{\mathsf{e}_A}\\
        \oc  A  \ar[r]_-{\Delta_A} &  \oc  A \otimes  \oc(A) & & K}\end{gathered}\end{equation}

\subsection{Monoidal Coalgebra Modalities and Linear Categories}\label{moncoalgsec}

 \begin{defi}\label{moncoalgmod} A \textbf{monoidal coalgebra modality} \cite{blute2015cartesian} (also called a \textbf{linear exponential modality} \cite{schalk2004categorical}) on a symmetric monoidal category $(\mathbb{X}, \otimes, K)$ is a septuple $(\oc, \delta, \varepsilon, \Delta, \mathsf{e}, \mathsf{m},\mathsf{m}_K)$ such that $(\oc, \delta, \varepsilon, \mathsf{m},\mathsf{m}_K)$ is a symmetric monoidal comonad and $(\oc, \delta, \varepsilon, \Delta, \mathsf{e})$ is a coalgebra modality, and such that $\Delta$ and $\mathsf{e}$ are monoidal transformations, that is, the following diagrams commute:  
\begin{equation}\label{Deltamonoidal}\begin{gathered}
\xymatrixcolsep{5pc}\xymatrix{\oc(A) \otimes \oc(B)  \ar[d]_-{\Delta_A \otimes \Delta_A} \ar[r]^-{\mathsf{m}_{A,B}}  & \oc(A \otimes B) \ar[dd]^-{\Delta_A}  \\
  \oc(A) \otimes \oc(A) \otimes \oc(B) \otimes \oc(B) \ar[d]_-{\tau_{\oc(A), \oc(A), \oc(B), \oc(B)}}    \\
  \oc(A) \otimes \oc(B) \otimes \oc(A) \otimes \oc(B) \ar[r]_-{\mathsf{m}_{A,B} \otimes \mathsf{m}_{A,B}} & \oc(A \otimes B) \otimes \oc(A \otimes B)  \\
  \oc(A) \otimes \oc(B) \ar[d]_-{\mathsf{e}_A \otimes \mathsf{e}_B} \ar[r]^-{\mathsf{m}_{A,B}} & \oc(A \otimes B) \ar[d]^-{\mathsf{e}_{A \otimes B}} \\
  K \otimes K \ar[r]_-{\ell_K}& K
  }   \\
  \xymatrixcolsep{5pc}\xymatrix{ K \ar[r]^-{\mathsf{m}_K} \ar[d]_-{\ell^{-1}_K} & \oc(K) \ar[d]^-{\Delta_K} & K \ar@{=}[dr]^-{} \ar[r]^-{\mathsf{m}_K} & \oc(K) \ar[d]^-{\mathsf{e}_K} \\
K \otimes K \ar[r]_-{\mathsf{m}_K \otimes\mathsf{m}_K}  &\oc(K) \otimes \oc(K) && K
  }   
\end{gathered}\end{equation}
and also that $\Delta_A$ and $e$ are $\oc$-coalgebra morphisms, that is, the following diagrams commute:
\begin{equation}\label{Delta!map}\begin{gathered}  \xymatrixcolsep{5pc}\xymatrix{\oc(A) \ar[d]_-{\Delta_A} \ar[rr]^-{\delta_A} & & \oc \oc(A) \ar[d]^-{\oc(\Delta_A)}  \\
    \oc(A) \otimes \oc(A) \ar[r]_-{\delta_A \otimes \delta_A} & \oc \oc(A) \otimes \oc \oc(A) \ar[r]_-{\mathsf{m}_{\oc(A),\oc(A)}} & \oc(\oc(A) \otimes \oc(A)) 
  } \\
   \xymatrixcolsep{5pc}\xymatrix{
    \oc(A) \ar[d]_-{\mathsf{e}_A} \ar[r]^-{\delta_A} & \oc \oc(A) \ar[d]^-{\oc(\mathsf{e}_A)} \\
     K \ar[r]_-{\mathsf{m}_K} & \oc(K)
  } \end{gathered}\end{equation}
A \textbf{linear category} \cite{bierman1995categorical} is a symmetric monoidal closed category with a monoidal coalgebra modality. 
\end{defi}

There are many equivalent alternative descriptions of a monoidal coalgebra modality \cite{schalk2004categorical}. Of particular interest to this paper is that a monoidal coalgebra modality can be equivalently be described as a symmetric monoidal comonad whose Eilenberg-Moore category is a Cartesian monoidal category (Definition \ref{cart}). So if $(\oc, \delta, \varepsilon, \Delta, \mathsf{e}, \mathsf{m},\mathsf{m}_K)$ is a monoidal coalgebra modality on a symmetric monoidal category $(\mathbb{X}, \otimes, K)$, then $\left(\mathbb{X}^\oc, \otimes^\mathsf{m}, (K, \mathsf{m}_{K}) \right)$ (as defined in Section \ref{SMcomonadsec}) is a Cartesian monoidal category. Furthermore, the canonical cocommutative comonoid structure from Lemma \ref{Cartcom} on a $\oc$-coalgebra $(A,\omega)$ is precisely given by $\Delta^\omega$ and $\mathsf{e}^\omega$ as defined in (\ref{!coalgcom}), that is, $((A,\omega), \Delta^\omega, \mathsf{e}^\omega)$ is a cocommutative comonoid in $\left(\mathbb{X}^\oc, \otimes^\mathsf{m}, (K, \mathsf{m}_{K}) \right)$. As linear categories are categorical models of $\mathsf{MELL}$ \cite{bierman1995categorical, mellies2003categorical, mellies2009categorical}, there is no shortage of examples of linear categories throughout the literature. We provide some examples of linear categories in Example \ref{linex}, while Hyland and Schalk provide a very nice list of various kinds examples in \cite[Section 2.4]{hyland2003glueing}.

An important source of examples of monoidal coalgebra modalities are those for which $\oc(A)$ is also the cofree cocommutative comonoid over $A$ (as shown by Lafont in his Ph.D. thesis \cite{lafont1988logiques}). Monoidal coalgebra modalities with this extra couniversal property are known as \emph{free exponential modalities} \cite{mellies2017explicit} and models of linear logic with free exponential modalities are known as \emph{Lafont categories} \cite{mellies2009categorical}. While free exponential modalities have been around since the beginning with Girard's free exponential modality for coherence spaces \cite{girard1987linear}, new free exponential modalities are still being constructed and studied \cite{crubille2017free, laird2013weighted, slavnov2015banach}, which shows the importance of these kinds models. In fact, Lafont categories are by far the most common examples of linear categories.

 \begin{defi} \label{freeexpo} A monoidal coalgebra modality $(\oc, \delta, \varepsilon, \Delta, \mathsf{e}, \mathsf{m},\mathsf{m}_K)$ is said to be a \textbf{free exponential modality} \cite{mellies2017explicit,slavnov2015banach} if for each object $A$, $\oc(A)$ is a \textbf{cofree cocommutative comonoid} over $A$, that is, if $(C, \Delta, \mathsf{e})$ is a cocommutative comonoid, then for every map $f: C \to A$ there exists a unique comonoid morphism $\hat{f}: (C, \Delta, \mathsf{e}) \to (\oc(A), \Delta, \mathsf{e})$ such that the following diagram commutes:
  \[  \xymatrixcolsep{5pc}\xymatrix{C \ar@{-->}[r]^-{\exists \oc ~\hat{f}} \ar[dr]_-{f} & \oc(A) \ar[d]^-{\varepsilon_A} \\
 & A  
  } \]      
A \textbf{Lafont category} \cite{mellies2003categorical} is a linear category whose monoidal coalgebra modality is a free exponential modality.
\end{defi}

Observe that every Cartesian monoidal category trivially admits a free exponential modality. 

\begin{lem}\label{Cartlafont} Let $(\mathbb{X}, \otimes, K)$ be a Cartesian monoidal category and let $1_\mathbb{X}: \mathbb{X} \to \mathbb{X}$ be the identity functor. Then $(\oc, \delta, \varepsilon, \Delta, \mathsf{t}, \mathsf{m},\mathsf{m}_K)$ is a free exponential modality where $\delta_A = \varepsilon_A = 1_A$, $\mathsf{m}_{A,B} = 1_A \otimes 1_B$, $\mathsf{m}_K = 1_K$, and $(A, \Delta_A, \mathsf{t}_A)$ is the unique cocommutative comonoid structure on $A$ from Lemma \ref{Cartcom}. Therefore every Cartesian closed category is a Lafont category. 
\end{lem}

The Eilenberg-Moore category of a free exponential modality is equivalent to the category of cocommutative comonoids of the base symmetric monoidal category \cite{mellies2003categorical, mellies2009categorical}. 

\begin{exas}\label{linex} \normalfont We conclude this section by briefly giving some examples of Lafont categories (and therefore examples of linear categories). In particular, we give examples of free exponential modalities on the symmetric monoidal closed categories of Example \ref{SMCex}. Many other examples of free exponential modalities can be found in \cite{hyland2003glueing}, while detailed constructions of free exponential modalities can be found in \cite{mellies2017explicit,slavnov2015banach}.   
\begin{enumerate}
\item Since $(\mathsf{SET}, \times, \lbrace \ast \rbrace)$ is a Cartesian closed category, it follows from Lemma \ref{coalgCart} that every comonad on $\mathsf{SET}$ is a coalgebra modality and also by Lemma \ref{Cartlafont} that $(\mathsf{SET}, \times, \lbrace \ast \rbrace)$ is trivially a Lafont category where the free exponential modality is given by the identity functor.  
\item $(\mathsf{REL}, \times, \lbrace \ast \rbrace)$ is a Lafont category with free exponential modality $(\oc, \delta, \varepsilon, \Delta, \mathsf{e}, \mathsf{m},\mathsf{m}_{\lbrace \ast \rbrace})$ where: \\ \noindent $(i)$ The endofunctor $\oc: \mathsf{REL} \to \mathsf{REL}$ maps a set $X$ to the set of all finite bags (also known as multisets) of $X$: 
\[\oc(X)=\lbrace \llbracket x_1,...,x_n \rrbracket \vert~ x_i \in X \rbrace\] 
while for relation $R: X \to Y$, $\oc(R): \oc (X) \to \oc (Y)$ is the relation which relates a bag of $X$ to a bag of $Y$ of the same size and such that the elements of the bags are related: 
\[\oc(R)= \lbrace (\llbracket x_1,...,x_n \rrbracket, \llbracket y_1,...,y_n \rrbracket) \vert ~ (x_i, y_i) \in R \rbrace \subseteq \oc(X) \times \oc(Y)\]
$(ii)$ $\varepsilon_X: \oc(X) \to X$ is the relation which relates one element bags to their element: 
\[\varepsilon_X= \lbrace ( \llbracket x \rrbracket, x) \vert~ x \in X \rbrace \subseteq \oc(X) \times X\] 
$(iii)$ $\delta_X:  \oc(X) \to \oc(\oc(X))$ is the relation which relates a bag to the bag of all possible bag splittings of the original bag:
\begin{align*}
\delta_X= \lbrace (B, \llbracket B_1, ..., B_n \rrbracket) \vert~ B, B_i \in \oc(X), ~  B_1 \sqcup ... \sqcup B_n=B \rbrace \subseteq  \oc(X) \times  \oc\oc(X)
\end{align*}
$(iv)$ $\Delta_X: \oc(X) \to \oc(X) \times \oc(X)$ is the relation which relates a bag to the pair of all possible two bag splittings of the original bag: 
\[\Delta_X= \lbrace (B, (B_1, B_2)) \vert~ B, B_i \in \oc(X),~B_1 \sqcup B_2 = B \rbrace \subseteq \oc(X) \times \left(\oc(X) \times \oc(X) \right)\] 
$(v)$ $\mathsf{e}_X: \oc(X) \to  \lbrace \ast \rbrace$ is the relation which relates the empty bag to the single element: 
\[\mathsf{e}_X= \lbrace  (\emptyset, \ast) \rbrace \subseteq  \oc(X) \times \lbrace \ast \rbrace\] 
$(vi)$ $\mathsf{m}_{X,Y}: \oc(X) \times \oc(Y) \to \oc(X \times Y)$ is the relation which relates bags of $X$ and bags of $Y$ of the same size to their Cartesian product bag: 
\begin{align*}
\mathsf{m}_{X,Y} = &\lbrace ( \llbracket x_1,...,x_n \rrbracket, \llbracket y_1,...,y_n \rrbracket, \llbracket (x_i, y_j) \vert~ 0 \leq i, j \leq n  \rrbracket ) \vert~ x_i \in X, y_j \in Y \rbrace  \\ &\subseteq \left(\oc(X) \times \oc(Y) \right) \times \oc(X \times Y)
\end{align*}
$(vii)$ $\mathsf{m}_{\lbrace \ast \rbrace}: \lbrace \ast \rbrace \to \oc(\lbrace \ast \rbrace)$ is the relation which relates $\ast$ to every element of $\oc(\lbrace \ast \rbrace)$: 
\[\mathsf{m}_{\lbrace \ast \rbrace} := \lbrace (\ast, B) \vert~ B \in \oc(\lbrace \ast \rbrace) \rbrace \subseteq \lbrace \ast \rbrace \times \oc(\lbrace \ast \rbrace) \] 
For more details on the free exponential modality of $(\mathsf{REL}, \times, \lbrace \ast \rbrace)$ see \cite{blute2006differential, hyland2003glueing}. 
\item Let $\mathbb{K}$ be a field. Then $(\mathsf{VEC}_\mathbb{K}, \otimes, \mathbb{K})$ is a Lafont category where full details of the free exponential modality and discussions on  cofree cocommutative $\mathbb{K}$-coalgebras can be found in \cite{hyland2003glueing, murfet2015sweedler, sweedler1969hopf}. In \cite{murfet2015sweedler}, Murfet provides a nice expression of cofree cocommutative $\mathbb{K}$-coalgebras when $\mathbb{K}$ is algebraically closed and of characteristic $0$, say for example the complex numbers $\mathbb{C}$. As a simple example in that case, the cofree cocommutative $\mathbb{K}$-coalgebra over the field $\mathbb{K}$ itself is isomorphic as $\mathbb{K}$-coalgebra to:
\[\oc(\mathbb{K}) \cong \bigoplus\limits_{k \in \mathbb{K}} \mathbb{K}[x]\]
where $\mathbb{K}[x]$ is the polynomial ring in one variable. 
\end{enumerate}
\end{exas}

\section{Symmetric Monoidal Mixed Distributive Laws} \label{SMCmixedsec}

Distributive laws between monads, which are natural transformations satisfying certain coherences with the monad structures, were introduced by Beck \cite{beck1969distributive} in order to both compose monads and lift one monad to the other's Eilenberg-Moore category. By lifting we mean that the forgetful functor from the Eilenberg-Moore category to base category preserves the monad strictly. In fact, there is a bijective correspondence between distributive laws between monads and lifting of monads. From a higher category theory perspective, a distributive law is a monad on the $2$-category of monads of a $2$-category \cite{STREET1972149}. There are also several other notions of distributive laws involving monads and bijective correspondence with certain liftings \cite{wisbauer2008algebras}. Of interest to this paper are mixed distributive laws of monads over comonads \cite{beck1969distributive,harmer2007categorical} and in particular mixed distributive laws of symmetric comonoidal monads over symmetric monoidal comonads. This latter kind of mixed distributive law will be in used in Section \ref{coalgmixsec} to lift coalgebra modalities and linear category structure, which is the main goal of this paper. For a more detailed introduction on distributive laws and liftings see \cite{wisbauer2008algebras}. 

\subsection{Mixed Distributive Laws}

\begin{defi}\label{mixeddef} Let $(\mathsf{T}, \mu, \eta)$ be a monad and $(\oc, \delta, \varepsilon)$ a comonad on the same category $\mathbb{X}$. A \textbf{mixed distributive law} \cite{harmer2007categorical,wisbauer2008algebras} of $(\mathsf{T}, \mu, \eta)$ over $(\oc, \delta, \varepsilon)$ is a natural transformation ${\lambda_A: \mathsf{T}\oc(A) \to \oc \mathsf{T}(A)}$ such that the following diagrams commute: 
\begin{equation}\label{mixeddist2}\begin{gathered}\xymatrixcolsep{4pc}\xymatrix{\mathsf{T}\mathsf{T}\oc(A) \ar[r]^-{\mathsf{T}(\lambda_A)} \ar[d]_-{\mu_{\oc(A)}} & \mathsf{T}\oc\mathsf{T}(A)  \ar[r]^-{\lambda_{\mathsf{T}(A)}} & \oc\mathsf{T}\mathsf{T}(A) \ar[d]^-{\oc(\mu_A)} & \oc(A) \ar[r]^-{\eta_{\oc(A)}}  \ar[dr]_-{\oc(\eta_A)} & \mathsf{T}\oc(A)   \ar[d]^-{\lambda_A}  \\
\mathsf{T} \oc(A)   \ar[rr]_-{\lambda_A} & & \oc\mathsf{T}(A) && \oc \mathsf{T}(A)} \end  {gathered}\end  {equation}
  \begin{equation}\label{mixeddist1}\begin{gathered}\xymatrixcolsep{4pc}\xymatrix{\mathsf{T}\oc(A) \ar[r]^-{\mathsf{T}(\delta_A)} \ar[d]_-{\lambda_A}  & \mathsf{T}\oc \oc(A)  \ar[r]^-{\lambda_{\oc(A)}} & \oc\mathsf{T}\oc(A) \ar[d]^-{\oc(\lambda_A)} & \mathsf{T}\oc(A)  \ar[dr]_-{\mathsf{T}(\varepsilon_A)}  \ar[r]^-{\lambda_A} & \oc\mathsf{T}(A)  \ar[d]^-{\varepsilon_{\mathsf{T}(A)}}  \\
  \oc\mathsf{T}(A) \ar[rr]_-{\delta_{\mathsf{T}(A)}} && \oc\oc\mathsf{T}(A)  && \mathsf{T}(A)
  } \end  {gathered}\end  {equation}
\end{defi}

As with distributive laws between monads, mixed distributive laws allow one to lift comonads to Eilenberg-Moore categories of monads and also to lift monads to Eilenberg-Moore categories of comonads, such that the respective forgetful functors preserve the monads or comonads strictly. In fact, mixed distributive laws are in bijective correspondence with these liftings. 

\begin{thmC}[{\cite[Theorem IV.1]{van1971sheaves}}]\label{monadcomonad}  Let $(\mathsf{T}, \mu, \eta)$ be a monad and $(\oc, \delta, \varepsilon)$ be a comonad on the same category $\mathbb{X}$. Then the following are in bijective correspondence:
\begin{enumerate}[{\em (i)}]
\item Mixed distributive laws of $(\mathsf{T}, \mu, \eta)$ over $(\oc, \delta, \varepsilon)$;
\item Liftings of the comonad $(\oc, \delta, \varepsilon)$ to $\mathbb{X}^\mathsf{T}$, that is, a comonad $(\tilde{\oc}, \tilde{\delta}, \tilde{\varepsilon})$ on $\mathbb{X}^\mathsf{T}$ such that the forgetful functor preserves the comonad strictly, that is, the following diagram commutes:
  \[  \xymatrixcolsep{5pc}\xymatrix{\mathbb{X}^\mathsf{T} \ar[d]_-{\mathsf{U}^\mathsf{T}}  \ar[r]^-{\tilde \oc} & \mathbb{X}^\mathsf{T} \ar[d]^-{\mathsf{U}^\mathsf{T}} &  \\
  \mathbb{X} \ar[r]_-{\oc}  & \mathbb{X}
  } \]
and for each $\mathsf{T}$-algebra $(A, \nu)$, $\mathsf{U}^\mathsf{T}(\tilde{\delta}_{(A,\nu)})=\delta_A$ and $\mathsf{U}^\mathsf{T}(\tilde {\varepsilon}_{(A,\nu)})=\varepsilon_A$. 
\item Liftings of the monad $(\mathsf{T}, \mu, \eta)$ to $\mathbb{X}^\oc$, that is, a monad $(\tilde{\mathsf{T}}, \tilde{\mu_A}, \tilde{\eta_A})$ on $\mathbb{X}^\oc$ such that the forgetful functor $\mathsf{U}^\oc$ preserves the comonad strictly, that is, the following diagram commutes:
  \[  \xymatrixcolsep{5pc}\xymatrix{\mathbb{X}^\oc \ar[d]_-{\mathsf{U}^\oc}  \ar[r]^-{\tilde{\mathsf{T}}} & \mathbb{X}^\oc \ar[d]^-{\mathsf{U}^\oc} &  \\
  \mathbb{X} \ar[r]_-{\mathsf{T}}  & \mathbb{X}
  } \]
and for each $\oc$-coalgebra $(A, \omega)$, $\mathsf{U}^\oc(\tilde{\mu}_{(A,\omega)})=\mu_A$ and $\mathsf{U}^\oc(\tilde {\eta}_{(A,\omega)})=\eta_A$. 
\end{enumerate}
\end{thmC} 
We briefly review how to construct mixed distributive laws from liftings and vice-versa (for more details see \cite{wisbauer2008algebras}). Let $\lambda_A$ be a mixed distributive law of a monad $(\mathsf{T}, \mu, \eta)$ over a comonad $(\oc, \delta, \varepsilon)$ over the same category $\mathbb{X}$. For a $\mathsf{T}$-algebra $(A, \nu)$, define the map $\nu^\sharp: \mathsf{T} \oc(A) \to \oc(A)$ as follows: 
\begin{equation}\label{nusharp}\begin{gathered}
\nu^\sharp := \xymatrixcolsep{5pc}\xymatrix{\mathsf{T} \oc(A) \ar[r]^-{\lambda_A} & \oc \mathsf{T}(A) \ar[r]^-{\oc(\nu)} & \oc(A)
  } \end  {gathered}\end  {equation}
 Then the pair $(\oc(A), \nu^\sharp)$ is a $\mathsf{T}$-algebra. Dually, if $(A, \omega)$ is a $\oc$-coalgebra, then define the map $\omega^\flat: \mathsf{T}  A \to \oc \mathsf{T}(A)$ as follows:
 \begin{equation}\label{omegaflat}\begin{gathered}
\omega^\flat := \xymatrixcolsep{5pc}\xymatrix{\mathsf{T}  A \ar[r]^-{\mathsf{T}(\omega)} & \mathsf{T}  \oc(A) \ar[r]^-{\lambda_A} &\oc \mathsf{T}  A
  } \end{gathered}\end{equation}
Then the pair  $(\mathsf{T}(A), \omega^\flat)$ is a $\oc$-coalgebra. The induced comonad $(\tilde{\oc}, \tilde{\delta}, \tilde{\varepsilon})$ on $\mathbb{X}^\mathsf{T}$ and the induced comonad $(\tilde{\mathsf{T}}, \tilde{\mu}, \tilde{\eta})$ on $\mathbb{X}^\oc$ are defined in the obvious way. 

Conversly, to obtain a mixed distributive law, we can either start with a lifting of a monad to the Eilenberg-Moore category of a comonad, or a lifting of a comonad to the Eilenberg-Moore category of a monad. We review explicitly the latter of these options as this is the construction we will use in many of the upcoming proofs of this paper. Let $(\tilde{\oc}, \tilde{\delta}, \tilde{\varepsilon})$ be a lifting of $(\oc, \delta, \varepsilon)$ to $\mathbb{X}^\mathsf{T}$. This implies that for each free $\mathsf{T}$-algebra $(\mathsf{T}(A), \mu_A)$ we have that $\tilde{\oc}(\mathsf{T}(A), \mu_A) = (\oc \mathsf{T}(A), \mu_A^\sharp)$ for some map $\mu_A^\sharp: \mathsf{T} \oc \mathsf{T}(A) \to \oc \mathsf{T}(A)$. Define the natural transformation $\lambda_A: \mathsf{T} \oc(A) \to \oc \mathsf{T}(A)$ as follows:
\begin{equation}\label{lambdabuild}\begin{gathered} \lambda_A := \xymatrixcolsep{5pc}\xymatrix{ \mathsf{T} \oc(A) \ar[r]^-{\mathsf{T}\oc (\eta_A)} & \mathsf{T} \oc \mathsf{T}(A) \ar[r]^-{\mu_A^\sharp} & \oc \mathsf{T}(A)
  } \end  {gathered}\end  {equation}
Then $\lambda_A$ is a mixed distributive law of $(\mathsf{T}, \mu, \eta)$ over $(\oc, \delta, \varepsilon)$.

\subsection{Symmetric Monoidal Mixed Distributive Laws}

We now extend the definition of a mixed distributive law to a special mixed distributive law of symmetric comonoidal monads over symmetric monoidal comonads such that the liftings are again symmetric (co)monoidal. At this point, the author would like to thank the reviewers for their suggestion of expanding out this part of the story, which was not done in the extended abstract version of this paper \cite{lemay2018lifting}. 

\begin{defi}\label{mixedmondef} Let $(\mathsf{T}, \mu, \eta, \mathsf{n}, \mathsf{n}_{K})$ be a symmetric comonoidal monad and $(\oc, \delta, \varepsilon, \mathsf{m}, \mathsf{m}_K)$ a symmetric monoidal comonad on the same symmetric monoidal category $(\mathbb{X}, \otimes, K)$. A \textbf{symmetric monoidal mixed distributive law} of $(\mathsf{T}, \mu, \eta, \mathsf{n}, \mathsf{n}_{K})$ over $(\oc, \delta, \varepsilon, \mathsf{m}, \mathsf{m}_K)$ is a mixed distributive law $\lambda$ of $(\mathsf{T}, \mu, \eta)$ over $(\oc, \delta, \varepsilon)$ such that the following diagrams commute: 
\begin{equation}\begin{gathered}\label{mixedsymco}   \xymatrixcolsep{5pc}\xymatrix{\mathsf{T}(\oc(A) \otimes \oc(B)) \ar[d]_-{\mathsf{T}(\mathsf{m}_{A,B})} \ar[r]^-{\mathsf{n}_{\oc(A),\oc(B)}} & \mathsf{T}\oc(A) \otimes \mathsf{T} \oc(B) \ar[r]^-{\lambda_A \otimes \lambda_B} & \oc \mathsf{T}(A) \otimes \oc\mathsf{T}(B)\ar[d]^-{\mathsf{m}_{\mathsf{T}(A),\mathsf{T}(B)}} \\
 \mathsf{T}\oc(A \otimes B) \ar[r]_-{\lambda_{A \otimes B}} & \oc\mathsf{T}(A \otimes B) \ar[r]_-{\oc(\mathsf{n}_{A,B})} & \oc(\mathsf{T}(A) \otimes \mathsf{T} B)   } \\
    \xymatrixcolsep{5pc}\xymatrix{ \mathsf{T}(K) \ar[d]_-{\mathsf{T}(\mathsf{m}_K)} \ar[rr]^-{\mathsf{n}_{K}} & & K \ar[d]^-{\mathsf{m}_K} \\
 \mathsf{T} \oc(K)  \ar[r]_-{\lambda_K} & \oc \mathsf{T}(K) \ar[r]_-{\oc(\mathsf{n}_{K})} & \oc(K)
  } \end{gathered}\end{equation}
\end{defi}

As we will see in the Proposition \ref{liftsymmix}, the coherences of a symmetric monoidal mixed distributive law are precisely the requirements needed in order to lift symmetric monoidal comonads to the Eilenberg-Moore categories of symmetric comonoidal monads, and vice-versa. In particular, the left diagram of (\ref{mixedsymco}) implies that $\mathsf{m}$ is a $\mathsf{T}$-algebra morphism and similarly that $\mathsf{n}$ is a $\oc$-coalgebra morphism. While the right diagram of (\ref{mixedsymco}) implies that $(K, \mathsf{m}_K)$ is a $\mathsf{T}$-algebra and that $(K, \mathsf{n}_K)$ is a $\oc$-coalgebra. 

\begin{prop}\label{liftsymmix} Let $(\mathsf{T}, \mu, \eta, \mathsf{n}, \mathsf{n}_{K})$ be a symmetric comonoidal monad and $(\oc, \delta, \varepsilon, \mathsf{m}, \mathsf{m}_K)$ a symmetric monoidal comonad on the same symmetric monoidal category $(\mathbb{X}, \otimes, K)$. Then the following are in bijective correspondence:
\begin{enumerate}
\item Symmetric monoidal mixed distributive laws of $(\mathsf{T}, \mu, \eta, \mathsf{n}, \mathsf{n}_{K})$ over $(\oc, \delta, \varepsilon, \mathsf{m},\mathsf{m}_K)$;
\item Liftings of the symmetric monoidal comonad $(\oc, \delta, \varepsilon, \mathsf{m}, \mathsf{m}_K)$ to $(\mathbb{X}^\mathsf{T}, \otimes^{\mathsf{n}}, (K,\mathsf{n}_K))$, that is, a symmetric monoidal comonad $(\tilde{\oc}, \tilde{\delta}, \tilde{\varepsilon},  \tilde{\mathsf{m}}, \tilde{\mathsf{m}}_{(K,\mathsf{n}_K)})$ on $(\mathbb{X}^\mathsf{T}, \otimes^{\mathsf{n}}, (K,\mathsf{n}_K))$ which is a lifting of the underlying comonad to $\mathbb{X}^\mathsf{T}$ (in the sense of Theorem \ref{monadcomonad}) such that for each pair of $\mathsf{T}$-algebras $(A, \nu)$ and $(B, \nu^\prime)$, $\mathsf{U}^\mathsf{T}(\tilde{\mathsf{m}}_{(A, \nu),(B,\nu^\prime)})=\mathsf{m}_{A,B}$, and also that $\mathsf{U}^\mathsf{T}(\tilde{\mathsf{m}}_{(K,\mathsf{n}_K)})=\mathsf{m}_K$. 
\item Liftings of the symmetric comonoidal monad $(\mathsf{T}, \mu, \eta, \mathsf{n}, \mathsf{n}_{K})$ to $(\mathbb{X}^\oc, \otimes^\mathsf{m}, (K, \mathsf{m}_K))$, that is, a symmetric comonoidal monad $(\tilde{\mathsf{T}}, \tilde{\mu}, \tilde{\eta}, \tilde{\mathsf{n}}, \tilde{\mathsf{n}}_{(K,\mathsf{m}_K)})$ on $(\mathbb{X}^\oc, \otimes^\mathsf{m}, (K, \mathsf{m}_K))$ which is a lifting of the underlying monad to $\mathbb{X}^\oc$ (in the sense of Theorem \ref{monadcomonad}) such that for each pair of $\oc$-algebras $(A, \omega)$ and $(B, \omega^\prime)$, $\mathsf{U}^\oc(\tilde{\mathsf{n}}_{(A, \omega),(B,\omega^\prime)})=\mathsf{n}_{A,B}$, and also that $\mathsf{U}^\mathsf{T}(\tilde{\mathsf{n}}_{(K,\mathsf{m}_K)})=\mathsf{n}_K$. 
\end{enumerate}
\end{prop}
\begin{proof} The bijective correspondence will follow from Theorem \ref{monadcomonad}. It remains to show that the induced lifting of the comonad (resp. monad) from the mixed distributive law is also a lifting of the symmetric monoidal comonad (resp. symmetric comonoidal monad), and similarly that the induced mixed distributive law from the liftings also satisfies (\ref{mixedsymco}). 

Let $\lambda$ be a symmetric mixed distributive law of $(\mathsf{T}, \mu, \eta, \mathsf{n}, \mathsf{n}_{K})$ over $(\oc, \delta, \varepsilon, \mathsf{m},\mathsf{m}_K)$. Consider first the induced lifting of $(\oc, \delta, \varepsilon)$ from Theorem \ref{monadcomonad}. To prove that we have a lifting of the symmetric monoidal comonad, it suffices to show that $\mathsf{m}$ and $\mathsf{m}_K$ are $\mathsf{T}$-algebra morphisms. The right diagram of (\ref{mixedsymco}) is precisely the statement that $\mathsf{m}_K$ is a $\mathsf{T}$-algebra morphism. Then if $(A,\nu)$ and $(B, \nu^\prime)$ are $\mathsf{T}$-algebras, commutativity of the following diagrams show that $\mathsf{m}_{A,B}$ is a $\mathsf{T}$-algebra morphism: 
  \[ \xymatrixcolsep{4pc}\xymatrix{\mathsf{T}(\oc(A) \otimes \oc(B))  \ar@{}[drr]|{(\ref{mixedsymco})} \ar[d]_-{\mathsf{T}(\mathsf{m}_{A,B})} \ar[r]^-{\mathsf{n}_{A,B}} & \mathsf{T}\oc(A) \otimes \mathsf{T} \oc(B) \ar[r]^-{\lambda_A \otimes \lambda_B} & \oc \mathsf{T}(A) \otimes \oc\mathsf{T}(B)\ar@{}[dr]|{\text{Nat. of $\mathsf{m}$}} \ar[d]_-{\mathsf{m}_{\mathsf{T}(A),\mathsf{T}(B)}} \ar[r]^-{\oc(\nu) \otimes \oc(\nu^\prime)} & \oc(A) \otimes \oc(B) \ar[d]^-{\mathsf{m}_{A,B}}  \\
 \mathsf{T}\oc(A \otimes B) \ar[r]_-{\lambda_{A \otimes B}} & \oc\mathsf{T}(A \otimes B) \ar[r]_-{\oc(\mathsf{n}_{A,B})} & \oc(\mathsf{T}(A) \otimes \mathsf{T} B)  \ar[r]_-{\oc(\nu \otimes \nu^\prime)} &  \oc (A \otimes B)  
  } \]
 Similarly, one can show that we also obtain a lifting of the symmetric comonoidal monad by proving that $\mathsf{n}$ and $\mathsf{n}_K$ are $\oc$-coalgebra morphisms. 
 
Conversly, let $(\tilde{\oc}, \tilde{\delta}, \tilde{\varepsilon},  \tilde{\mathsf{m}}, \tilde{\mathsf{m}}_K)$ be a lifting of $(\oc, \delta, \varepsilon, \mathsf{m},\mathsf{m}_K)$ to $\mathbb{X}^\mathsf{T}$. This implies that $\mathsf{m}$ and $\mathsf{m}_K$ are $\mathsf{T}$-algebra morphisms, and in particular for free $\mathsf{T}$-algebras $(\mathsf{T}(A), \mu_A)$ and the $\mathsf{T}$-algebra $(K, \mathsf{n}_{K})$, we have that the following diagrams commute:
 \begin{equation}\label{mtatalg} \begin{gathered} \xymatrixcolsep{5pc}\xymatrix{\mathsf{T}(\oc \mathsf{T}(A) \otimes \oc \mathsf{T} B) \ar[d]_-{\mathsf{T}(\mathsf{m}_{\mathsf{T}(A),\mathsf{T}(B)})} \ar[r]^-{\mathsf{n}_{\oc\mathsf{T}(A),\oc\mathsf{T}(B)}} & \mathsf{T}\oc \mathsf{T}(A) \otimes \mathsf{T}\oc \mathsf{T}(B) \ar[r]^-{\mu_A^\sharp \otimes \mu_B^\sharp} & \oc \mathsf{T}(A) \otimes \oc\mathsf{T}(B)\ar[d]^-{\mathsf{m}_{\mathsf{T}(A),\mathsf{T}(B)}} \\
 \mathsf{T}\oc(\mathsf{T}(A) \otimes \mathsf{T} B) \ar[rr]_-{(\mu_A \otimes^{\mathsf{n}} \mu_A)^\sharp} && \oc(\mathsf{T}(A) \otimes \mathsf{T} B) 
  } \\
  \xymatrixcolsep{5pc}\xymatrix{\mathsf{T}(K) \ar[d]_-{\mathsf{T}(\mathsf{m}_K)} \ar[r]^-{\mathsf{n}_{K}} & K \ar[d]^-{\mathsf{m}_K} \\
\mathsf{T} \oc(K)  \ar[r]_-{\mathsf{n}_{K}^\sharp} & \oc(K)
  } \end{gathered}\end{equation}
    where recall for a $\mathsf{T}$-algebra $(A, \nu)$, the map $\nu^\sharp$ is the induced $\mathsf{T}$-algebra on $\tilde \oc(A, \nu)=(\oc(A), \nu^\sharp)$, and $\mu_A \otimes^{\mathsf{n}} \mu_A$ is defined as in (\ref{eilentensor}). Notice that since both $\mathsf{n}_{A,B}$ and $\mathsf{n}_{K}$ are $\mathsf{T}$-algebra morphisms, the lifting implies that $\oc(\mathsf{n}_{A,B})$ and $\oc(\mathsf{n}_{K})$ are also, that is, the following diagrams commute:
     \begin{equation}\label{n2talg} \begin{gathered} \xymatrixcolsep{4pc}\xymatrix{\mathsf{T}\oc\mathsf{T}(A \otimes B)   \ar[d]_-{\mathsf{T}\oc(\mathsf{n}_{A,B})} \ar[r]^-{\mu_{A\otimes B}^\sharp} &  \oc\mathsf{T}(A \otimes B) \ar[d]^-{\oc(\mathsf{n}_{A,B})} & \mathsf{T} \oc \mathsf{T}(K)  \ar[d]_-{\mathsf{T} \oc(\mathsf{n}_{K})} \ar[r]^-{\mu_K^\sharp} & \oc \mathsf{T}(K) \ar[d]^-{\oc(\mathsf{n}_{K})}\\
     \mathsf{T}\oc(\mathsf{T}(A) \otimes \mathsf{T} B)  \ar[r]_-{(\mu_A \otimes^{\mathsf{n}} \mu_B)^\sharp} & \oc(\mathsf{T}(A) \otimes \mathsf{T} B) & \mathsf{T} \oc(K) \ar[r]_-{\mathsf{n}_{K}^\sharp} & \oc(K)
  } \end{gathered}\end{equation}
Consider now the induced mixed distributive law $\lambda$ of $(\mathsf{T}, \mu, \eta)$ over $(\oc, \delta, \varepsilon)$ as defined in (\ref{lambdabuild}). Then that $\lambda$ is a symmetric monoidal mixed distributive law of $(\mathsf{T}, \mu, \eta, \mathsf{n}, \mathsf{n}_{K})$ over $(\oc, \delta, \varepsilon, \mathsf{m},\mathsf{m}_K)$ follows from commutativity of the following diagrams: 
    \[  \xymatrixcolsep{5pc}\xymatrixrowsep{5pc}\xymatrix{ \mathsf{T}(\oc(A) \otimes \oc(B)) \ar[r]^-{\mathsf{n}_{\oc(A),\oc(B)}}  \ar[d]_-{\mathsf{T}(\mathsf{m}_{A,B})} \ar[dr]|-{\mathsf{T}(\oc(\eta_A) \otimes \oc(\eta_B))} & \mathsf{T}\oc(A) \otimes \mathsf{T}\oc(B) \ar@{}[d]|-{\text{Nat. of $\mathsf{n}$}}  \ar[r]^-{\mathsf{T}\oc(\eta_A) \otimes \mathsf{T}\oc(\eta_B)}   & \mathsf{T}\oc \mathsf{T}(A) \otimes \mathsf{T}\oc \mathsf{T}(B)  \ar[d]^-{\mu_A^\sharp \otimes \mu_B^\sharp} \\
  \mathsf{T}\oc(A \otimes B) \ar@{}[r]|-{\text{Nat. of $\mathsf{m}$}}  \ar[dr]|-{\mathsf{T}\oc(\eta_A \otimes \eta_B)} \ar[dd]_-{\mathsf{T}\oc(\eta_{A\otimes B})}    & \mathsf{T}(\oc \mathsf{T}(A) \otimes \oc \mathsf{T}(B)) \ar[d]^-{\mathsf{T}(\mathsf{m}_{\mathsf{T}(A),\mathsf{T}(B)})} \ar[ur]|-{\mathsf{n}_{\oc\mathsf{T}(A),\oc\mathsf{T}(B)}}  &  \oc \mathsf{T}(A) \otimes \oc\mathsf{T}(B)\ar[dd]^-{\mathsf{m}_{\mathsf{T}(A),\mathsf{T}(B)}} \\ 
  &\mathsf{T}\oc(\mathsf{T}(A) \otimes \mathsf{T}(B)) \ar@{}[d]|-{(\ref{n2talg})}  \ar@{}[l]|-{(\ref{symbimonad})} \ar[dr]|-{(\mu_A \otimes^{\mathsf{n}} \mu_B)^\sharp}  \ar@{}[ur]|-{(\ref{symbimonad})} & \\
 \mathsf{T}\oc\mathsf{T}(A \otimes B)  \ar[ur]|-{\mathsf{T}\oc(\mathsf{n}_{A,B})}  \ar[r]_-{\mu_{A \otimes B}^\sharp} & \oc\mathsf{T}(A \otimes B) \ar[r]_-{\oc(\mathsf{n}_{A,B})} & \oc(\mathsf{T}(A) \otimes \mathsf{T} B)
  } \]
\[\xymatrixcolsep{5pc}\xymatrix{\mathsf{T}(K) \ar[dd]_-{\mathsf{T}(\mathsf{m}_K)} \ar[rrr]^-{\mathsf{n}_{K}} &&& K \ar[dd]^-{\mathsf{m}_K} \\ 
&\mathsf{T} \oc(K)  \ar@{}[ur]|-{(\ref{mtatalg})} \ar@/^1pc/[drr]^-{\mathsf{n}_{K}^\sharp} & \\
\mathsf{T}\oc(K) \ar@{=}[ur]^-{}\ar[r]_-{\mathsf{T}\oc (\eta_K)} & \mathsf{T}\oc\mathsf{T}(K) \ar@{}[ul]|(0.35){(\ref{symbimonad})}  \ar@{}[ur]|-{(\ref{n2talg})}  \ar[u]_-{\mathsf{T}\oc(\mathsf{n}_{K})}\ar[r]_-{\mu_K^\sharp} & \oc \mathsf{T}(K) \ar[r]_-{\oc(\mathsf{n}_{K})} & \oc(K)
  }\] 
One could have also started with the lifting of the symmetric comonoidal monad to obtain the desired symmetric monoidal mixed distributive law.   
\end{proof} 

As the main goal of this paper is to lift linear category structure to categories of modules over monoids, we now study symmetric monoidal mixed distributive laws in the case when the symmetric comonoidal monad is induced by a cocommutative bimonoid (Example \ref{bimonadex}). 

\begin{prop} \label{symmonresult} Let $(\oc, \delta, \varepsilon, \mathsf{m}, \mathsf{m}_K)$ be a symmetric monoidal comonad on a symmetric monoidal category $(\mathbb{X}, \otimes, K)$. Then the following are in bijective correspondence:
\begin{enumerate}
\item Cocommutative bimonoids in $(\mathbb{X}^\oc, \otimes^\mathsf{m}, (K, \mathsf{m}_K))$;
\item Cocommutative bimonoids $(A, \nabla, \mathsf{u}, \Delta, \mathsf{e})$ in $(\mathbb{X}, \otimes, K)$ equipped with a natural transformation $\lambda_X: A \otimes \oc(X) \to \oc(A \otimes X)$ such that $\lambda$ is a symmetric monoidal mixed distributive law of $(A \otimes -, \mu^\nabla, \eta^\mathsf{u}, \mathsf{n}^\Delta, \mathsf{n}^\mathsf{e}_K)$ (as defined in Example \ref{bimonadex}) over $(\oc, \delta, \varepsilon, \mathsf{m}, \mathsf{m}_K)$ and the following diagram commutes: 
 \begin{equation}\label{nablastrong}\begin{gathered}\xymatrixcolsep{4.75pc}\xymatrix{A \otimes \left(\oc(X) \otimes \oc(Y) \right)\ar[r]^-{\alpha_{A,\oc(X),\oc(Y)}} \ar[d]_-{1_A \otimes \mathsf{m}_{X,Y}} & \left(A \otimes \oc(X)\right) \otimes \oc(Y)  \ar[r]^-{\lambda_X \otimes 1_{\oc(Y)}}   & \oc(A \otimes X) \otimes \oc(Y) \ar[d]^-{\mathsf{m}_{A \otimes X,Y}} \\
A \otimes \oc (X \otimes Y) \ar[r]_-{\lambda_{X \otimes Y}} &  \oc\left(A \otimes (X \otimes Y) \right) \ar[r]_-{\oc(\alpha_{A,X,Y})} & \oc\left((A \otimes X) \otimes Y\right) 
  } \end{gathered}\end{equation}
\end{enumerate}
\end{prop} 

\begin{proof}We first show that we can construct one from the other. 
  
  \noindent $(1)\Rightarrow (2)$: Let $((A,\omega), \nabla, \mathsf{u}, \Delta, \mathsf{e})$ be a cocommutative bimonoid in $(\mathbb{X}^\oc, \otimes^\mathsf{m}, (K, \mathsf{m}_K))$. Explicitly this means that $(A, \omega)$ is a $\oc$-coalgebra, that $(A, \nabla, \mathsf{u}, \Delta, \mathsf{e})$ is a cocommutative bimonoid in $(\mathbb{X}, \otimes, \mathbb{K})$, and that $\nabla$, $\mathsf{u}$, $\Delta$, and $\mathsf{e}$ are all $\oc$-coalgebra morphisms, that is, the following diagrams commute: 
\begin{equation}\label{nabla!map}\begin{gathered}\xymatrixcolsep{4pc}\xymatrix{A \otimes A \ar[d]_-{\omega \otimes \omega} \ar[rr]^-{\nabla} & & A \ar[d]^-{\omega} & K \ar[d]_-{\mathsf{m}_K} \ar[r]^-{\mathsf{u}} & A \ar[d]^-{\omega} \\
    \oc(A) \otimes \oc(A) \ar[r]_-{\mathsf{m}_{A,A}} & \oc(A \otimes A) \ar[r]_-{\oc(\nabla)} & \oc(A) &   \oc(K) \ar[r]_-{\oc(\mathsf{u})} & \oc(A)
  } \end{gathered}\end{equation}
    \begin{equation}\label{Delta!map2}\begin{gathered}  \xymatrixcolsep{3pc}\xymatrix{A \ar[d]_-{\Delta} \ar[rr]^-{\omega} && \oc(A) \ar[d]^-{\oc(\Delta)} & A \ar[d]_-{\mathsf{e}} \ar[r]^-{\omega} & \oc(A) \ar[d]^-{\oc(\mathsf{e})} \\
    \oc(A) \otimes \oc(A) \ar[r]_-{\omega \otimes \omega} & \oc(A) \otimes \oc(A) \ar[r]_-{\mathsf{m}_{A,A}} & \oc(A \otimes A) & K \ar[r]_-{\mathsf{m}_K} & \oc(K) 
  } \end{gathered}\end{equation}
  Define the natural transformation (natural by construction) $\omega_X^\natural: A \otimes \oc(X)\to \oc(A \otimes X)$ as follows: 
\begin{equation}\label{alphaomega}\begin{gathered}\omega_X^\natural:= \xymatrixcolsep{5pc}\xymatrix{A \otimes \oc(X) \ar[r]^-{\omega \otimes 1_{\oc(X)}} & \oc(A) \otimes \oc(X) \ar[r]^-{\mathsf{m}_{A,X}} & \oc(A \otimes X)  
  } \end{gathered}\end{equation}
First that $\omega^\natural$ is a mixed distributive law of $(A \otimes -, \mu^\nabla, \eta^\mathsf{u})$ over $(\oc, \delta, \varepsilon)$ follows from commutativity of the following diagrams: 
    \[  \xymatrixcolsep{5pc}\xymatrix{ \oc(X) \ar[r]^-{\ell^{-1}_{\oc(X)}}   \ar@/_2pc/[ddr]_-{\oc(\ell^{-1}_X)} & K \otimes \oc(X) \ar[d]_-{\mathsf{m}_K \otimes 1} \ar[r]^-{\mathsf{u} \otimes 1_{\oc(X)}} & A \otimes \oc(X) \ar[d]^-{\omega \otimes 1_{\oc(X)}}   \\
&\oc(K) \otimes \oc(X) \ar@{}[ur]^-{(\ref{nabla!map})} \ar@{}[l]|(0.55){(\ref{smcmonendo})} \ar@{}[dr]|-{\text{Nat. of $\mathsf{m}$}} \ar[d]^-{\mathsf{m}_{K,B}} \ar[r]^-{\oc(\mathsf{u}) \otimes 1_{\oc(X)}}& \oc(A) \otimes \oc(X) \ar[d]^-{\mathsf{m}_{A,X}} \\
   &     \oc(K \otimes X) \ar[r]_-{\oc(\mathsf{u} \otimes 1_X)} & \oc(A \otimes X)  } \]
 \[  \xymatrixcolsep{5pc}\xymatrix{ A \otimes \oc(X) \ar[r]^-{\omega \otimes 1_{\oc(X)}} \ar@/_2pc/[drr]_-{1_A \otimes \varepsilon_X} & \oc(A) \otimes \oc(X)  \ar@{}[d]_-{(\ref{!coalg})}\ar[r]^-{\mathsf{m}_{A,X}} \ar[dr]_-{\varepsilon_A \otimes \varepsilon_X} & \oc(A \otimes X)  \ar@{}[dl]^(0.25){(\ref{symcomonad})}  \ar[d]^-{\varepsilon_{A \otimes X}}   \\
    && A \otimes X 
  } \]
  \[  \xymatrixcolsep{5pc}\xymatrix{A \otimes \oc(X)  \ar@{}[dr]|-{(\ref{!coalg})}  \ar[d]_-{1_A \otimes \delta_X} \ar[r]^-{\omega \otimes 1_{\oc(X)}} & \oc(A) \otimes \oc(X) \ar[d]^-{\delta_A \otimes \delta_X} \ar[r]^-{\mathsf{m}_{A,X}} & \oc(A \otimes X) \ar[ddd]^-{\delta_{A \otimes X}}  \\
  A \otimes \oc\oc(X)  \ar[d]_-{\omega \otimes 1_{\oc\oc(X)}} & \oc\oc(A) \otimes \oc\oc(X) \ar[dd]^-{\mathsf{m}_{\oc(A) \otimes \oc(X)}} \ar@{}[r]|-{(\ref{smcmonendo})} &  \\
  \oc(A) \otimes \oc\oc(X) \ar@{}[r]|-{\text{Nat. of $\mathsf{m}$}} \ar[d]_-{\mathsf{m}_{A, \oc(X)}} \ar[ur]|-{\oc(\omega) \otimes 1_{\oc\oc(X)}} & \\
  \oc(A \otimes \oc(X)) \ar[r]_-{\oc(\omega \otimes 1_{\oc(X)})} & \oc(\oc(A) \otimes \oc(X)) \ar[r]_-{\oc(\mathsf{m}_{A,B})} & \oc\oc(A \otimes X)
  } \]

  \[  \xymatrixcolsep{5pc}\xymatrix{A \otimes (A \otimes \oc(X))  \ar[r]^-{1_A \otimes (\omega \otimes 1_{\oc(X)})}  \ar[d]_-{\alpha_{A,A,\oc(X)}} & A \otimes \left(\oc(A) \otimes \oc(X) \right) \ar[r]^-{1_A \otimes \mathsf{m}_{A,X}} & A \otimes \oc (A \otimes X)   \ar[d]^-{\omega \otimes 1_{\oc(A \otimes X)}} \\
 (A \otimes A) \otimes \oc(X)  \ar@{}[ddr]|-{(\ref{nabla!map})}  \ar[dd]_-{\nabla \otimes 1_{\oc(X)}} \ar[r]^-{(\omega \otimes \omega) \otimes 1_{\oc(X)}} &  \left(\oc(A) \otimes \oc(A) \right) \otimes \oc(X) \ar@{}[r]|-{(\ref{smcmonendo})} \ar[dd]^-{\mathsf{m}_{A,A} \otimes 1_{\oc(X)}} &  \oc(A) \otimes \oc (A \otimes X)  \ar[d]^-{\mathsf{m}_{A,A \otimes X}} \\
  &&  \oc \left(A \otimes (A \otimes X) \right) \ar[d]^-{\oc(\alpha_{A,A,X})} \\
 A \otimes \oc(X)   \ar[d]_-{\omega \otimes 1_{\oc(X)}}  &\oc(A \otimes A) \otimes \oc(X) \ar[r]^-{\mathsf{m}_{A \otimes A, X}} \ar[dl]|-{\oc(\nabla) \otimes 1_{\oc(X)}}  \ar@{}[d]|-{\text{Nat. of $\mathsf{m}$}} &  \oc\left( (A \otimes A) \otimes X \right) \ar[d]^-{\oc(\nabla \otimes 1_X)} \\
 \oc(A) \otimes \oc(X) \ar[rr]_-{\mathsf{m}_{A,X}} &&  \oc(A \otimes X)
  } \]
  Next that $\omega^\natural$ is a symmetric monoidal mixed distributive law follows from commutativity of the following diagrams: 
      \[  \xymatrixcolsep{3pc}\xymatrix{A \otimes K  \ar[d]_-{1_A \otimes\mathsf{m}_K} \ar[rrr]^-{\rho_A} & && A \ar[d]^-{\omega} \ar[r]^-{\mathsf{e}} & K \ar[d]^-{\mathsf{m}_K} \\
       A \otimes \oc(K) \ar@{}[urrr]|-{(\ref{smcmonendo})} \ar[r]_-{\omega \otimes 1_{\oc(K)}} & \oc(A) \otimes \oc(K) \ar[r]_-{\mathsf{m}_{A,K}} & \oc(A \otimes K) \ar[r]_-{\oc(\rho_A)}  & \oc(A) \ar[r]_-{\oc(\mathsf{e})} \ar@{}[ur]|-{(\ref{Delta!map2})} & \oc(K)
  } \]
 {\scriptsize \[\xymatrixcolsep{5pc}\xymatrixrowsep{5pc}\xymatrix{A \otimes \left(\oc(X) \otimes \oc(Y) \right)  \ar@{}[ddr]|-{(\ref{Delta!map2})}  \ar[d]_-{1_A \otimes \mathsf{m}_{X,Y}} \ar[r]^-{\Delta \otimes (1_{\oc(X)} \otimes 1_{\oc(Y)})} & (A \otimes A) \otimes \left( \oc(X) \otimes \oc(Y) \right) \ar[r]^-{\tau_{A,A, \oc(X), \oc(Y)}}  \ar@{}[dr]|-{~~~~\text{Nat. of } \tau} \ar[d]_-{(\omega \otimes \omega) \otimes (1_{\oc(X)} \otimes 1_{\oc(Y)})} & (A \otimes \oc(X)) \otimes (A \otimes \oc(Y)) \ar[d]|-{(\omega \otimes 1_{\oc(X)}) \otimes (\omega \otimes 1_{\oc(Y)})} \\
A \otimes \oc(X \otimes Y) \ar[d]_-{\omega \otimes 1_{\oc(X \otimes Y)}} & \left( \oc(A) \otimes \oc(A) \right) \otimes \left( \oc(X) \otimes \oc(Y) \right) \ar@{}[ddr]|-{(\ref{smcmonendo})}  \ar[d]^-{\mathsf{m}_{A,A} \otimes \mathsf{m}_{X,Y}} \ar[r]_-{\tau_{\oc(A),\oc(A),\oc(X),\oc(Y)}}  & (\oc(A) \otimes \oc(X)) \otimes (\oc(A) \otimes \oc(Y)) \ar[d]|-{\mathsf{m}_{A,X} \otimes \mathsf{m}_{A,Y}} \\
  \oc(A) \otimes \oc(X \otimes Y) \ar@{}[dr]|-{\text{Nat. of $\mathsf{m}$}} \ar[r]^-{\oc(\Delta) \otimes 1_{\oc(X \otimes Y)}} \ar[d]_-{\mathsf{m}_{A, X \otimes Y}}& \oc(A \otimes A)  \otimes \oc(X \otimes Y) \ar[d]^-{\mathsf{m}_{A \otimes A, X \otimes Y}} & \oc(A \otimes X) \otimes \oc(A \otimes Y) \ar[d]|-{\mathsf{m}_{A \otimes X, A \otimes Y}} \\
\oc\left( A \otimes (X \otimes Y)\right) \ar[r]_-{\oc(\Delta  \otimes (1_{X} \otimes 1_{Y})) }  & \oc\left((A \otimes A) \otimes \left( X \otimes Y \right) \right) \ar[r]_-{\oc(\tau_{A,A, X, Y}) } & \oc\left( (A \otimes X) \otimes (A \otimes Y) \right)
  } \]  }%
Finally that $\omega^\natural$ satisfies (\ref{nablastrong}) follows immediately from the symmetric monoidal endofunctor coherences (\ref{smcmonendo}).   
  
   \noindent $(2) \Rightarrow (1)$: Let $(A, \nabla, \mathsf{u}, \Delta, \mathsf{e})$ be a cocommutative bimonoid in $(\mathbb{X}, \otimes, K)$ equipped with a symmetric monoidal mixed distributive law $\lambda$ satisfying (\ref{nablastrong}). Define the map $\lambda^\diamondsuit: A \to \oc(A)$ as follows: 
    \begin{equation}\label{diamond}\begin{gathered}\lambda^\diamondsuit := \xymatrixcolsep{3pc}\xymatrix{ A \ar[r]^-{\rho^{-1}_A} & A \otimes K \ar[r]^-{1_A \otimes \mathsf{m}_K} & A \otimes \oc(K) \ar[r]^-{\lambda_K} & \oc(A \otimes K) \ar[r]^-{\oc(\rho_A)} & \oc(A)
  } \end{gathered}\end{equation}
That $(A, \lambda^\diamondsuit)$ is a $\oc$-coalgebra follows from commutativity of the following diagrams:
    \[  \xymatrixcolsep{4pc}
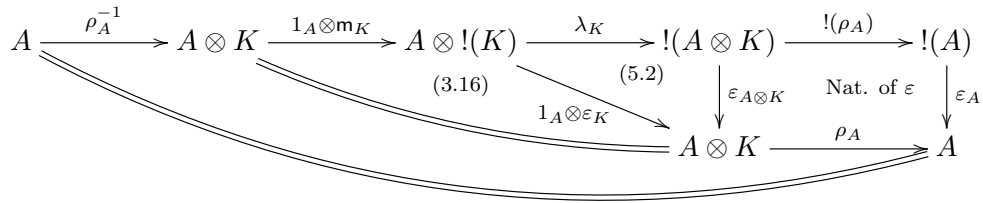
\xymatrix{ A \ar[r]^-{\rho^{-1}_A}  \ar@/_3pc/@{=}[drrrr]^-{} & A \otimes K  \ar@/_1pc/@{=}[drr]^-{} \ar[r]^-{1_A \otimes \mathsf{m}_K} & A \otimes \oc(K) \ar@{}[d]|(0.4){(\ref{symcomonad})}\ar[dr]_-{1_A \otimes \varepsilon_K} \ar[r]^-{\lambda_K} & \oc(A \otimes K) \ar@{}[dl]|(0.3){(\ref{mixeddist1})} \ar@{}[dr]^-{\text{Nat. of $\varepsilon$}~~~~~~~}  \ar[d]^-{\varepsilon_{A \otimes K}} \ar[r]^-{\oc(\rho_A)} & \oc(A) \ar[d]^-{\varepsilon_A} \\
    &&& A \otimes K \ar[r]^-{\rho_A}& A   
  } \] 
  \[  \xymatrixcolsep{3pc}\xymatrix{ A \ar[r]^-{\rho^{-1}_A} \ar[d]_-{\rho^{-1}_A} & A \otimes K  \ar@{}[ddr]|-{(\ref{symcomonad})} \ar[dd]_-{1_A \otimes \mathsf{m}_K} \ar[r]^-{1_A \otimes  \mathsf{m}_K}& A \otimes \oc(K) \ar@{}[ddddr]|-{(\ref{mixeddist1})} \ar[dd]_-{1_A \otimes \delta_K}\ar[r]^-{\lambda_K} & \oc(A \otimes K)  \ar@{}[ddddr]|-{\text{Nat. of $\delta$}}  \ar[dddd]_-{\delta_{A \otimes K}}  \ar[r]^-{\oc(\rho_A)} & \oc(A) \ar[dddd]^-{\delta_A} \\
  A \otimes K \ar[d]_-{1_A \otimes \mathsf{m}_K} &  \\
   A \otimes \oc(K) \ar[d]_-{\lambda_K} &  A \otimes \oc(K)  \ar@{}[ddr]|-{\text{Nat. of $\lambda$}} \ar[dd]_-{\lambda_{K}}\ar[r]_-{1_A \otimes \oc(\mathsf{m}_K)} & A \otimes \oc\oc(K) \ar[dd]_-{\lambda_{\oc(K)}} \\
    \oc(A \otimes K) \ar[d]_-{\oc(\rho_A)} \ar@/^1pc/@{=}[dr]^-{}  \\
     \oc(A)  \ar[r]_-{\oc(\rho^{-1}_A)} & \oc(A \otimes K) \ar[r]_-{\oc(1_A \otimes \mathsf{m}_K)} & \oc\left(A \otimes \oc(K)\right) \ar[r]_-{\oc(\lambda_K)} & \oc\oc(A \otimes K)  \ar[r]_-{\oc\oc(\rho_A)} & \oc\oc(A) 
  } \]    
  Now since $\lambda$ is a symmetric monoidal mixed distributive law of $(A \otimes -, \mu^\nabla, \eta^\mathsf{u}, \mathsf{n}^\Delta, \mathsf{n}^\mathsf{e}_K)$ (as defined in Example \ref{bimonadex}) over $(\oc, \delta, \varepsilon, \mathsf{m}, \mathsf{m}_K)$, by Proposition \ref{liftsymmix} we have obtained a symmetric comonoidal monad $(\widetilde{A \otimes -}, \tilde{\mu^\nabla}, \tilde{\eta^\mathsf{u}}, \tilde{\mathsf{n}^\Delta}, \tilde{\mathsf{n}^\mathsf{e}}_{(K, \mathsf{m}_K)})$ on $(\mathbb{X}^\oc, \otimes^\mathsf{m}, (K, \mathsf{m}_K))$. Recall that on objects the endofunctor $\widetilde{A \otimes -}: \mathbb{X}^\oc \to \mathbb{X}^\oc$ is defined as $(\widetilde{A \otimes -})(X, \omega) := (A \otimes X, \omega^\flat)$, where $\omega^\flat$ is defined as in  (\ref{omegaflat}), that is: 
\[ \omega^\flat:= \xymatrixcolsep{5pc}\xymatrix{A \otimes X \ar[r]^-{1_A \otimes \omega} & A \otimes \oc(X) \ar[r]^-{\lambda_X} & \oc(A \otimes X) } \]
  Now notice that the following diagram also commutes: 
  \[  \xymatrixcolsep{2.95pc}\xymatrix{ A \otimes X \ar[dd]_-{1_A \otimes \omega} \ar[r]^-{\rho^{-1}_A \otimes \omega} & (A \otimes K) \otimes \oc(X) \ar@{}[dd]|-{(\ref{SMCaxioms}) + (\ref{smcmonendo})} \ar[rr]^-{(1_A \otimes \mathsf{m}_K) \otimes 1_{\oc(X)}} && (A \otimes \oc(K)) \otimes \oc(X)  \ar[d]^-{\lambda_K \otimes 1_{\oc(X)}} \ar[dl]|-{\alpha^{-1}_{A, \oc(K), \oc(X)}} \\
  &&A \otimes (\oc(K) \otimes \oc(X)) \ar@{}[dd]|-{(\ref{nablastrong})}  \ar[dl]|-{1_A \otimes \mathsf{m}_{K,X}} &  \oc(A \otimes K) \otimes \oc(X) \ar[dddl]_-{\mathsf{m}_{A \otimes K, X}} \ar[d]^-{\oc(\rho_A) \otimes 1_{\oc(X)}}  \\
 A \otimes \oc(X) \ar[dd]_-{\lambda_X} \ar@{}[ddr]|-{\text{Nat. of $\lambda$}}  \ar[r]^-{1_A \otimes \oc( \ell^{-1}_X)}& A \otimes \oc(K \otimes X)  \ar[dd]^-{\lambda_{K \otimes X}} &   &   \oc(A) \otimes \oc(X) \ar[dd]^-{\mathsf{m}_{A,X}} \ar@{}[ddl]|(0.35){\text{Nat. of $\mathsf{m}$}} \\
&&  \\ 
\oc(A \otimes X) \ar@/_3.5pc/@{=}[rrr]^-{} \ar[r]^-{\oc(1_A \otimes \ell^{-1}_X)} &  \oc(A \otimes (K \otimes X))  \ar@{}[dr]|-{(\ref{SMCaxioms})}  \ar[r]^-{\oc(\alpha_{A,K,X})} & \oc\left( (A \otimes K) \otimes X \right) \ar[r]^-{\oc(\rho_A \otimes 1_X)} & \oc(A \otimes X) \\
& & 
  } \]
The above diagram implies that $(A \otimes X, \omega^\flat)= (A, \lambda^\diamondsuit) \otimes^\mathsf{n} (X, \omega)$ and therefore $\widetilde{A \otimes -} = (A, \lambda^\diamondsuit) \otimes^\mathsf{n} -$. As explained in Example \ref{bimonadex}, a symmetric comonoidal monad structure on ${(A, \lambda^\diamondsuit) \otimes^\mathsf{n} -}$ induces a cocommutative bimonoid structure on $(A, \lambda^\diamondsuit)$. Furthermore, one easily checks that the induced cocommutative bimonoid structure on $(A, \lambda^\diamondsuit)$ from $((A, \lambda^\diamondsuit) \otimes^\mathsf{n} -, \tilde{\mu^\nabla}, \tilde{\eta^\mathsf{u}}, \tilde{\mathsf{n}^\Delta}, \tilde{\mathsf{n}^\mathsf{e}}_{(K, \mathsf{m}_K)})$ will be precisely $\left( (A, \lambda^\diamondsuit),  \nabla, \mathsf{u}, \Delta, \mathsf{e} \right)$. 
  
To prove the bijective correspondence, it suffices to show that these constructions are inverse to each other, that is, $(\omega^\natural)^\diamondsuit = \omega$ and $(\lambda^\diamondsuit)^\natural = \lambda$. Starting with a cocommutative bimonoid $((A,\omega), \nabla, \mathsf{u}, \Delta, \mathsf{e})$ in $(\mathbb{X}^\oc, \otimes^\mathsf{m}, (K, \mathsf{m}_K))$, that ${(\omega^\natural)^\diamondsuit = \omega}$ follows immediately from the symmetric monoidal endofunctor coherences:  
        \[ \xymatrixcolsep{3pc}\xymatrix{ A  \ar@/_3.5pc/[rrrrr]_-{\omega} \ar[r]^-{\rho^{-1}_A} & A \otimes K \ar[r]^-{1_A \otimes \mathsf{m}_K} & A \otimes \oc(K) \ar[r]^-{\omega \otimes 1_{\oc(K)}} & \oc(A) \otimes \oc(K) \ar@{}[d]|(0.5){(\ref{smcmonendo})}\ar[r]^-{\mathsf{m}_{A,K}} & \oc(A \otimes K) \ar[r]^-{\oc(\rho_A)} & \oc(A) \\
        & &&      
  } \]
Conversly, let $(A, \nabla, \mathsf{u}, \Delta, \mathsf{e})$ be a cocommutative bimonoid in $(\mathbb{X}, \otimes, K)$ equipped with a symmetric monoidal mixed distributive law $\lambda$ which satisfies (\ref{nablastrong}). That $(\lambda^\diamondsuit)^\natural = \lambda$ follows from commutativity of the following diagram:
  \[  \xymatrixcolsep{2.25pc}\xymatrix{ A \otimes \oc(X) \ar[ddrr]_-{1_A \otimes \oc(\ell^{-1}_X)}  \ar[dddd]_-{\lambda_X}   \ar[r]^-{\rho_A \otimes 1_{\oc(X)}} & (A \otimes K) \otimes \oc(X) \ar@{}[dr]|-{(\ref{SMCaxioms}) + (\ref{smcmonendo})} \ar[rr]^-{(1_A \otimes \mathsf{m}_K) \otimes 1_{\oc(X)}} && (A \otimes \oc(K)) \otimes \oc(X) \ar[dl]|-{\alpha^{-1}_{A,\oc(K),\oc(X)}}  \ar[d]^-{\lambda_K \otimes 1_{\oc(X)}} \\
  &&A \otimes (\oc(K) \otimes \oc(X)) \ar[d]^-{1_A \otimes \mathsf{m}_{K,X}} & \oc(A \otimes K) \otimes \oc(X) \ar[d]^-{\oc(\rho_A) \otimes 1_{\oc(X)}} \\
  &&A \otimes \oc(K \otimes X) \ar@{}[r]|-{(\ref{nablastrong})} \ar[dd]_-{\lambda_{K \otimes X}}  &  \oc(A) \otimes \oc(X) \ar[dd]^-{\mathsf{m}_{A,X}} \\ 
  &&& \\
 A \otimes \oc(K \otimes X) \ar@/_3.5pc/@{=}[rrr]^-{} \ar[rr]^-{\oc(1_A \otimes \ell^{-1}_X)}   && \oc(A \otimes (K \otimes X)) \ar[r]^-{\oc(1_A \otimes \ell_X)}&  \oc(A \otimes X)  
  } \]
 \end{proof}

We conclude this section with how to extend Proposition \ref{symmonresult} from bimonoids to Hopf monoids, as we will need the latter to lift symmetric monoidal closed structure. 

\begin{prop} \label{symmonresult2} Let $(\oc, \delta, \varepsilon, \mathsf{m}, \mathsf{m}_K)$ be a symmetric monoidal comonad on a symmetric monoidal category $(\mathbb{X}, \otimes, K)$. Then the following are in bijective correspondence:
\begin{enumerate}
\item Cocommutative Hopf monoids in $(\mathbb{X}^\oc, \otimes^\mathsf{m}, (K, \mathsf{m}_K))$;
\item Cocommutative Hopf monoids $(H, \nabla, \mathsf{u}, \Delta, \mathsf{e}, \mathsf{S})$ in $(\mathbb{X}, \otimes, K)$ equipped with a natural transformation $\lambda_X: H \otimes \oc(X) \to \oc(H \otimes X)$ such that $\lambda$ is a symmetric monoidal mixed distributive law of $(H \otimes -, \mu^\nabla, \eta^\mathsf{u}, \mathsf{n}^\Delta, \mathsf{n}^\mathsf{e}_K)$, $\lambda$ satisfies (\ref{nablastrong}) and the following diagram commutes: 
\begin{equation}\label{extraS}\begin{gathered} \xymatrixcolsep{5pc}\xymatrix{ H \otimes \oc(X) \ar[d]_-{\lambda_X} \ar[r]^-{\mathsf{S} \otimes 1_{\oc(X)}} & H \otimes \oc(X) \ar[d]^-{\lambda_X} \\
\oc(H \otimes X) \ar[r]_-{\oc(\mathsf{S} \otimes 1_X)} & \oc(H \otimes X)
  } \end{gathered}\end{equation}
\end{enumerate}
\end{prop} 
\begin{proof} The bijective correspondence will follow immediately from Proposition \ref{symmonresult} and that antipodes are unique for Hopf monoids (Lemma \ref{lemmahopf} (i)). It remains to show that the constructed symmetric mixed distributive law satisfies (\ref{extraS}) and conversly that the antipode is a $\oc$-coalgebra morphism. \\
\noindent $(1) \Rightarrow (2)$: Let $((H, \omega), \nabla, \mathsf{u}, \Delta, \mathsf{e}, \mathsf{S})$ be a cocommutative Hopf monoid in $(\mathbb{X}^\oc, \otimes^{\mathsf{m}}, (K,\mathsf{m}_K))$. In particular $\mathsf{S}: (H, \omega) \to (H, \omega)$ is a $\oc$-coalgebra morphism. Therefore that $\omega^\natural: A \otimes \oc(X) \to \oc(A \otimes X)$ (as defined in (\ref{alphaomega})) satisfies (\ref{extraS}) follows commutativity of the following diagram: 
  \[  \xymatrixcolsep{5pc}\xymatrix{ H \otimes \oc(X) \ar@{}[dr]|-{(\ref{!coalgmap})} \ar[d]_-{\mathsf{S} \otimes 1_{\oc(X)}} \ar[r]^-{\omega \otimes 1_{\oc(X)}} & \oc(H) \otimes \oc(X) \ar@{}[dr]|-{\text{Nat. of $\mathsf{m}$}}  \ar[d]_-{\oc(\mathsf{S}) \otimes 1_{\oc(X)}} \ar[r]^-{\mathsf{m}_{H,X}} & \oc(H \otimes X) \ar[d]^-{\oc(\mathsf{S} \otimes 1_X)} \\
    H \otimes \oc(X) \ar[r]_-{\omega \otimes 1_{\oc(X)}} & \oc(H) \otimes \oc(X) \ar[r]_-{\mathsf{m}_{H,X}} & \oc(H \otimes X)
  } \]
  $(2) \Rightarrow (1)$: Let $(H, \nabla, \mathsf{u}, \Delta, \mathsf{e}, \mathsf{S})$ be a cocommutative Hopf monoid with symmetric monoidal distributive law $\lambda_X: H \otimes \oc(X) \to \oc(H \otimes X)$ satisfying (\ref{nablastrong}) and (\ref{extraS}). That $\mathsf{S}: (H, \lambda^\diamondsuit) \to (H, \lambda^\diamondsuit)$ (as defined in (\ref{diamond})) is a $\oc$-coalgebra morphism follows from commutativity of the following diagram:
   \[\xymatrixcolsep{4.75pc}\xymatrix{ H \ar[d]_-{\mathsf{S}} \ar@{}[dr]|-{\text{Nat. of $\rho^{-1}$}} \ar[r]^-{\rho^{-1}_H} & H \otimes K \ar[d]^-{\mathsf{S} \otimes 1_K} \ar[r]^-{1_H \otimes \mathsf{m}_K} & H \otimes \oc(K) \ar@{}[dr]|-{(\ref{extraS})} \ar[d]^-{\mathsf{S} \otimes 1_{\oc(K)}} \ar[r]^-{\lambda_K} & \oc(H \otimes K) \ar@{}[dr]|-{~~~~~\text{Nat. of $\rho$}} \ar[d]^-{\oc(\mathsf{S} \otimes 1_K)} \ar[r]^-{\oc(\rho_H)} & \oc(H) \ar[d]^-{\oc(\mathsf{S})} \\
   H \ar[r]_-{\rho^{-1}_H} & H \otimes K \ar[r]_-{1_H \otimes \mathsf{m}_K} & H \otimes \oc(K) \ar[r]_-{\lambda_K} & \oc(H \otimes K) \ar[r]_-{\oc(\rho_H)} & \oc(H)
  } \]
\end{proof} 

\section{Mixed Distributive Laws for Coalgebra Modalities}\label{coalgmixsec}

In this section, we lift both coalgebra modalities and monoidal coalgebra modalities to Eilenberg-Moore categories of symmetric comonoidal monads, the latter of which requires symmetric monoidal mixed distributive laws. We also introduce the notions of exponential lifting monads and $\mathsf{MELL}$ lifting monads, where in particular the Eilenberg-Moore category of a $\mathsf{MELL}$ lifting monad is a linear category -- which is the main result of this paper. 

\subsection{Lifting Coalgebra Modalities}

In order to lift coalgebra modalities to Eilenberg-Moore categories of symmetric comonoidal monads, one requires that $\Delta$ and $\mathsf{e}$ be $\mathsf{T}$-algebra morphism. This can equivalently be expressed as needing a mixed distributive law which is also a comonoid morphism since symmetric comonoidal endofunctors preserve cocommutative comonoids \cite{moerdijk2002monads}. Indeed, let $(\mathsf{T}, \mathsf{n}, \mathsf{n}_K)$ be a symmetric comonoidal endofunctor on a symmetric monoidal category $(\mathbb{X}, \otimes, K)$. If $(A, \Delta, \mathsf{e})$ is a cocommutative comonoid in $(\mathbb{X}, \otimes, K)$ then so is the triple $(\mathsf{T}(A), \Delta^\mathsf{T}, \mathsf{e}^\mathsf{T})$ where the comultiplication and counit are defined as follows:
\begin{equation}\label{Tcom}\begin{gathered} \Delta^\mathsf{T} :=  \xymatrixcolsep{5pc}\xymatrix{ \mathsf{T}(A) \ar[r]^-{\mathsf{T}(\Delta)} & \mathsf{T}(A \otimes A) \ar[r]^-{\mathsf{n}_{A,A}} & \mathsf{T}(A) \otimes \mathsf{T}(A) 
  } \\
   \mathsf{e}^\mathsf{T} :=  \xymatrixcolsep{5pc}\xymatrix{ \mathsf{T}(A) \ar[r]^-{\mathsf{T}(\mathsf{e})} & \mathsf{T}(K) \ar[r]^-{\mathsf{n}_{K}} & K 
  } \end{gathered}\end{equation}
Note that if $(\mathbb{X}, \otimes, K)$ is a Cartesian monoidal category, then $(\mathsf{T}(A), \Delta^\mathsf{T}, \mathsf{e}^\mathsf{T})$ must be equal to the unique cocommutative comonoid structure on $\mathsf{T}(A)$ from Lemma \ref{Cartcom}. 

 \begin{defi}\label{coalgmix} Let $(\mathsf{T}, \mu, \eta, \mathsf{n}, \mathsf{n}_{K})$ be a symmetric comonoidal monad and $(\oc, \delta, \varepsilon, \Delta, \mathsf{e})$ be a coalgebra modality on the same symmetric monoidal category $(\mathbb{X}, \otimes, K)$. A \textbf{coalgebra mixed distributive law} of $(\mathsf{T}, \mu, \eta, \mathsf{n}, \mathsf{n}_{K})$ over $(\oc, \delta, \varepsilon, \Delta, \mathsf{e})$ is a mixed distributive law $\lambda$ of $(\mathsf{T}, \mu, \eta)$ over $(\oc, \delta, \varepsilon)$ such that $\lambda_A: (\mathsf{T}\oc(A), \Delta_A^\mathsf{T}, \mathsf{e}_A^\mathsf{T}) \to (\oc\mathsf{T}(A), \Delta_{\mathsf{T}(A)}, \mathsf{e}_{\mathsf{T}(A)})$ is a comonoid morphism, that is, the following diagrams commute: 
\begin{equation}\label{mixeddistcoalg} \begin{gathered}\xymatrixcolsep{4pc}\xymatrix{\mathsf{T}\oc(A) \ar[d]_-{\mathsf{T}(\Delta_A)} \ar[r]^-{\lambda_A}& \oc \mathsf{T}(A) \ar[dd]^-{\Delta_{\mathsf{T}(A)}} & \mathsf{T}\oc(A) \ar[d]_-{\mathsf{T}(\mathsf{e}_A)} \ar[r]^-{\lambda_A} & \oc\mathsf{T}(A) \ar[d]^-{\mathsf{e}_{\mathsf{T}(A)}} \\
    \mathsf{T}(\oc(A) \otimes \oc(A)) \ar[d]_-{\mathsf{n}_{\oc(A),\oc(A)}} && \mathsf{T}K \ar[r]_-{\mathsf{n}_{K}} & K\\
    \mathsf{T}\oc(A) \otimes \mathsf{T}\oc(A) \ar[r]_-{\lambda_A \otimes \lambda_A} & \oc\mathsf{T}(A) \otimes \oc \mathsf{T}(A)
  } \end{gathered}\end{equation}
\end{defi}

We first observe that these mixed distributive laws preserve the induced comonoid structure on $\oc$-coalgebras in the following sense: 

\begin{lem}\label{preserve!com}  Let $(\mathsf{T}, \mu, \eta, \mathsf{n}, \mathsf{n}_{K})$ be a symmetric comonoidal monad and $(\oc, \delta, \varepsilon, \Delta, \mathsf{e})$ be a coalgebra modality on the same symmetric monoidal category $(\mathbb{X}, \otimes, K)$, and let $\lambda$ be a coalgebra mixed distributive law of $(\mathsf{T}, \mu, \eta, \mathsf{n}, \mathsf{n}_{K})$ over $(\oc, \delta, \varepsilon, \Delta, \mathsf{e})$. If $(A, \omega)$ is a $\oc$-coalgebra, then the following diagrams commute:
\begin{equation}\label{mixcompres}\begin{gathered} \xymatrixcolsep{5pc}\xymatrix{\mathsf{T}(A)  \ar[dr]_-{\Delta^{\omega^\flat}} \ar[r]^-{\mathsf{T}(\Delta^\omega)}& \mathsf{T} (A \otimes A) \ar[d]^-{\mathsf{n}_{A,A}} & \mathsf{T}(A) \ar[r]^-{\mathsf{T}(\mathsf{e}^\omega)} \ar[dr]_-{\mathsf{e}^{\omega^\flat}} & \mathsf{T}(K) \ar[d]^-{\mathsf{n}_{K}} \\
  & \mathsf{T}(A) \otimes \mathsf{T}(A) & & K
  }  \end{gathered}\end{equation}
where $\omega^\flat$ is defined as in (\ref{omegaflat}), and both $\Delta^{\omega^\flat}$ and $\mathsf{e}^{\omega^{\flat}}$ are defined as in (\ref{!coalgcom}). 
\end{lem} 
\begin{proof} The lemma follows from commutativity of the following diagrams: 
  \[\xymatrixcolsep{3pc}\xymatrix{\mathsf{T}(A) \ar@{=}[dd]^-{} \ar[r]^-{\mathsf{T}(\omega)} & \mathsf{T} \oc(A)  \ar@{}[ddrr]|-{(\ref{mixeddistcoalg})} \ar@{=}[dd]^-{} \ar[r]^-{\lambda_A} & \oc \mathsf{T}(A) \ar[r]^-{\Delta_{\mathsf{T}(A)} } & \oc \mathsf{T}(A) \otimes \oc \mathsf{T}(A) \ar[r]^-{\varepsilon_{\mathsf{T}(A)} \otimes \varepsilon_{\mathsf{T}(A)} } & \mathsf{T}(A) \otimes \mathsf{T}(A) \ar@{}[dl]^-{(\ref{mixeddist1})}  \ar@{=}[dd]^-{}  \\ 
 &&& \mathsf{T}\oc(A) \otimes \mathsf{T} \oc(A) \ar[u]_-{\lambda_A \otimes \lambda_A}\ar[dr]^-{~~\mathsf{T}(\varepsilon_A) \otimes \mathsf{T}(\varepsilon_A)} \\
  \mathsf{T}(A) \ar[r]_-{\mathsf{T}(\omega)}  & \mathsf{T} \oc(A) \ar[r]_-{\mathsf{T}(\Delta_A)} & \mathsf{T} (\oc(A) \otimes \oc(A)) \ar[ur]^-{\mathsf{n}_{A,A}} \ar[r]_-{\mathsf{T}(\varepsilon_A \otimes \varepsilon_A)} & \mathsf{T} (A \otimes A) \ar@{}[u]|-{\text{Nat. of $\mathsf{n}$}} \ar[r]_-{\mathsf{n}_{A,A}} & \mathsf{T}(A) \otimes \mathsf{T}(A)
  } \]
    \[  \adjustbox{valign=b}{\xymatrixcolsep{5pc}\xymatrix{\mathsf{T}(A)   \ar@{=}[dd]^-{} \ar[r]^-{\mathsf{T}(\omega)} & \mathsf{T} \oc(A) \ar@{}[ddrr]|-{(\ref{mixeddistcoalg})}   \ar@{=}[dd]^-{} \ar[r]^-{\lambda_A} & \oc \mathsf{T}(A) \ar[r]^-{\mathsf{e}_{\mathsf{T}(A)}} & K \ar@{=}[dd]^-{} \\
    \\
 \mathsf{T}(A) \ar[r]_-{\mathsf{T}(\omega)} & \mathsf{T} \oc(A) \ar[r]_-{\mathsf{T}(\mathsf{e}_A)} &\mathsf{T}(K) \ar[r]_-{\mathsf{n}_{K}} & K
}} \tag*{\qedhere}
\]
\end{proof} 

\begin{prop}\label{liftingcoalgmod}  Let $(\mathsf{T}, \mu, \eta, \mathsf{n}, \mathsf{n}_{K})$ be a symmetric comonoidal monad and $(\oc, \delta, \varepsilon, \Delta, \mathsf{e})$ be a coalgebra modality on the same symmetric monoidal category $(\mathbb{X}, \otimes, K)$. Then the following are in bijective correspondence:
\begin{enumerate}
\item Coalgebra mixed distributive laws of $(\mathsf{T}, \mu, \eta, \mathsf{n}, \mathsf{n}_{K})$ over $(\oc, \delta, \varepsilon, \Delta, \mathsf{e})$;
\item Liftings of $(\oc, \delta, \varepsilon, \Delta, \mathsf{e})$ to $(\mathbb{X}^\mathsf{T}, \otimes^{\mathsf{n}}, (K,\mathsf{n}_K))$, that is, a coalgebra modality $(\tilde{\oc}, \tilde{\delta}, \tilde{\varepsilon}, \tilde{\Delta}, \tilde{\mathsf{e}})$ on $(\mathbb{X}^\mathsf{T}, \otimes^{\mathsf{n}}, (K,\mathsf{n}_K))$ which is a lifting of the underlying comonad $(\oc, \delta, \varepsilon)$ to $\mathbb{X}^\mathsf{T}$ (in the sense of Theorem \ref{monadcomonad}) such that for every $\oc$-algebras $(A, \omega)$, $\mathsf{U}^\mathsf{T}(\tilde{\Delta}_{(A,\omega)})=\Delta_A$ and $\mathsf{U}^\mathsf{T}(\tilde{\mathsf{e}}_{(A,\omega)})=\mathsf{e}_A$. 
\end{enumerate}
\end{prop} 
\begin{proof} We take the same approach as in the proof of Proposition \ref{liftsymmix}. The bijective correspondence will follow from Theorem \ref{monadcomonad}, and so it remains to show that we obtain one from the other. 

\noindent $(1) \Rightarrow (2)$: Let $\lambda$ be a coalgebra modality mixed distributive of $(\mathsf{T}, \mu, \eta, \mathsf{n}_{A,B}, \mathsf{n}_{K})$ over $(\oc, \delta, \varepsilon, \Delta, \mathsf{e})$. Consider the induced lifting of $(\oc, \delta, \varepsilon)$ from Theorem \ref{monadcomonad}. To prove that we have a lifting of the coalgebra modality, it suffices to show that $\Delta$ and $\mathsf{e}$ are $\mathsf{T}$-algebra morphisms. If $(A,\nu)$ is a $\mathsf{T}$-algebra, commutativity of the following diagrams shows that $\Delta_A$ and $\mathsf{e}_A$ are $\mathsf{T}$-algebra morphisms: 
  \[ \xymatrixcolsep{4pc}\xymatrix{\mathsf{T}(\oc(A))  \ar@{}[drr]|-{(\ref{mixeddistcoalg})} \ar[d]_-{\mathsf{T}(\Delta_A)} \ar[rr]^-{\lambda_A}&& \oc \mathsf{T}(A) \ar@{}[dr]|-{\text{Nat. of $\Delta$}} \ar[d]^-{\Delta_{\mathsf{T}(A)}} \ar[r]^-{\oc(\nu)} & \oc(A) \ar[d]^-{\Delta_A}\\
    \mathsf{T}(\oc(A) \otimes \oc(A)) \ar[r]_-{\mathsf{n}_{A,A}} & \mathsf{T}\oc(A) \otimes \mathsf{T}\oc(A) \ar[r]_-{\lambda_A \otimes \lambda_A} & \oc\mathsf{T}(A) \otimes \oc \mathsf{T}(A) \ar[r]_-{\oc(\nu) \otimes \oc(\nu)}& \oc(A) \otimes \oc(A)
  } \]
    \[  \xymatrixcolsep{5pc}\xymatrix{ \mathsf{T}\oc(A) \ar@{}[dr]|-{(\ref{mixeddistcoalg})} \ar[d]_-{\mathsf{T}(\mathsf{e}_A)} \ar[r]^-{\lambda_A} & \oc\mathsf{T}(A) \ar@{}[dr]|(0.35){\text{Nat. of $\mathsf{e}$}}  \ar[d]_-{\mathsf{e}_{\mathsf{T}(A)}} \ar[r]^-{\oc(\nu)} & \oc(A) \ar@/^1pc/[dl]^-{\mathsf{e}_A} \\
    \mathsf{T}(K) \ar[r]_-{\mathsf{n}_{K}} & K &
  } \]
  $(2)\Rightarrow (1)$: Let $(\tilde{\oc}, \tilde{\delta}, \tilde{\varepsilon}, \tilde{\Delta}, \tilde{\mathsf{e}})$ be a lifting of $(\oc, \delta, \varepsilon, \Delta, \mathsf{e})$ to $(\mathbb{X}^\mathsf{T}, \otimes^{\mathsf{n}}, (K,\mathsf{n}_K))$. This implies that $\Delta$ and $\mathsf{e}$ are $\mathsf{T}$-algebra morphisms, where in particular for free $\mathsf{T}$-algebras $(\mathsf{T}(A), \mu_A)$, the following diagrams commute:
 \begin{equation}\label{deltatalg} \begin{gathered} \xymatrixcolsep{5pc}\xymatrix{\mathsf{T}\oc \mathsf{T}(A) \ar[d]_-{\mathsf{T}(\Delta_{\mathsf{T}(A)})} \ar[rr]^-{\mu_A^\sharp}&& \oc \mathsf{T}(A) \ar[d]^-{\Delta_{\mathsf{T}(A)}}  \\
    \mathsf{T}(\oc(A) \otimes \oc(A)) \ar[r]_-{\mathsf{n}_{\oc(A),\oc(A)}} &  \mathsf{T}\oc \mathsf{T}(A) \otimes \mathsf{T}\oc \mathsf{T}(A) \ar[r]_-{\mu_A^\sharp \otimes \mu_A^\sharp}  &  \oc\mathsf{T}(A) \otimes \oc \mathsf{T}(A)
  } \\
   \xymatrixcolsep{5pc}\xymatrix{ \mathsf{T}\oc \mathsf{T}(A) \ar[d]_-{\mathsf{T}(\mathsf{e}_{\mathsf{T}(A)})} \ar[r]^-{\mu_A^\sharp} & \oc\mathsf{T}(A) \ar[d]^-{\mathsf{e}_{\mathsf{T}(A)}} \\
 \mathsf{T}(K) \ar[r]_-{\mathsf{n}_{K}} & K
  } \end{gathered}\end{equation}
    where recall $\mu_A^\sharp$ is the $\mathsf{T}$-algebra structure of ~$\tilde{\oc} (\mathsf{T}(A), \mu_A)=(\oc \mathsf{T}(A), \mu_A^\sharp)$. Consider now the induced mixed distributive law $\lambda$ of $(\mathsf{T}, \mu, \eta)$ over $(\oc, \delta, \varepsilon)$ as defined in (\ref{lambdabuild}). Then that $\lambda$ is a comonoid morphism, and therefore also a coalgebra mixed distributive law, follows from commutativity of the following diagrams: 
   \[ \xymatrixcolsep{7pc}\xymatrix{\mathsf{T} \oc(A)  \ar@{}[dr]|{\text{Nat. of $\Delta$}} \ar[d]_-{\mathsf{T}(\Delta_A)}\ar[r]^-{\mathsf{T}\oc(\eta_A)} & \mathsf{T}\oc \mathsf{T}(A) \ar@{}[ddr]|{(\ref{deltatalg})} \ar[d]_-{\mathsf{T}(\Delta_A)} \ar[r]^-{\mu_A^\sharp}& \oc \mathsf{T}(A) \ar[dd]^-{\Delta_{\mathsf{T}(A)}}  \\
   \mathsf{T}(\oc(A) \otimes \oc(A)) \ar@{}[dr]|{\text{Nat. of $\mathsf{n}$}} \ar[d]_-{\mathsf{n}_{\oc(A),\oc(A)}} \ar[r]_-{\mathsf{T}(\oc(\eta_A) \otimes \oc(\eta_A))}&  \mathsf{T}(\oc \mathsf{T}(A) \otimes \oc \mathsf{T}(A)) \ar[d]_-{\mathsf{n}_{\mathsf{}\oc\mathsf{T}(A),\oc\mathsf{T}(A)}} \\
  \mathsf{T} \oc(A) \otimes \mathsf{T} \oc(A)  \ar[r]_-{\mathsf{T}\oc(\eta_A) \otimes \mathsf{T}\oc(\eta_A)}  &  \mathsf{T}\oc \mathsf{T}(A) \otimes \mathsf{T}\oc \mathsf{T}(A) \ar[r]_-{\mu_A^\sharp \otimes \mu_A^\sharp} & \oc\mathsf{T}(A) \otimes \oc \mathsf{T}(A)
  } \]
     \[ \adjustbox{valign=b}{\xymatrixcolsep{5pc}\xymatrix{ \mathsf{T} \oc(A)   \ar[r]^-{\mathsf{T} \oc(\eta_A)} \ar@/_1pc/[dr]_-{\mathsf{T}(\mathsf{e}_A)} & \mathsf{T}\oc\mathsf{T}(A)  \ar@{}[dr]|{(\ref{deltatalg})} \ar@{}[dl]|(0.35){\text{Nat. of $\mathsf{e}$}} \ar[d]^-{\mathsf{T}(\mathsf{e}_{\mathsf{T}(A)})} \ar[r]^-{\mu_A^\sharp} & \oc\mathsf{T}(A) \ar[d]^-{\mathsf{e}_{\mathsf{T}(A)}} \\
  & \mathsf{T}K \ar[r]_-{\mathsf{n}_{K}} & K
  }}\tag*{\qedhere} \]
\end{proof} 

\subsection{Lifting Monoidal Coalgebra Modalities}

We first show that for monoidal coalgebra modalities, every symmetric monoidal mixed distributive law is in fact also a coalgebra mixed distributive law. 

\begin{lem}\label{biglemma} Let $(\mathsf{T}, \mu, \eta, \mathsf{n}, \mathsf{n}_{K})$ be a symmetric comonoidal monad and $(\oc, \delta, \varepsilon, \Delta, \mathsf{e}, \mathsf{m},\mathsf{m}_K)$ be a monoidal coalgebra modality on the same symmetric monoidal category $(\mathbb{X}, \otimes, K)$, and let $\lambda$ be a symmetric monoidal mixed distributive law of $(\mathsf{T}, \mu, \eta, \mathsf{n}, \mathsf{n}_{K})$ over $(\oc, \delta, \varepsilon, \mathsf{m},\mathsf{m}_K)$. Then $\lambda$ is a coalgebra mixed distributive law of $(\mathsf{T}, \mu, \eta, \mathsf{n}, \mathsf{n}_{K})$ over $(\oc, \delta, \varepsilon, \Delta,\mathsf{e})$. 
\end{lem}
\begin{proof} Let $\lambda$ be a symmetric monoidal mixed distributive law of $(\mathsf{T}, \mu, \eta, \mathsf{n}, \mathsf{n}_{K})$ over \\ \noindent $(\oc, \delta, \varepsilon, \mathsf{m},\mathsf{m}_K)$. By Proposition \ref{liftsymmix}, this induces a symmetric comonoidal monad \\ \noindent $(\tilde{\mathsf{T}}, \tilde{\mu}, \tilde{\eta}, \tilde{\mathsf{n}}, \tilde{\mathsf{n}}_{(K,\mathsf{m}_K)})$ on the Eilenberg-Moore category $(\mathbb{X}^\oc, \otimes^\mathsf{m}, (K, \mathsf{m}_K))$, which in particular maps cocommutative comonoids to cocommutative comonoids. Since for every $\oc$-coalgebra $(A, \omega)$, $((A, \omega), \Delta^\omega, \mathsf{e}^\omega)$ is a cocommutative comonoid in $(\mathbb{X}^\oc, \otimes^\mathsf{m}, (K, \mathsf{m}_K))$, then so is $\left((\mathsf{T}(A), \omega^\flat), (\Delta^\omega)^\mathsf{T},  (\mathsf{e}^\omega)^\mathsf{T} \right)$, where $\omega^\flat$ is defined as in  (\ref{omegaflat}), and $(\Delta^\omega)^\mathsf{T}$ and $(\mathsf{e}^\omega)^\mathsf{T}$ are defined as in (\ref{Tcom}). However $(\mathbb{X}^\oc, \otimes^\mathsf{m}, (K, \mathsf{m}_K))$ is a Cartesian monoidal category and therefore by Lemma \ref{Cartcom} we have:
\[(\Delta^\omega)^\mathsf{T} = \Delta^{\omega^\flat} \quad \quad (\mathsf{e}^\omega)^\mathsf{T}= \mathsf{e}^{\omega^\flat}\]
Applying (\ref{omegaflat}) to a cofree $\oc$-coalgebra $(\oc(A), \delta_A)$, we obtain the $\oc$-coalgebra $(\mathsf{T}\oc(A), \delta_A^\flat)$ and that $\lambda_A: (\mathsf{T}\oc(A), \delta_A^\flat) \to (\oc\mathsf{T}(A), \delta_{\mathsf{T}(A)})$ is a $\oc$-coalgebra morphism. Since $(\oc, \delta, \varepsilon, \Delta, \mathsf{e})$ is a coalgebra modality we also have that $\lambda_A: (\mathsf{T}\oc(A), \Delta^{\delta_A^\flat}, \mathsf{e}^{\delta_A^\flat}) \to (\oc\mathsf{T}(A), \Delta_{\mathsf{T}(A)}, \mathsf{e}_{\mathsf{T}(A)})$ is a comonoid morphism (since recall that $\Delta^{\delta_A}= \Delta_A$ and $\mathsf{e}^{\delta_A}=\mathsf{e}_A$). However applying the above identity to to a cofree $\oc$-coalgebra $(\oc(A), \delta_A)$ we get that $\Delta_A^\mathsf{T} = \Delta^{\delta_A^\flat}$ and that $\mathsf{e}_A^\mathsf{T} = \mathsf{e}^{\delta_A^\flat}$. Therefore we conclude that $\lambda_A: (\mathsf{T}\oc(A), \Delta_A^\mathsf{T}, \mathsf{e}_A^\mathsf{T}) \to (\oc\mathsf{T}(A), \Delta_{\mathsf{T}(A)}, \mathsf{e}_{\mathsf{T}(A)})$ is a comonoid morphism and that $\lambda$ is a coalgebra mixed distributive law. 
\end{proof} 

Gathering all of our results together, we can show that symmetric monoidal mixed distributive laws are precisely what is needed to lift monoidal coalgebra modalities. 

\begin{prop}\label{liftmoncoalg}  Let $(\mathsf{T}, \mu, \eta, \mathsf{n}, \mathsf{n}_{K})$ be a symmetric comonoidal monad and \\ \noindent $(\oc, \delta, \varepsilon, \Delta, \mathsf{e}, \mathsf{m},\mathsf{m}_K)$ be a monoidal coalgebra modality on the same symmetric monoidal category $(\mathbb{X}, \otimes, K)$. Then the following are in bijective correspondence:
\begin{enumerate}
\item Symmetric monoidal mixed distributive laws of $(\mathsf{T}, \mu, \eta, \mathsf{n}, \mathsf{n}_{K})$ over $(\oc, \delta, \varepsilon, \mathsf{m},\mathsf{m}_K)$;
\item Liftings of the monoidal coalgebra modality $(\oc, \delta, \varepsilon, \Delta, \mathsf{e}, \mathsf{m},\mathsf{m}_K)$ to $\mathbb{X}^\mathsf{T}$, that is, a monoidal coalgebra modality $(\tilde{\oc}, \tilde{\delta_A}, \tilde{\varepsilon}, \tilde{\Delta}, \tilde{\mathsf{e}},  \tilde{\mathsf{m}}, \tilde{\mathsf{m}}_{(K, \mathsf{n}_K)})$ on $(\mathbb{X}^\mathsf{T}, \otimes^{\mathsf{n}}, (K,\mathsf{n}_K))$ which is a lifting of the underlying coalgebra modality $(\oc, \delta, \varepsilon, \Delta, \mathsf{e})$ to $(\mathbb{X}^\mathsf{T}, \otimes^{\mathsf{n}}, (K,\mathsf{n}_K))$ (in the sense of Proposition \ref{liftingcoalgmod}) and a lifting of the underlying symmetric monoidal comonad $(\oc, \delta, \varepsilon, \mathsf{m},\mathsf{m}_K)$ to $(\mathbb{X}^\mathsf{T}, \otimes^{\mathsf{n}}, (K,\mathsf{n}_K))$ (in the sense of Proposition \ref{liftsymmix});
\item Liftings of the symmetric comonoidal monad $(\mathsf{T}, \mu, \eta, \mathsf{n}, \mathsf{n}_{K})$ to $(\mathbb{X}^\oc, \otimes^\mathsf{m}, (K, \mathsf{m}_K))$ (in the sense of Proposition \ref{liftsymmix}). 
\end{enumerate}
\end{prop}
\begin{proof} That $(1) \Leftrightarrow (3)$ is precisely part of Proposition \ref{liftsymmix}. While that $(1) \Leftrightarrow (2)$ follows from Proposition \ref{liftsymmix}, Lemma \ref{biglemma}, and Proposition \ref{liftingcoalgmod}. The bijective correspondence also follows from Proposition \ref{liftsymmix} and Proposition \ref{liftingcoalgmod}. 
\end{proof} 

We give a name to symmetric comonoidal monads with symmetric monoidal mixed distributive laws over monoidal coalgebra modalities: 

\begin{defi}\label{expliftdef} Let $(\oc, \delta, \varepsilon, \Delta, \mathsf{e}, \mathsf{m},\mathsf{m}_K)$ be a monoidal coalgebra modality on a symmetric monoidal category $(\mathbb{X}, \otimes, K)$. An \textbf{exponential lifting monad} of $(\oc, \delta, \varepsilon, \Delta, \mathsf{e}, \mathsf{m},\mathsf{m}_K)$ is a sextuple \\
\noindent $(\mathsf{T}, \mu, \eta, \mathsf{n}, \mathsf{n}_{K}, \lambda)$ consisting of a symmetric comonoidal monad $(\mathsf{T}, \mu, \eta, \mathsf{n}, \mathsf{n}_{K})$ on $(\mathbb{X}, \otimes, K)$ and a symmetric monoidal mixed distributive law $\lambda$ of $(\mathsf{T}, \mu, \eta, \mathsf{n}, \mathsf{n}_{K})$ over $(\oc, \delta, \varepsilon, \Delta, \mathsf{e}, \mathsf{m},\mathsf{m}_K)$. 
\end{defi}

Proposition \ref{liftmoncoalg} implies that the Eilenberg-Moore category of an exponential lifting monad admits a monoidal coalgebra modality which is strictly preserved by the forgetful functor. We can also extend on Proposition \ref{symmonresult} for monoidal coalgebra modalities. 

\begin{prop} \label{bigresultprop} Let $(\oc, \delta, \varepsilon, \Delta, \mathsf{e}, \mathsf{m},\mathsf{m}_K)$ be a monoidal coalgebra modality on a symmetric monoidal category $(\mathbb{X}, \otimes, K)$. Then the following are in bijective correspondence:
\begin{enumerate}
\item Monoids in $(\mathbb{X}^\oc, \otimes^\mathsf{m}, (K, \mathsf{m}_K))$;
\item Cocommutative bimonoids $(A, \nabla, \mathsf{u}, \Delta, \mathsf{e})$ in $(\mathbb{X}, \otimes, K)$ equipped with a natural transformation $\lambda_X: A \otimes \oc(X) \to \oc(A \otimes X)$ such that $(A \otimes -, \mu^\nabla, \eta^\mathsf{u}, \mathsf{n}^\Delta, \mathsf{n}^\mathsf{e}_K, \lambda)$ is an exponential lifting monad (where the symmetric comonoidal monad is defined as in Example \ref{bimonadex}) of $(\oc, \delta, \varepsilon, \Delta, \mathsf{e}, \mathsf{m},\mathsf{m}_K)$ and $\lambda$ satisfies (\ref{nablastrong}). 
\end{enumerate}
Therefore, each monoid in the Eilenberg-Moore category of a monoidal coalgebra modality induces an exponential lifting monad. 
\end{prop} 
\begin{proof} At first glance, it may be somewhat surprising that one only requires a monoid to obtain an exponential lifting monad while a cocommutative bimonoid was needed in Proposition \ref{symmonresult} to construct a symmetric monoidal mixed distributive law. However since the Eilenberg-Moore category of a monoidal coalgebra modality is a Cartesian monoidal category it follows that every monoid is in fact a cocommutative bimonoid (Lemma \ref{Cartbim}). Therefore the bijective correspondence follows immediately from Proposition \ref{symmonresult}. 
\end{proof} 

\subsection{Lifting Linear Category Structure}

The last piece of the puzzle for lifting the linear category structure is being able to lift the symmetric monoidal closed structure strictly. This is easily achieved by required the exponential lifting monads to be symmetric Hopf monads.   

\begin{defi}\label{mellliftdef} A \textbf{$\mathsf{MELL}$ lifting monad} on a linear category is an exponential lifting monad (with respect to the linear category's monoidal coalgebra modality) whose underlying symmetric comonoidal monad is also a symmetric Hopf monad. 
\end{defi}

It is worth mentioning that in the definitions of an exponential lifting monad and of a $\mathsf{MELL}$ lifting monad, we do not require that the underlying endofunctors of these monads be linearly distributive functors between linear categories in the sense of \cite{hyland2003glueing, mellies2004comparing}.

$\mathsf{MELL}$ lifting monads provide us with following main result of this paper: 

\begin{thm}\label{mellthm} The Eilenberg-Moore category of a $\mathsf{MELL}$ lifting monad is a linear category such that the forgetful functor preserves the linear category structure strictly. 
\end{thm} 

Recalling that cocommutative Hopf monoids in a Cartesian monoidal category are known as groups (Definition \ref{groupdef}), then using Proposition \ref{symmonresult2} one can easily extend Proposition \ref{bigresultprop} to the context of linear categories. 

\begin{thm} \label{bigresultthm} Let $(\oc, \delta, \varepsilon, \Delta, \mathsf{e}, \mathsf{m},\mathsf{m}_K)$ be a monoidal coalgebra modality on a symmetric monoidal closed category $(\mathbb{X}, \otimes, K)$, that is, $(\mathbb{X}, \otimes, K)$ is a linear category. Then the following are in bijective correspondence:
\begin{enumerate}
\item Groups in $(\mathbb{X}^\oc, \otimes^\mathsf{m}, (K, \mathsf{m}_K))$;
\item Cocommutative Hopf monoids $(H, \nabla, \mathsf{u}, \Delta, \mathsf{e}, \mathsf{S})$ in $(\mathbb{X}, \otimes, K)$ equipped with a natural transformation $\lambda_X: H \otimes \oc(X) \to \oc(H \otimes X)$ such that $(H \otimes -, \mu^\nabla, \eta^\mathsf{u}, \mathsf{n}^\Delta, \mathsf{n}^\mathsf{e}_K, \lambda)$ is a $\mathsf{MELL}$ lifting monad (where the symmetric Hopf monad is defined as in Example \ref{hopfmonadex}), $\lambda$ satisfies (\ref{nablastrong}) and (\ref{extraS}). 
\end{enumerate}
Therefore, for each group $((H, \omega), \nabla, \mathsf{u}, \mathsf{S})$ in $(\mathbb{X}^\oc, \otimes^{\mathsf{m}}, (K,\mathsf{m}_K))$, $\mathsf{MOD}(H, \nabla, \mathsf{u})$ is a linear category. 
\end{thm} 

Finally recall that the Eilenberg-Moore category of a free exponential modality is equivalent to the category of cocommutative comonoids on the base category, it follows that groups in this Eilenberg-Moore category are in bijective correspondence to cocommutative Hopf monoids in the base category. Therefore we can conclude the following for Lafont categories: 

\begin{cor}\label{Lafontcor} If $(H, \nabla, \mathsf{u}, \Delta, \mathsf{e}, \mathsf{S})$ is a cocommutative Hopf monoid in a Lafont category, then $\mathsf{MOD}(H, \nabla, \mathsf{u})$ is a Lafont category. 
\end{cor}

\section{Mixed Distributive Laws for Differential Categories}\label{diffsec}

In this final section, we discuss lifting differential category structure and how additive structure induces a source of exponential lifting monads and $\mathsf{MELL}$ lifting monads. 

\subsection{Additive Symmetric Monoidal Categories} Here we mean ``additive'' in the Blute, Cockett, and Seely sense of the term \cite{blute2006differential}, that is, enriched over commutative monoids. In particular, we do not assume negatives (at least not yet... see Theorem \ref{hopfneg}) nor do we assume biproducts which differs from other definitions of an additive category found in the literature \cite{mac2013categories}. Also in this section, we will often use equational representation instead of commutative diagrams as it is easier to work directly with equations when addition is involved. We will use diagrammatic order for composition: this means that the composite maps $f: A \to B$ and $g: B \to C$ is denoted $f;g: A \to C$, which is the map which first does $f$ then $g$.

\begin{defi}\label{ASMC} An \textbf{additive category} \cite{blute2006differential} is a category $\mathbb{X}$ such that for each pair of objects $A$ and $B$, the hom-set $\mathbb{X}(A,B)$ is a commutative monoid (in the classical sense) with binary operation $+: \mathbb{X}(A,B) \times \mathbb{X}(A,B) \to \mathbb{X}(A,B)$ called \textbf{addition} and unit $0_{A,B}: A \to B$ called the \textbf{zero map}, and furthermore that composition preserves the additive structure in the sense that: 
\[k;(f+g);h=(k;f;h)+(k;g;h) \quad \quad \quad 0;f=0=f;0\] 
An \textbf{additive symmetric monoidal category} \cite{blute2006differential} is a symmetric monoidal category $(\mathbb{X}, \otimes, K)$ such that $\mathbb{X}$ is an additive category and the monoidal product $\otimes$ is compatible with the additive structure in the sense that:
\[(f+g)\otimes h= f\otimes h + g\otimes h \quad \quad \quad 0\otimes f=0\] 
\end{defi}

\begin{exas} \normalfont Here are some examples of additive symmetric monoidal categories. Furthermore in these examples, the additive structures are induced from biproducts (see \cite{mac2013categories} for how biproducts induced an additive category structure).
\begin{enumerate}
\item The category of sets $\mathsf{SET}$ \emph{cannot} be made into an additive category. If $\mathsf{SET}$ was an additive category then categorical product (which is the Cartesian product) would coincide with the categorical coproduct (which is the disjoint union) -- which is clearly not the case. Therefore $(\mathsf{SET}, \times, \lbrace \ast \rbrace)$ is not an additive symmetric monoidal category. 
\item On the other hand, $(\mathsf{REL}, \times, \lbrace \ast \rbrace)$ is an additive symmetric monoidal where the sum of relations $R: X \to Y$ and $S: X \to Y$, which recall are subsets $R,S \subseteq X \times Y$, is defined as their union $R + S := R \cup S$, while the zero map $0_{X,Y}: X \to Y$ is the empty relation $0_{X,Y} := \emptyset$. 
\item Let $\mathbb{K}$ be a field. Then $(\mathsf{VEC}_\mathbb{K}, \otimes, \mathbb{K})$ is an additive symmetric monoidal category where the sum of $\mathbb{K}$-linear maps $f: V \to W$ and $g: V \to W$ is the standard sum of linear maps $f+g$ defined pointwise, $(f+g)(v) = f(v) +g(v)$, and where the zero map $0_{V,W}: V \to W$ is the $\mathbb{K}$-linear map which maps every element of $V$ to the zero element of $W$. \end{enumerate}
\end{exas} 

It is worth mentioning that every additive category can be completed to a category with biproducts  \cite{mac2013categories}  (which is itself an additive category), and similarly every additive symmetric monoidal category can be completed to an additive symmetric monoidal category with distributive biproducts. For this reason, it is possible to argue \cite{fiore2007differential} that one should always assume a setting with biproducts. The problem is that arbitrary coalgebra modalities do not necessarily extend to the biproduct completion. However, monoidal coalgebra modalities induce monoidal coalgebra modalities on the biproduct completion (see \cite{Blute2019} for more details). 

\subsection{Exponential Lifting Monads from Additive Structure}\label{addsec} In this section, we turn to monoidal coalgebra modalities over additive symmetric monoidal categories to provide us with a source of monoids in the Eilenberg-Moore category, which by Proposition \ref{bigresultprop} also provides a source of exponential lifting monads. 

If $(\oc, \delta, \varepsilon, \Delta, \mathsf{e}, \mathsf{m},\mathsf{m}_K)$ is a monoidal coalgebra modality on an additive symmetric monoidal category $(\mathbb{X}, \otimes, K)$, then there are two natural transformations $\nabla_A: \oc(A) \otimes \oc(A) \to \oc(A)$ and $\mathsf{u}_A: K \to \oc(A)$ such that $(\oc(A), \nabla_A, \mathsf{u}_A, \Delta, \mathsf{e}_A)$ is a commtuative and cocommutative bimonoid \cite{Blute2019}. Explicitly $\nabla_A$ and $\mathsf{u}_A$ are defined respectively as:
\begin{equation}\label{nablaaddbialg}\begin{gathered}
\nabla_A:=  \xymatrixrowsep{0.5pc}\xymatrixcolsep{5pc}\xymatrix{\oc(A) \otimes \oc(A)  \ar[r]^-{\delta_A \otimes \delta_A} & \oc \oc(A) \otimes \oc \oc(A) \ar[r]^-{\mathsf{m}_{\oc(A),\oc(A)}} &\\
 \oc(\oc(A) \otimes \oc(A)) \ar[rr]^-{\oc\left( \left((\varepsilon_A \otimes \mathsf{e}_A);\rho_A \right) + \left((\mathsf{e}_A \otimes \varepsilon_A);\ell_A \right) \right)} && \oc(A)
  }\\
   \mathsf{u}_A:= \xymatrixcolsep{4pc}\xymatrix{K \ar[r]^-{\mathsf{m}_K} & \oc(K) \ar[r]^-{\oc(0_{K,A})} & \oc(A)
  }\end{gathered}\end{equation}
When post-composing $\nabla$ and $\mathsf{u}$ by $\varepsilon$, it follows immediately from the coherences of a symmetric monoidal comonad (\ref{symcomonad}) that the following diagrams commute \cite{Blute2019}:
  \begin{equation}\label{epsilonnabla}\begin{gathered}
 \xymatrixcolsep{4pc}\xymatrix{\oc(A) \otimes \oc(A) \ar[rr]^-{\left((\varepsilon_A \otimes \mathsf{e}_A);\rho_A \right) + \left((\mathsf{e}_A \otimes \varepsilon_A);\ell_A \right)} \ar[dr]_-{\nabla_A} && A & K \ar[dr]_-{\mathsf{u}_K}  \ar[rr]^-{0_{K,A}} && A  \\
  & \oc(A) \ar[ur]_-{\varepsilon_A} & & & \oc(A) \ar[ur]_-{\varepsilon_A}  }
  \end{gathered}\end{equation}
Furthermore $\nabla_A$ and $\mathsf{u}_A$ are both $\oc$-coalgebra morphisms \cite{Blute2019} that is, the following diagrams commute: 
\begin{equation}\label{nabla!map2} \xymatrixcolsep{3.5pc}\begin{gathered}\xymatrix{\oc(A) \otimes \oc(A) \ar[d]_-{\delta_A \otimes \delta_A} \ar[rr]^-{\nabla_A} & & \oc(A) \ar[d]^-{\delta_A} & K \ar[d]_-{\mathsf{m}_K} \ar[r]^-{\mathsf{u}_A} & \oc(A) \ar[d]^-{\delta_A} \\
    \oc \oc(A) \otimes \oc \oc(A) \ar[r]_-{\mathsf{m}_{\oc(A), \oc(A)}} & \oc(\oc(A) \otimes \oc(A)) \ar[r]_-{\oc(\nabla_A)} & \oc \oc(A) &   \oc(K) \ar[r]_-{\oc(\mathsf{u}_A)} & \oc\oc(A)
  } \end{gathered}\end{equation}
  Therefore, $((\oc(A), \delta_A), \nabla_A, \mathsf{u}_A)$ is a commutative monoid in the Eilenberg-Moore category $(\mathbb{X}^\oc, \otimes^\mathsf{m}, (K, \mathsf{m}_K))$. There is a also a relation between the monoidal coalgebra modality and the additive structure given by bialgebra convolution \cite{majid2000foundations}, that is, the following diagrams commute: 
\begin{equation}\label{addbialgdef}\begin{gathered}
\xymatrixcolsep{4pc}\xymatrix{\oc(A) \ar[d]_-{\Delta_A}  \ar[r]^-{\oc(f+g)} & \oc(B) \ar[d]^-{\Delta_B} & \oc(A) \ar[dr]_-{\mathsf{e}_A} \ar[rr]^-{\oc(0_{A,B})} && \oc(B) \\
\oc(A) \otimes \oc(A) \ar[r]_-{\oc(f) \otimes \oc(g)} & \oc(B) \otimes \oc(B) && K \ar[ur]_-{\mathsf{u}_B}
  }\end{gathered}\end{equation}
In fact, this natural (co)commutative bimonoid structure provides an equivalent definition of a monoidal coalgebra modality on an additive symmetric monoidal category. This equivalent definition is known as an \textbf{additive bialgebra modality} \cite{Blute2019}. We conclude this section with the observation that for a monoidal coalgebra modality on an additive symmetric monoidal category, every object induces an exponential lifting monad. 
  
\begin{prop}\label{!bimonoid} Let $(\oc, \delta, \varepsilon, \Delta, \mathsf{e}, \mathsf{m},\mathsf{m}_K)$ be a monoidal coalgebra on an additive symmetric monoidal category $(\mathbb{X}, \otimes, K)$. Then for every object $A$, $(\oc(A) \otimes -, \mu^{\nabla_A}, \eta^{\mathsf{u}_A}, \mathsf{n}^{\Delta_A}, \mathsf{n}^{\mathsf{e}_K}, \delta_A^\flat)$ is an exponential lifting monad of $(\oc, \delta, \varepsilon, \Delta, \mathsf{e}, \mathsf{m},\mathsf{m}_K)$ where the symmetric comonoidal monads is defined as in Example \ref{bimonadex} and the symmetric monoidal mixed distributive law is defined as in (\ref{alphaomega}). 
\end{prop}   
\begin{proof} This follows immediately from Proposition \ref{bigresultprop}. 
\end{proof} 
  
\subsection{$\mathsf{MELL}$ Lifting Monads from Negatives}\label{addsec2} In this section we add negatives to the story to obtain a source of groups in the Eilenberg-Moore category of a monoidal coalgebra modality and therefore also a source of $\mathsf{MELL}$ lifting monads. In particular, we will show that the cofree coalgebras of monoidal coalgebra modalities on additive symmetric monoidal categories are Hopf monoids precisely when the additive symmetric monoidal category also admits negatives. Here we mean ``admits negatives'' in the sense of being enriched over abelian groups. 

\begin{defi} \normalfont An additive category $\mathbb{X}$ is said to have \textbf{negatives} if for each map admits an additive inverse, that is, for every map $f: A \to B$ there is a map $-f: A \to B$ such that $f+(-f)= 0_{A,B}$. 
\end{defi}

We first note that for an additive symmetric monoidal category to have negatives: it is sufficient for the identity map of the monoidal unit to have an additive inverse. 

\begin{lem}\label{neg} An additive symmetric monoidal category $(X, \otimes, K)$ admits negatives if and only if the identity of the monoidal unit $1_K: K \to K$ admits an additive inverse. 
\end{lem}
\begin{proof} One direction is automatic by definition. For the converse, suppose that $1_K$ admits an additive inverse $-1_K: K \to K$ such that $1_K + (-1_K) = 0$. For a map $f: A \to B$, define the map $-f: A \to B$ as follows: 
  \[ -f := \xymatrixcolsep{5pc}\xymatrix{ A \ar[r]^-{\ell^{-1}_A} & K \otimes A \ar[r]^-{(-1_K) \otimes f} & K \otimes B \ar[r]^-{\ell_B} & B 
  } \]
Then by the additive symmetric monoidal structure, it follows that $f + (-f)=0$.  
  \begin{align*}
f+(-f) &=~ (\ell^{-1}_A; \ell_A; f) + \left( \ell^{-1}_A; ((-1_K) \otimes f) ; \ell_B \right) \\
&=~ (\ell^{-1}_A; (1_K \otimes f) ; \ell_B) + \left( \ell^{-1}_A; ((-1_K) \otimes f) ; \ell_B \right) \tag{Nat. of $\ell$} \\
&=~ \ell^{-1}_A; \left((1_K + (-1_K) \otimes f \right); \ell_B \tag{Definition \ref{ASMC}} \\
&=~ \ell^{-1}_A; \left(0_{K,K} \otimes f \right); \ell_B \\
&=~ \ell^{-1}_A; 0_{K \otimes A, K \otimes B}; \ell_B \tag{Definition \ref{ASMC}} \\
&=~ 0_{A,B} \tag{Definition \ref{ASMC}}
\end{align*}
\end{proof} 

Furthermore, in an additive symmetric monoidal category which has negatives, the tensor product is compatible with negatives in the sense that $-(f \otimes g)=(-f) \otimes g = f \otimes (-g)$. 

\begin{prop}\label{hopfneg} Let  $(\oc, \delta, \varepsilon, \Delta, \mathsf{e}, \mathsf{m},\mathsf{m}_K)$ be a monoidal coalgebra modality on an additive symmetric monoidal category $(\mathbb{X}, \otimes, K)$. Then there exists a natural transformation $\mathsf{S}_A: \oc(A) \to \oc(A)$ such that for each object $A$, the septuple $(\oc(A), \nabla_A, \mathsf{u}_A, \Delta_A, \mathsf{e}_A, \mathsf{S}_A)$ is a cocommutative Hopf monoid (where $\nabla_A$ and $\mathsf{u}_A$ are defined as in (\ref{nablaaddbialg})) if and only if the additive symmetric monoidal category $(\mathbb{X}, \otimes, K)$ is enriched over abelian groups. 
\end{prop} 
\begin{proof} $\Leftarrow$: Suppose that $(\mathbb{X}, \otimes, K)$ is enriched over abelian groups. Define the natural transformation $\mathsf{S}_A: \oc(A)\to \oc(A)$ as $\mathsf{S}_A := \oc(-1_A)$. Since the bimonoid is (co)commutative $(\oc(A), \nabla_A, \mathsf{u}_A, \Delta_A, \mathsf{e}_A)$, we only need to check one half of (\ref{hopf}) to show that \\ \noindent $(\oc(A), \nabla_A, \mathsf{u}_A, \Delta_A, \mathsf{e}_A, \mathsf{S}_A)$ is a Hopf monoid. Then that $\mathsf{S}$ is an antipode for \\ \noindent  $(\oc(A), \nabla_A, \mathsf{u}_A, \Delta_A, \mathsf{e}_A)$ follows from commutativity of the following diagram: 
\[ \xymatrix{& \oc(A) \otimes \oc(A) \ar[rr]^-{1_{\oc(A)} \otimes \oc(-1_A)} & &  \oc(A) \otimes \oc(A) \ar[dr]^-{\nabla_A}  \\
\oc(A) \ar[ur]^-{\Delta_A}  \ar[drr]_-{\mathsf{e}_A} \ar[rrrr]^-{\oc(0_{A,A})} &   &  & & \oc(A) \\
 && K \ar@{}[u]|-{(\ref{addbialgdef})} \ar@{}[uu]|(0.8){(\ref{addbialgdef})} \ar[urr]_-{\mathsf{u}_A}
  } \]
Furthermore note that by definition, $\mathsf{S}_A: (\oc(A), \delta_A) \to (\oc(A), \delta_A)$ is a $\oc$-coalgebra morphism. 

\noindent $\Rightarrow$: Suppose there exists a natural transformation $\mathsf{S}_A: \oc(A)\to \oc(A)$ such that for each object $A$, the septuple $(\oc(A), \nabla_A, \mathsf{u}_A, \Delta_A, \mathsf{e}_A, \mathsf{S}_A)$ is a cocommutative Hopf monoid. By Lemma \ref{neg} it suffices to show that $1_K: K \to K$ has an additive inverse. Define $-1_K: K \to K$ as follows: 
  \[  \xymatrixcolsep{5pc}\xymatrix{K \ar[r]^-{\mathsf{m}_K} & \oc(K)  \ar[r]^-{\mathsf{S}_K} & \oc(K) \ar[r]^-{\varepsilon_K} & K 
  } \]
That $1_K + (-1_K) = 0_{K,K}$ follows from the following equalities: 
\begin{align*}
0_{K,K} &=~\mathsf{u}_K; \varepsilon_K \tag{\ref{epsilonnabla}} \\
&=~ \mathsf{m}_K; \mathsf{e}_K; \mathsf{u}_K; \varepsilon_K \tag{\ref{Deltamonoidal}} \\
&=~  \mathsf{m}_K; \Delta_K; (\mathsf{S}_K \otimes 1_{\oc(K)}); \nabla_K; \varepsilon_K \tag{\ref{hopf}} \\
&=~ \mathsf{m}_K; \Delta_K; (\mathsf{S}_K \otimes 1_{\oc(K)}); \left[\left((\varepsilon_K \otimes \mathsf{e}_K);\rho_K \right) + \left((\mathsf{e}_K \otimes \varepsilon_K);\ell_K \right) \right] \tag{\ref{epsilonnabla}} \\
&=~ \left[\mathsf{m}_K; \Delta_K; (\mathsf{S}_K \otimes 1_{\oc(K)}); (\varepsilon_K \otimes \mathsf{e}_K); \rho_K \right] + \left[ \mathsf{m}_K; \Delta_K; (\mathsf{S}_K \otimes 1_{\oc(K)}); (\mathsf{e}_K \otimes \varepsilon_K);\ell_K \right] \tag{Definition \ref{ASMC}} \\
&=~ \left[\mathsf{m}_K; \Delta_K; (\mathsf{S}_K \otimes 1_{\oc(K)}); (\varepsilon_K \otimes \mathsf{e}_K); \rho_K \right] + \left[ \mathsf{m}_K; \Delta_K; (\mathsf{e}_K \otimes \varepsilon_K);\ell_K \right] \tag{Lemma \ref{lemmahopf} (ii)} \\ 
&=~ \left(\mathsf{m}_K; \mathsf{S}_K; \varepsilon_K \right) + \left( \mathsf{m}_K; \varepsilon_K \right) \tag{\ref{comonoid}} \\ 
&=~ (-1_K) + 1_K \\
&=~ 1_K + (-1_K)
                      \tag*{\qedhere}
\end{align*} 
\end{proof}

Therefore we obtain the following: 

\begin{thm}\label{!hopf}  Let $(\mathbb{X}, \otimes, K)$ be a linear category with monoidal coalgebra modality \\ \noindent $(\oc, \delta, \varepsilon, \Delta, \mathsf{e}, \mathsf{m},\mathsf{m}_K)$ such that $(\mathbb{X}, \otimes, K)$ is also an additive symmetric monoidal category with negatives. Then for every object $A$, $(\oc(A) \otimes -, \mu^{\nabla_A}, \eta^{\mathsf{u}_A}, \mathsf{n}^{\Delta_A}, \mathsf{n}^{\mathsf{e}_K}, \delta_A^\flat)$ is a $\mathsf{MELL}$ lifting monad and therefore $\left( \mathsf{MOD}(\oc(A), \nabla_A, \mathsf{u}_A), \otimes^{\mathsf{n}^{\Delta_A}}, (K, \mathsf{n}^{\mathsf{e}_A}_K) \right)$ is a linear category. 
\end{thm} 
\begin{proof} This theorem follows directly from Proposition \ref{hopfneg} and Theorem \ref{bigresultthm}.
\end{proof} 

\subsection{Lifting Differential Category Structure} 

We briefly give the definition of a differential category. For more details on differential categories see \cite{blute2006differential, blute2015cartesian, Blute2019}.

\begin{defi} A differential category \cite{blute2006differential} is an additive symmetric monoidal category $(\mathbb{X}, \otimes, K)$ with a coalgebra $(\oc, \delta, \varepsilon, \Delta, \mathsf{e})$ equipped with a deriving transformation, that is, a natural transformation $\mathsf{d}_A: \oc(A) \otimes A \to \oc(A)$ satisfying the identities found in \cite[Definition 2.5]{blute2006differential}. 
\end{defi} 

In order to be able to lift differential category structure, we will need that the Eilenberg-Moore category of our symmetric comonoidal monad is an additive symmetric monoidal category. In order to achieve this, we will need that the underlying endofunctor preserves the additive structure. 

\begin{defi} An \textbf{additive functor} between additive categories is a functor which preserves the additive structure strictly, that is, a functor $\mathsf{T}$ such that $\mathsf{T}(f+g)=\mathsf{T}(f) + \mathsf{T}(g)$ and $\mathsf{T}(0)=0$. 
\end{defi}

One can easily check that for a monad on an additive category whose underlying endofunctor is additive, that its Eilenberg-Moore category is also an additive category such that the forgetful functor preserves the additive structure strictly. Similarly, for a symmetric comonoidal monad on an additive symmetric monoidal category whose underlying endofunctor is additive, its Eilenberg-Moore category is also an additive category such that the forgetful functor preserves the additive symmetric monoidal structure strictly. Luckily, for any additive symmetric monoidal category our favourite endofunctor $A \otimes -$ is additive for any object $A$. 

\begin{defi}\label{difflift} Let $(\mathbb{X}, \otimes, K)$ be a differential category with coalgebra modality\\\noindent $(\oc, \delta, \varepsilon, \Delta, \mathsf{e})$ equipped with deriving transformation $\mathsf{d}$, and let $(\mathsf{T}, \mu, \eta, \mathsf{n}, \mathsf{n}_{K})$ be a symmetric comonoidal monad on $(\mathbb{X}, \otimes, K)$ whose underlying endofunctor $\mathsf{T}$ is additive. A \textbf{differential mixed distributive law of $(\mathsf{T}, \mu, \eta, \mathsf{n}, \mathsf{n}_{K})$ over $(\oc, \delta, \varepsilon, \Delta, \mathsf{e}, \mathsf{d})$} is a coalgebra mixed distributive law $\lambda$ of $(\mathsf{T}, \mu, \eta, \mathsf{n}, \mathsf{n}_{K})$ over $(\oc, \delta, \varepsilon, \Delta, \mathsf{e})$ such that the following diagram commutes: 
\begin{equation}\label{distderive} \begin{gathered}\xymatrixcolsep{5pc}\xymatrix{\mathsf{T}(\oc(A) \otimes A) \ar[d]_-{\mathsf{T}(\mathsf{d}_A)} \ar[r]^-{\mathsf{n}_{\oc(A),A}} & \mathsf{T} \oc(A) \otimes \mathsf{T}(A) \ar[r]^-{\lambda_A \otimes 1_{\mathsf{T}(A)}}& \oc \mathsf{T}(A) \otimes \mathsf{T}(A) \ar[d]^-{\mathsf{d}_{\mathsf{T}(A)}} \\
\mathsf{T} \oc(A) \ar[rr]_-{\lambda_A} & & \oc \mathsf{T}(A)
  } \end{gathered}\end{equation}
\end{defi}

\begin{prop}\label{liftingdiff} Let $(\mathbb{X}, \otimes, K)$ be a differential category with coalgebra modality \\ \noindent $(\oc, \delta, \varepsilon, \Delta, \mathsf{e})$ equipped with deriving transformation $\mathsf{d}$, and let $(\mathsf{T}, \mu, \eta, \mathsf{n}, \mathsf{n}_{K})$ be a symmetric comonoidal monad on $(\mathbb{X}, \otimes, K)$ whose underlying endofunctor is additive. Then the following are in bijective correspondence:
\begin{enumerate}
\item Differential mixed distributive laws $(\mathsf{T}, \mu, \eta, \mathsf{n}, \mathsf{n}_{K})$ over $(\oc, \delta, \varepsilon, \Delta, \mathsf{e}, \mathsf{d})$;
\item Liftings of $\mathsf{d}$ to $(\mathbb{X}^\mathsf{T}, \otimes^{\mathsf{n}}, (K,\mathsf{n}_K))$, that is, a deriving transformation $\tilde{\mathsf{d}}$ for the induced lifted coalgebra modality $(\tilde{\oc}, \tilde{\delta}, \tilde{\varepsilon}, \tilde{\Delta}, \tilde{\mathsf{e}})$ on $(\mathbb{X}^\mathsf{T}, \otimes^{\mathsf{n}}, (K,\mathsf{n}_K))$ from Proposition \ref{liftingcoalgmod} such that for every $\oc$-colagbera $(A, \omega)$, $\tilde{\mathsf{d}}_{(A, \omega)} = \mathsf{d}_A$. 
\end{enumerate}
Therefore if $\lambda$ is a differential mixed distributive law of $(\mathsf{T}, \mu, \eta, \mathsf{n}, \mathsf{n}_{K})$ over $(\oc, \delta, \varepsilon, \Delta, \mathsf{e}, \mathsf{d})$, then $(\mathbb{X}^\mathsf{T}, \otimes^{\mathsf{n}}, (K,\mathsf{n}_K))$ is a differential category with coalgebra modality $(\tilde{\oc}, \tilde{\delta}, \tilde{\varepsilon}, \tilde{\Delta}, \tilde{\mathsf{e}})$ and deriving transformation $\tilde{\mathsf{d}}$.  
\end{prop}
\begin{proof} The bijective correspondence will follow immediately from Proposition \ref{liftingcoalgmod}. Therefore, it remains to show that we can obtain one from the other. 

\noindent $(1) \Rightarrow (2)$: Let $\lambda$ be a differential mixed distributive law of $(\mathsf{T}, \mu, \eta, \mathsf{n}, \mathsf{n}_{K})$ over $(\oc, \delta, \varepsilon, \Delta, \mathsf{e}, \mathsf{d})$. Consider the induced lifting of $(\oc, \delta, \varepsilon, \Delta, \mathsf{e})$ from Proposition \ref{liftingcoalgmod}. To prove that we have a lifting of the deriving transformation it suffices to show that $\mathsf{d}$ is a $\mathsf{T}$-algebra morphism. Then if $(A,\nu)$ is a $\mathsf{T}$-algebra, commutativity of the following diagram shows that $\mathsf{d}_A$ is a $\mathsf{T}$-algebra morphism:
 \[\xymatrixcolsep{5pc}\xymatrix{\mathsf{T}(\oc(A) \otimes A) \ar@{}[drr]|-{(\ref{distderive})} \ar[d]_-{\mathsf{T}(\mathsf{d}_A)} \ar[r]^-{\mathsf{n}_{\oc(A),A}} & \mathsf{T} \oc(A) \otimes \mathsf{T}(A) \ar[r]^-{\lambda_A \otimes 1_{\mathsf{T}(A)}}& \oc \mathsf{T}(A) \otimes \mathsf{T}(A) \ar[d]^-{\mathsf{d}_{\mathsf{T}(A)}} \ar[r]^-{\oc(\nu) \otimes \nu} \ar@{}[dr]|-{\text{Nat. of. } \mathsf{d}_A} & \oc(A) \otimes A \ar[d]^-{\mathsf{d}} \\
\mathsf{T} \oc(A) \ar[rr]_-{\lambda_A} & & \oc \mathsf{T}(A) \ar[r]_-{\oc(\nu)} & \oc(A)
  }\]
$(2) \Rightarrow (1)$: Let $\tilde{\mathsf{d}}$ be a lifting of $\mathsf{d}$ to $(\mathbb{X}^\mathsf{T}, \otimes^{\mathsf{n}}, (K,\mathsf{n}_K))$. This implies that $\mathsf{d}$ is a $\mathsf{T}$-algebra morphism which in particular for free $\mathsf{T}$-algebras $(\mathsf{T}(A), \mu_A)$ the following diagram commutes:
\begin{equation}\label{distderive2} \begin{gathered}  \xymatrixcolsep{5pc}\xymatrix{\mathsf{T}(\oc \mathsf{T}(A) \otimes \mathsf{T}(A)) \ar[d]_-{\mathsf{T}(\mathsf{d})} \ar[r]^-{\mathsf{n}_{\oc(A),A}} & \mathsf{T}\oc \mathsf{T}(A) \otimes \mathsf{T} \mathsf{T}(A) \ar[r]^-{\mu_A^\sharp \otimes \mu_A} & \oc \mathsf{T}(A) \otimes \mathsf{T}(A) \ar[d]^-{\mathsf{d}} \\
\mathsf{T} \oc \mathsf{T}(A) \ar[rr]_-{\mu_A^\sharp} & & \oc \mathsf{T}(A)
  } \end{gathered}\end{equation}
Consider now the induced coalgebra mixed distributive law $\lambda$ of $(\mathsf{T}, \mu, \eta, \mathsf{n}, \mathsf{n}_{K})$ over \\ \noindent $(\oc, \delta, \varepsilon, \Delta, \mathsf{e})$ from Proposition \ref{liftingcoalgmod}. Then that $\lambda$ satisfies the extra necessary condition follows from commutativity of the following diagram: 
\[ \xymatrixcolsep{3.5pc}\xymatrix{\mathsf{T}(\oc(A) \otimes A) \ar[dd]_-{\mathsf{T}(\mathsf{d}_A)} \ar[dr]_-{\mathsf{T}(\oc(\eta_A) \otimes \eta_A)~~} \ar[r]^-{\mathsf{n}_{\oc(A),A}} & \mathsf{T}\oc(A) \otimes \mathsf{T}(A)  \ar@{}[d]|-{\text{Nat. of. } \mathsf{n}} \ar[r]^-{\mathsf{T}\oc(\eta_A) \otimes 1_{\mathsf{T}(A)}}& \mathsf{T} \oc \mathsf{T}(A) \otimes  \mathsf{T}(A) \ar@{}[dr]|(0.35){(\ref{monadeq})}  \ar[d]_-{1_{\mathsf{T}\oc\mathsf{T}(A)} \otimes \mathsf{T}(\eta_A)}  \ar[r]^-{\mu_A^\sharp \otimes 1_{\mathsf{T}(A)}} & \oc \mathsf{T}(A) \otimes \mathsf{T}(A)  \ar[dd]^-{\mathsf{d}_{\mathsf{T}(A)}} \\
 & \mathsf{T}(\oc \mathsf{T}(A) \otimes \mathsf{T}(A)) \ar@{}[drr]|-{(\ref{distderive2})} \ar[d]_-{\mathsf{T}(\mathsf{d}_{\mathsf{T}(A)})} \ar[r]_-{\mathsf{n}_{A,B}}  & \mathsf{T}\oc \mathsf{T}(A) \otimes \mathsf{T} \mathsf{T}(A) \ar[ur]_-{\mu_A^\sharp \otimes \mu_A} &  \\
\mathsf{T}(\oc(A)) \ar@{}[ur]|-{\text{Nat. of. } \mathsf{d}} \ar[r]_-{\mathsf{T}\oc(\eta_A)}  & \mathsf{T} \oc \mathsf{T}(A) \ar[rr]_-{\mu_A^\sharp}  & & \oc \mathsf{T}(A)
  }\]
\end{proof} 

In a differential category whose coalgebra modality is also a monoidal coalgebra modality, the deriving transformation $\mathsf{d}$ and the monoidal coalgebra modality $(\oc, \delta, \varepsilon, \Delta, \mathsf{e}, \mathsf{m}_{A,B}.\mathsf{m}_K)$ are compatible in the sense that the following diagram commutes \cite[Theorem 25]{Blute2019}: 
\begin{equation}\label{monoidalrule} \begin{gathered}  \xymatrixcolsep{4pc} \xymatrix{ \oc(A) \otimes \left(\oc(B) \otimes B \right) \ar[rr]^-{\Delta_A \otimes (1_{\oc(B)} \otimes 1_B)} \ar[dddd]_-{1_{\oc(A)} \otimes \mathsf{d}_B} && (\oc(A) \otimes A) \otimes (\oc(B) \otimes B) \ar[d]^-{\tau_{\oc(A), A, \oc(B), B}} \\
&& \left(\oc(A) \otimes \oc(B) \right) \otimes \left(\oc(A) \otimes B\right) \ar[d]^-{(1_{\oc(A)} \otimes 1_{\oc(B)}) \otimes (\varepsilon_A \otimes 1_B)} \\
  & & \left(\oc(A) \otimes \oc(B) \right) \otimes (A \otimes B) \ar[d]^-{\mathsf{m}_{A,B} \otimes (1_A \otimes 1_B)} \\
  && \oc (A \otimes B) \otimes (A \otimes B) \ar[d]^-{\mathsf{d}_{A \otimes B}} \\
 \oc(A) \otimes \oc(B) \ar[rr]_-{\mathsf{m}_{A,B}} & & \oc (A \otimes B)
  }  \end{gathered}\end{equation}
  where $\mathsf{n}^{\Delta_A}_{\oc(B) \otimes B}$ is defined as in (\ref{bimonadn}). This coherence allows one to easily lift differential category structure to categories of modules. 

\begin{prop} Let $(\mathbb{X}, \otimes, K)$ be a differential category with a monoidal coalgebra modality $(\oc, \delta, \varepsilon, \Delta, \mathsf{e}, \mathsf{m}, \mathsf{m}_K)$ equipped with deriving transformation $\mathsf{d}$. Then for every monoid $((A, \omega), \nabla, \mathsf{u})$ in $(\mathbb{X}^\mathsf{T}, \otimes^{\mathsf{n}}, (K,\mathsf{n}_K))$, the induced symmetric monoidal mixed distributive law from Proposition \ref{bigresultprop} satisfies (\ref{distderive}) and therefore the category of modules over $(A, \nabla, \mathsf{u})$ is a differential category. \end{prop}
\begin{proof} Let $((A, \omega), \nabla, \mathsf{u})$ be a monoid in $(\mathbb{X}^\mathsf{T}, \otimes^{\mathsf{n}}, (K,\mathsf{n}_K))$ and consider the induced symmetric monoidal mixed distributive law satisfying (\ref{nablastrong}) $\omega^\natural$ as defined in (\ref{alphaomega}). Then the desired result follows from commutativity of the following diagram: 
{\scriptsize \[\xymatrixcolsep{1.05pc}\xymatrixrowsep{5pc} \xymatrix{A \otimes (\oc(X) \otimes X) \ar[ddrrrr]|-{\omega \otimes (1_{\oc(X)} \otimes 1_X)} \ar[rrrrr]^-{1_A \otimes \mathsf{d}_X} \ar[d]|-{\Delta^\omega \otimes (1_{\oc(X)} \otimes 1_X)} &&&&& A \otimes \oc(X) \ar[dd]|-{\omega \otimes 1_{\oc(X)}} \\
    (A \otimes A) \otimes (\oc(X) \otimes X)  \ar@{}[rr]|-{(\ref{omegacomon})}  \ar[d]|-{\tau_{A,A,\oc(X),X}} \ar[ddrrrr]|-{(\omega \otimes \omega) \otimes (1_{\oc(X)} \otimes 1_X)} &&  &&&  \\
     (A \otimes \oc(X)) \otimes (A \otimes X)  \ar[d]|-{(\omega \otimes 1_{\oc(X)}) \otimes (1_A \otimes 1_X)} && && \oc(A) \otimes (\oc(X) \otimes X) \ar[d]^-{\Delta_A \otimes (1_{\oc(X)} \otimes 1_X)}  \ar[r]^-{1_{\oc(A)} \otimes \mathsf{d}_X}& \oc(A) \otimes \oc(X) \ar[dd]|-{\mathsf{m}_{A,X}}\\
     ( \oc(A) \otimes \oc(X) ) \otimes (A \otimes X) \ar@{}[urr]|-{(\ref{!coalg})} \ar[d]|-{\mathsf{m}_{A,X} \otimes (1_A \otimes 1_X)} &&\ar[ll]_-{(1_{\oc(A)} \otimes 1_{\oc(X)}) \otimes (\varepsilon_A \otimes 1_X)} (\oc(A) \otimes \oc(X) \otimes \oc(A) \otimes X)&& \ar[ll]^-{\tau_{\oc(A),\oc(A),\oc(X),X}} (\oc(A) \otimes \oc(A)) \otimes (\oc(X) \otimes X) \\
      \oc(A \otimes X) \otimes (A \otimes X)  \ar@{}[urrrr]|-{(\ref{monoidalrule})} \ar[rrrrr]_-{\mathsf{d}_{A \otimes X}} &&&&& \oc(A \otimes X)}\] }%
      \end{proof}  

\section{Conclusion}

The main goal of this paper was to study for which monoids in a linear category was the category of modules of said monoid again a linear category. We showed that the answer to this question was groups in the Eilenberg-Moore category of the monoidal coalgebra modality of a linear category, that is, cocommutative Hopf monoids in the linear category which come equipped with a special natural transformation (Theorem \ref{bigresultthm}). Therefore categories of modules over these cocommutative Hopf monoids are linear categories and therefore a model of $\mathsf{MELL}$. Thus the work of this paper extends Blute and Scott's work \cite{blute2004category, blute1996hopf} by adding to the study of categories of modules over Hopf monoids as models of linear logic. This work should also hopefully pave the way for studying other applications of Hopf monoids in linear logic. For example, the bimonoid version of a coalgebra modality is known as a bialgebra modality, which is used in the study of differential categories \cite{blute2006differential, Blute2019}. As such one could study the notion of a  ``Hopf algebra modality'', a Hopf monoid version of a coalgebra modality, and look at what the antipode adds to the story. In particular, Proposition \ref{hopfneg} implies that a monoidal coalgebra modality over an additive symmetric monoidal category with negatives should be an example of a ``Hopf algebra modality''.  

This paper also introduces the notion of a symmetric monoidal mixed distributive law (Definition \ref{mixedmondef}) and showed these are in bijective correspondence with liftings of symmetric monoidal comonads to the Eilenberg-Moore category of symmetric comonoidal monads and vice-versa (Proposition \ref{liftsymmix}). While in this paper we focused mainly on symmetric comonoidal monads of the form $A \otimes -$, it would be interesting to find and study more exotic examples of symmetric monoidal mixed distributive laws (whether they are related to $\mathsf{MELL}$ or not). One other possible direction would be to study symmetric monoidal mixed distributive laws from a higher category theory perspective. 


\bibliographystyle{alpha}      
\bibliography{LiftingPaper}   

\begin{thebibliography}{LMMP13}

\bibitem[BCLS19]{Blute2019}
R.~F. Blute, J.~R.~B. Cockett, J.-S.~P. Lemay, and R.~A.~G. Seely.
\newblock Differential categories revisited.
\newblock {\em Applied Categorical Structures}, Jul 2019.

\bibitem[BCS06]{blute2006differential}
R.F. Blute, J.R.B Cockett, and R.A.G Seely.
\newblock Differential categories.
\newblock {\em Mathematical structures in computer science}, 16(6):1049--1083,
  2006.

\bibitem[BCS15]{blute2015cartesian}
R.F. Blute, J.R.B Cockett, and R.A.G Seely.
\newblock Cartesian differential storage categories.
\newblock {\em Theory and Applications of Categories}, 30(18):620--686, 2015.

\bibitem[Bec69]{beck1969distributive}
J.~Beck.
\newblock Distributive laws.
\newblock In {\em Seminar on triples and categorical homology theory}, pages
  119--140. Springer, 1969.

\bibitem[Bie95]{bierman1995categorical}
G.M. Bierman.
\newblock What is a categorical model of intuitionistic linear logic?
\newblock In {\em International Conference on Typed Lambda Calculi and
  Applications}, pages 78--93. Springer, 1995.

\bibitem[Blu96]{blute1996hopf}
R.F. Blute.
\newblock Hopf algebras and linear logic.
\newblock {\em Mathematical Structures in Computer Science}, 6(2):189--212,
  1996.

\bibitem[BLV11]{Bruguieres2011hopf}
A.~Bruguieres, S.~Lack, and A.~Virelizier.
\newblock Hopf monads on monoidal categories.
\newblock {\em Advances in Mathematics}, 227(2):745--800, 2011.

\bibitem[Bra14]{brandenburg2014tensor}
M.~Brandenburg.
\newblock {\em Tensor Categorical Foundations of Algebraic Geometry}.
\newblock Universit{\"a}ts- und Landesbibliothek M{\"u}nster, 2014.

\bibitem[BS04]{blute2004category}
R.F. Blute and P.~Scott.
\newblock Category theory for linear logicians.
\newblock {\em Linear logic in computer science}, 316:3--65, 2004.

\bibitem[BV07]{Bruguieres2007hopf}
A.~Bruguieres and A.~Virelizier.
\newblock Hopf monads.
\newblock {\em Advances in Mathematics}, 215(2):679--733, 2007.

\bibitem[BW90]{barr1990category}
M.~Barr and C.~Wells.
\newblock Category theory for computer scientists.
\newblock {\em preparation. Google Scholar}, 1990.

\bibitem[CEPT17]{crubille2017free}
R.~Crubill{\'e}, T.~Ehrhard, M.~Pagani, and C.~Tasson.
\newblock The free exponential modality of probabilistic coherence spaces.
\newblock In {\em International Conference on Foundations of Software Science
  and Computation Structures}, pages 20--35. Springer, 2017.

\bibitem[Fio07]{fiore2007differential}
M.P. Fiore.
\newblock Differential structure in models of multiplicative biadditive
  intuitionistic linear logic.
\newblock In {\em International Conference on Typed Lambda Calculi and
  Applications}, pages 163--177. Springer, 2007.

\bibitem[Gir87]{girard1987linear}
J-Y. Girard.
\newblock Linear logic.
\newblock {\em Theoretical computer science}, 50(1):1--101, 1987.

\bibitem[GL87]{girard1987linear2}
J-Y. Girard and Y.~Lafont.
\newblock Linear logic and lazy computation.
\newblock In {\em International Joint Conference on Theory and Practice of
  Software Development}, pages 52--66. Springer, 1987.

\bibitem[Has10]{hasegawa2010bialgebras}
M.~Hasegawa.
\newblock Bialgebras in {R}el.
\newblock {\em Electronic Notes in Theoretical Computer Science}, 265:337--350,
  2010.

\bibitem[HHM07]{harmer2007categorical}
R.~Harmer, M.~Hyland, and P-A. Mellies.
\newblock Categorical combinatorics for innocent strategies.
\newblock In {\em 22nd annual IEEE symposium on logic in computer science},
  pages pp--379, 2007.

\bibitem[HS03]{hyland2003glueing}
M.~Hyland and A.~Schalk.
\newblock Glueing and orthogonality for models of linear logic.
\newblock {\em Theoretical computer science}, 294(1-2):183--231, 2003.

\bibitem[JJ17]{jenvcova2017monoids}
A.~Jen{\v{c}}ov{\'a} and G.~Jen{\v{c}}a.
\newblock On monoids in the category of sets and relations.
\newblock {\em International Journal of Theoretical Physics},
  56(12):3757--3769, 2017.

\bibitem[Kel64]{kelly1964maclane}
G.M. Kelly.
\newblock On maclane's conditions for coherence of natural associativities,
  commutativities, etc.
\newblock {\em Journal of Algebra}, 1(4):397--402, 1964.

\bibitem[Laf88]{lafont1988logiques}
Y.~Lafont.
\newblock {\em Logiques, cat{\'e}gories et machines}.
\newblock PhD thesis, PhD thesis, Universit{\'e} Paris 7, 1988.

\bibitem[Lem18]{lemay2018lifting}
J-S.~P. Lemay.
\newblock Lifting coalgebra modalities and imell model structure to
  eilenberg-moore categories.
\newblock In {\em LIPIcs-Leibniz International Proceedings in Informatics},
  volume 108. Schloss Dagstuhl-Leibniz-Zentrum fuer Informatik, 2018.

\bibitem[LMMP13]{laird2013weighted}
J.~Laird, G.~Manzonetto, G.~McCusker, and M.~Pagani.
\newblock Weighted relational models of typed lambda-calculi.
\newblock In {\em Proceedings of the 2013 28th Annual ACM/IEEE Symposium on
  Logic in Computer Science}, pages 301--310. IEEE Computer Society, 2013.

\bibitem[Maj00]{majid2000foundations}
S.~Majid.
\newblock {\em Foundations of quantum group theory}.
\newblock Cambridge university press, 2000.

\bibitem[McC02]{mccrudden2002opmonoidal}
P.~McCrudden.
\newblock Opmonoidal monads.
\newblock {\em Theory Appl. Categ}, 10(19):469--485, 2002.

\bibitem[Mel03]{mellies2003categorical}
P-A. Mellies.
\newblock Categorical models of linear logic revisited.
\newblock 2003.

\bibitem[Mel04]{mellies2004comparing}
P-A. Mellies.
\newblock Comparing hierarchies of types in models of linear logic.
\newblock {\em Information and Computation}, 189(2):202--234, 2004.

\bibitem[Mel09]{mellies2009categorical}
P-A. Mellies.
\newblock Categorical semantics of linear logic.
\newblock {\em Panoramas et syntheses}, 27:15--215, 2009.

\bibitem[ML13]{mac2013categories}
S.~Mac~Lane.
\newblock {\em Categories for the working mathematician}.
\newblock Springer-Verlag, New York, Berlin, Heidelberg, 1971, revised 2013.

\bibitem[Moe02]{moerdijk2002monads}
I.~Moerdijk.
\newblock Monads on tensor categories.
\newblock {\em Journal of Pure and Applied Algebra}, 168(2):189--208, 2002.

\bibitem[MTT17]{mellies2017explicit}
P-A. Melli{\`e}s, N.~Tabareau, and C.~Tasson.
\newblock An explicit formula for the free exponential modality of linear
  logic.
\newblock {\em Mathematical Structures in Computer Science}, pages 1--34, 2017.

\bibitem[Mur15]{murfet2015sweedler}
D.~Murfet.
\newblock On sweedler's cofree cocommutative coalgebra.
\newblock {\em Journal of Pure and Applied Algebra}, 219(12):5289--5304, 2015.

\bibitem[Por08]{porst2008categories}
H-E. Porst.
\newblock On categories of monoids, comonoids, and bimonoids.
\newblock {\em Quaestiones Mathematicae}, 31(2):127--139, 2008.

\bibitem[Sch04]{schalk2004categorical}
A.~Schalk.
\newblock What is a categorical model of linear logic?
\newblock {\em Manuscript, available from
  http://www.cs.man.ac.uk/{~}schalk/notes/llmodel.pdf}, 2004.

\bibitem[Sel10]{selinger2010survey}
P.~Selinger.
\newblock A survey of graphical languages for monoidal categories.
\newblock In {\em New structures for physics}, pages 289--355. Springer, 2010.

\bibitem[Sla17]{slavnov2015banach}
S.~Slavnov.
\newblock On banach spaces of sequences and free linear logic exponential
  modality.
\newblock {\em Mathematical Structures in Computer Science}, pages 1--28, 2017.

\bibitem[Str72]{STREET1972149}
R.~Street.
\newblock The formal theory of monads.
\newblock {\em Journal of Pure and Applied Algebra}, 2(2):149 -- 168, 1972.

\bibitem[Swe69]{sweedler1969hopf}
M.E. Sweedler.
\newblock Hopf algebras. mathematical lecture note series, 1969.

\bibitem[VO71]{van1971sheaves}
D.H. Van~Osdol.
\newblock Sheaves in regular categories.
\newblock In {\em Exact Categories and Categories of Sheaves}, pages 223--239.
  Springer, 1971.

\bibitem[Wis75]{wischnewsky1975linear}
M.~Wischnewsky.
\newblock On linear representations of affine groups. i.
\newblock {\em Pacific Journal of Mathematics}, 61(2):551--572, 1975.

\bibitem[Wis08]{wisbauer2008algebras}
R.~Wisbauer.
\newblock Algebras versus coalgebras.
\newblock {\em Applied Categorical Structures}, 16(1):255--295, 2008.

\end{thebibliography}

\end{document}